\def\cmntsoff{} 
\documentclass[aps,pra,superscriptaddress,12pt,nofootinbib,longbibliography]{revtex4-2}

\makeatletter
\renewcommand{\p@subsection}{}
\renewcommand{\p@subsubsection}{}
\makeatother

\makeatletter
\renewcommand{\subsection}{%
    \@startsection{subsection}{2}{\z@}
    {-3.25ex \@plus -1ex \@minus -.2ex}
    {-1em}
    {\normalfont\normalsize\bfseries}
}
\makeatother

\usepackage{graphicx} 
\usepackage{dcolumn}  
\usepackage{stmaryrd} 
\usepackage{dsfont} 
\usepackage{bm}        
\usepackage{amsmath,amssymb,amsfonts,amsthm}
\usepackage{mathtools}
\usepackage[normalem]{ulem} 

\newtheorem{theorem}{Theorem}[section]
\newtheorem*{theorem*}{Theorem}
\newtheorem{proposition}[theorem]{Proposition}
\newtheorem{proposition*}{Proposition}
\newtheorem{lemma}[theorem]{Lemma}
\newtheorem*{lemma*}{Lemma}

\usepackage[mathscr]{eucal}
\usepackage{color}
\usepackage{xcolor}
\usepackage{tikz}
\usetikzlibrary{quotes, arrows.meta, calc, shapes.geometric, positioning}
\usepackage{quantikz}
\tikzset{
  meter box/.style={
    draw,
    rounded corners=8pt,
    minimum height=1.3cm,
    minimum width=1.6cm,
    align=center,
    fill=white,
    inner sep=3pt
  }
}

\usepackage[utf8]{inputenc}

\usepackage{xparse}
\usepackage{hyperref}
\usepackage{cleveref}

\usepackage{listings}
\usepackage{color} 
\definecolor{mygreen}{RGB}{28,172,0} 
\definecolor{mylilas}{RGB}{170,55,241}

\usepackage{xfp} 
\usepackage{xcolor}

\usepackage[normalem]{ulem} 

\usepackage{physics} 
\NewCommandCopy{\svqty}{\qty}
\usepackage{siunitx} 
\RenewCommandCopy{\qty}{\svqty}

\providecommand{\ignore}[1]{}

\newif\ifcmnt
\cmnttrue
\ifdefined\cmntsoff\cmntfalse\fi

\ifcmnt
    \providecommand{\aucmnt}[1]{#1}
    \providecommand{\MKcolor}{\color{teal}}

    \providecommand{\Pcolor}{\color{gray}}
\else
    \providecommand{\aucmnt}[1]{}
    \providecommand{\MKcolor}{}

    \providecommand{\Pcolor}{}
\fi

\newcommand{\MKc}[1]{\aucmnt{{\MKcolor [MK: {\color{darkgray}#1}]}}}

\newcommand{\Pc}[1]{\aucmnt{{\Pcolor \textbf{[}Permanent comment: {\color{gray}#1}\textbf{]}}}}

\newcommand{\mySI}[1]{\;\textrm{#1}} 
\newcommand{\cB}{\mathcal{B}}
\newcommand{\cC}{\mathcal{C}}
\newcommand{\cD}{\mathcal{D}}

\newcommand{\cH}{\mathcal{H}}

\newcommand{\cO}{\mathcal{O}}
\newcommand{\cY}{\mathcal{Y}}

\newcommand{\one}{\mathds{1}}

\newcommand{\rls}{\mathds{R}}
\newcommand{\cmplxs}{\mathds{C}}

\newcommand{\tot}{{\rm tot}}

\newcommand{\pmatr}[1]{\begin{pmatrix}#1\end{pmatrix}}

\begin{document} 
\title{Quantified convergence of general homodyne measurements with applications to continuous variable quantum computing}
\author{Emanuel Knill}
\affiliation{National Institute of Standards and Technology, Boulder, Colorado 80305, USA}
\affiliation{Center for Theory of Quantum Matter, University of Colorado, Boulder, Colorado 80309, USA}
\author{Ezad Shojaee}
\affiliation{National Institute of Standards and Technology, Boulder, Colorado 80305, USA}
\affiliation{Department of Physics, University of Colorado, Boulder, Colorado, 80309, USA}
\affiliation{IonQ, College Park, MD 20740}
\author{James R. van Meter}
\affiliation{National Institute of Standards and Technology, Boulder, Colorado 80305, USA}
\author{Akira Kyle}
\affiliation{Department of Physics, University of Colorado, Boulder, Colorado, 80309, USA}
\author{Scott Glancy}
\affiliation{National Institute of Standards and Technology, Boulder, Colorado 80305, USA}
\date{\today}

\begin{abstract}
In Ref.~\cite{bbp1}\ignore{\Pc{arXiv:2503.00188}} we introduced \emph{broadband pulsed (BBP) homodyne measurements} as a generalization of standard pulsed homodyne quadrature measurements.  BBP can take advantage of detectors such as calorimeters that have the potential for high efficiency over a broad spectral range.  BBP homodyne retains the advantages of standard pulsed homodyne, enabling measurement of arbitrary quadratures in the limit of large amplitude local oscillators (LO).   Here we quantify the convergence of standard and BBP homodyne quadrature measurements to those of the quadrature of interest. We obtain lower bounds on the fidelity of the post-measurement classical-quantum state of outcomes and unmeasured modes, and the fidelity of the states obtained after applying operations conditional on measurement outcomes.  The bounds depend on the LO amplitude and the moments of number operators. We demonstrate the practical relevance of these bounds by evaluating them for standard pulsed homodyne used for estimating values of the characteristic function of the Wigner distribution, expectations of moments, for quantum teleportation and for continuous variable error correction with GKP codes.
\end{abstract}

\maketitle
\tableofcontents

\section{Introduction}

Homodyne measurements of the complex amplitude of quantum light are
indispensable in quantum optics.  Homodyne measurements have a
fundamentally simple setup.  This setup consists of no more than a
beamsplitter combining the signal with a local oscillator (LO), and
two detectors after the beamsplitter.  The two detector outputs are
subtracted and rescaled to obtain the measurement outcomes. The
detectors are often photodiodes that record intensity as a function of
time. When the LO is always on, the time-resolved record of the subtracted
intensities gives information about many amplitudes in
parallel. These amplitudes are quadratures of modes of
light. Alternatively, one can use a pulsed LO, and measure the total
number of photons in each arm without time resolution to measure a
specific quadrature determined by the reference pulse.

Our work on general homodyne quadrature measurements~\cite{bbp1} was
motivated by the problem of measuring extremely broadband,
octave-spanning modes of light. These are modes of light whose photons
are created into a wavepacket that has a multi-octave spectral
width. Such photons are present in femto-second optical frequency
combs used for metrology~\cite{cundiff2003colloquium} and occur in the
inertial-frame description of the Unruh effect associated with
accelerating detectors in relativistic quantum
optics~\cite{PRL2011}. For octave-spanning modes of light, there are
no efficient, sufficiently fast detectors for time-resolved homodyne
measurements. At the cost of a significantly more involved
experimental setup, one can separate the frequency band into its
components and measure each component separately, or one can use
broadband pre-amplification followed by a classical measurement of the
spectrum~\cite{Shaked}.  Alternatively, one can simply use a short
pulsed LO and measure each arm as before but with efficient broadband
detectors that can be slow. One class of such detectors that may
achieve the necessary broadband efficiency consists of
calorimeters~\cite{cabrera1998detection,hattori2022optical}.  Because
calorimeters measure energy instead of  photon number, the pulsed homodyne
measurement scheme needs to be generalized. For this purpose, we
introduced broadband pulsed (BBP) homodyne in Ref.~\cite{bbp1}. We
showed BBP homodyne retains all the basic properties of standard
pulsed homodyne. In particular, as the amplitude of the LO grows
large, the moments of the measurement outcomes converge to those of
the target quadrature, and the probability distribution of the
outcomes converges weakly to that of the ideal quadrature measurement.

In order to achieve a desired measurement precision in homodyne
measurements it is necessary to use a sufficiently high amplitude
LO. The fundamental results of homodyne measurement theory guarantee
convergence of moments and weak convergence of probability
distributions~\cite{kiukas2008moment_jmo,bbp1}.  Standard arguments
for convergence of homodyne to quadrature measurements are based on
one-mode semiclassical approximations such as
in~\cite{leonhardt1995measuring}, or on convergence heuristics of
expressions for the homodyne observables at finite amplitude, for
examples see
Refs.~\cite{braunstein1990homodyne,Banaszek,combes2022homodyne}.
However, it is necessary to quantify this convergence in order to
determine a minimum amplitude sufficient to achieve the desired
precision. Such quantifications are available for signal states in
specific families of states such as Gaussian states, and observables
such as low-degree quadrature moments. The standard advice is to
ensure that the average number of photons in the LO pulse is much
larger than the average number of photons of the signal. This advice
needs modification depending on the goal of the measurement. For
example, see Refs.~\cite{braunstein1990homodyne,Banaszek,Sanders}.
Besides convergence of observables, modern applications of homodyne
measurements require convergence of the quantum state of the
unmeasured mode after these modes have been modified conditional on
measurement outcomes. Such modifications are essential in continuous
variable (CV) quantum
teleportation~\cite{braunstein1998teleportation,furusawa1998teleportation}
and CV quantum computation~\cite{gottesman2001encoding}.  Both of
these situations involve non-Gaussian quantum states, for which
general, quantified bounds on convergence are lacking.

Here we develop and implement strategies for obtaining general bounds
on convergence for BBP homodyne measurements of quadratures. We use
state distances and fidelities to directly compare the states
resulting from homodyne measurement at finite LO amplitude with the
states that would be obtained with a direct but realistic quadrature
measurement, where we model realistic quadrature measurements as an
ideal measurement with realistic noise that can be arbitrarily
small. In order to compare the states, we take advantage of a
representation of the measurements in terms of von Neumann measurement
models~\cite{vonneumann1955foundations}.  These models involve
prepared ancilla modes of a quantum apparatus that are unitarily
coupled to the modes being measured. The ancilla modes may then be
measured with a fixed von Neumann measurement. Because fidelity
between states can only increase under identical data
processing~\cite{hou2012fidelity}, it suffices to lower bound the
fidelity between the states before the final measurement.

We develop fidelity bounds as a function of LO amplitude for two broad
scenarios. The first is for measurement, where we require high
fidelity for the joint outcomes and outcome-conditional states
remaining in unmeasured systems.  The bounds for this scenario can be
used to quantify the difference between the expectations of quadrature
functions such as moments and the corresponding functions of BBP
measurement outcomes. The second scenario involves applying quantum
operations conditional on measurement outcomes.  This scenario applies
to CV quantum teleportation and computation.

To demonstrate the relevance of our convergence bounds, we apply them
to standard pulsed homodyne.  We calculate fidelity bounds as a
function of LO amplitude for the classical-quantum state involving
measurement outcomes and unmeasured systems, and for measuring values
of the characteristic function of the Wigner distribution and
expectations of quadrature moments. The bounds for the
classical-quantum state and values of the characteristic function are
practically useful and are expressed in terms of expectations of
number operators. The bound on the difference of expectations of
moments depends as expected on high-order moments. The dependence on
the LO amplitude is not optimal, but the bound can still serve as a
guide.  For examples involving conditional actions, we apply our
bounds to CV quantum teleportation and CV quantum error
correction. These bounds depend on expectations of low-degree products
of quadratures and number operators. The bounds for CV quantum error
correction can be used to choose the LO amplitude to ensure small
added error without unnecessarily large LO powers.  We show that the
required LO amplitude is comparable to amplitudes used in existing
implementations of pulsed homodyne measurements with a pair of matched
photodiodes~\cite{hansen2001ultrasensitive,gerrits2011balanced}.

\section{Overview}

We establish notation and review the principles and formalism of BBP
homodyne in Sect.~\ref{sec:preliminaries}. We then introduce the
measurement models used to represent direct quadrature measurements
and BBP measurements of the quadrature (Sect.~\ref{sec:functana}). The
models are based on the von Neuman measurement models involving
couplings between the measured system and a quantum apparatus.  We use
these models to mathematically represent any actual measurement with
classical outcomes as a quantum process followed by a positive
operator-valued measure (POVM) on the apparatus with classical output.
When comparing direct and BBP measurements of a quadrature, we use the
same POVM on the apparatus.  Because of monotonicity of fidelity under
quantum processing~\cite{hou2012fidelity}, the fidelity between the
states obtained by the quantum process serves as a lower bound on the
final fidelity of the joint state of the unmeasured quantum systems
and the classical outcomes. This fidelity serves as a measure of
joint closeness of the classical outcome distributions and of the
states of the unmeasured quantum systems.

Our results quantifying the fidelity of BBP measurement outcomes
and remaining quantum states are in Sect.~\ref{sec:sconv}. The main
result, Thm.~\ref{thm:mainbnd} provides a general bound for the
distance between the state obtained by coupling the quantum apparatus
to the target quadrature, and the state obtained by coupling the
apparatus to the effective BBP observable defined in
Eq.~\eqref{eq:calop}. Thm.~\ref{thm:qmeas} applies the general
bound to the specific case where the apparatus is initialized with a
pure-state wavefunction according to the von Neumann measurement
model.  We also determine bounds on the distance between states
obtained by applying functions of the target quadrature or the BBP
observable, which can be interpreted as a quantified strong
convergence for certain bounded continuous functions. For other
functions we establish a regularization procedure in
Prop.~\ref{prop:regfexp} that makes it possible to replace the
function with a better-behaved one.

We analyze fidelities for outcome-conditional actions such as those
required for quantum teleportation in
Sect.~\ref{sec:condoperations}. We first consider conditional
unitaries acting on the unmeasured systems and defined by operators
\(\hat{U}(x)\) for each real \(x\). They can either be applied
conditional on the target quadrature's values or conditional on the
BBP observable's values. In experimental settings, measurements
such as those defined in Sect.~\ref{sec:sconv} are performed, 
and \(\hat{U}(x)\) is applied if the outcome is \(x\). We split the fidelity bounds into two
steps. In the first step, we compare the conditional unitaries applied
without measurement directly to the system of the target quadrature
and the unmeasured systems. Prop.~\ref{prop:conddisp} establishes
fidelity bounds for the case where the unitaries \(\hat{U}(x)\) are
commuting displacements. Applying this proposition may require
regularization of the operators \(\hat{U}(x)\) by applying
Lem.~\ref{lem:condreg}.  For the second step,
Prop.~\ref{prop:conditionalU} determines the added loss of fidelity
associated with first measuring, then applying the appropriate unitary
conditional on the outcome. This proposition is a general prescription
that is implemented for conditional displacements in
Prop.~\ref{prop:conditionalD}.

To illustrate the use of these bounds, we apply them to standard
pulsed homodyne in Sect.~\ref{sec:onemode}. We calculate the bound on
the classical-quantum fidelity of realistic homodyne measurements of
quadratures. The bounds depend linearly on the total photon number
expectation in the modes being measured, as expected from the
heuristic guidelines for setting LO amplitudes. The fidelity
approaches \(1\) quadratically in the inverse LO amplitude.  The
bounds also depend on the type and amount of added noise of the
measurement implementation. The less the added noise, the worse is the
bound. This effect is expected for fidelity measures that include a
comparison of the complete probability distributions. It can be
attributed to the impossibility of performing ideal quadrature
measurements.  For measures based on differences between expectations
of observables there is no such effect, as we demonstrate by applying
the bounds to measuring the values of the characteristic function of
the Wigner distribution.  The fidelities again depend linearly on
total photon number expectation.  The bounds worsen when evaluating
the characteristic function at points distant from the origin,
requiring high LO amplitude for a given fidelity.  Although our bounds
are estimated conservatively, the constants are moderate and the
bounds can be applied in practice.  As another example, we determine
conservative bounds on the difference between the expectation of the
moments of the target quadrature and the moments of the homodyne
observable. In this case, the bounds depend on high-order moments of
quadratures.  To illustrate the use of our distance bounds for
quadrature-conditional unitaries, we apply the bounds to quantum
teleportation.  For the final example, we estimate the added error
from homodyne quadrature measurements in CV quantum error
correction. For this example, we consider correction of displacement
noise for finite-energy Gottesman-Kitaev-Preskill (GKP) codes. By
treating the Gaussian added noise from photodiodes as added
displacement noise, it is possible to apply the results for the
fidelity of states and measurement outcomes. To illustrate the
relevance of the bounds, we evaluate them for the case where the
measurement is implemented with a matched pair of photo-diodes, where
we use the added noise reported in
Refs.~\cite{hansen2001ultrasensitive,gerrits2011balanced} for specific
values. We find that our error bounds are acceptably small for these
values of added noise and commonly used LO amplitudes.

For a summary of functional analysis conventions and background
required for the results of this paper, see App.~\ref{app:pvm}. We
assume without mention that functions that we introduce are
measurable, domains of projection-valued measures are standard Borel
spaces, and Hilbert spaces are separable.  The appendix also contains
frequently used identities and inequalities
(App.~\ref{app:inequalities}), including relationships between
fidelity, overlaps and distances.

\section{BBP homodyne}
\label{sec:preliminaries}

Standard pulsed and BBP homodyne measurements aim to measure a
quadrature of one mode out of many, where the mode measured is
determined by the LO pulse shape.  Thus we consider finitely many
modes defined by mode annihilation operators \(\hat{a}_{k}\) for
\(k=1,\ldots, N\). The annihilation operators are assumed to define
orthogonal modes satisfying
\([\hat{a}_{k}^{\dagger},\hat{a}_{l}] = \delta_{k,l}\one\), where
\(\one\) is the identity operator on the modes' Hilbert space.  In
general, we consider multiple families of modes consisting of modes
\(a_{k}, b_{l},c_{m},\ldots\) with associated mode operators
\(\hat{a}_{k},\hat{b}_{l}, \hat{c}_{m},\ldots\). For simplicity, the
mode indices all run from \(1\) to \(N\), where we fill in extra
orthogonal modes as needed to match the numbers of modes. We write
\(\bm{a} = (a_{k})_{k=1}^{N}\) and
\(\hat{\bm{a}} = (\hat{a}_{k})_{k=1}^{N}\), and similarly for other
families of modes, where we think of vectors such as \(\hat{\bm{a}}\)
as column vectors with operator entries.

The Hilbert space associated with a family of orthogonal modes
\(\bm{a}\) is the Fock space, which factors as a tensor product of
single-mode Hilbert spaces.  The Hilbert space of one mode is that of
a quantum harmonic oscillator and spanned by the number states
\(\ket{0},\ket{1},\ldots\). We label kets, bras and operators with mode labels as
needed. For example, the state \(\otimes_{k=1}^{N} \ket{n_{k}}_{a_{k}}\) is
the state with exactly \(n_{k}\) photons in mode \(a_{k}\) for each \(k\).

For identical length vectors \(\bm{x}\) and \(\bm{y}\), the inner
product \(\bm{x}\cdot\bm{y} \) is evaluated as \(\sum_{i}x_{i}y_{i}\),
where the product is the scalar product or operator product according
to the type of the vectors.  For complex vectors \(\bm{x}\),
\(\bm{x}^{*}\) is the entry-wise complex conjugate. The length of
\(\bm{x}\) is \(|\bm{x}|=\sqrt{\bm{x}^{*}\cdot\bm{x}}\).  The Hadamard
product \(\bm{x}*\bm{y}\) is the vector obtained by multiplying each
coordinate separately, that is \((\bm{x}*\bm{y})_{k} = x_{k}y_{k}\).
The complex zero vector is denoted by \(\bm{0}\).  For vectors
\(\bm{x}\) whose entries are operators, \(\bm{x}^{\dagger}\) is a
vector of the same shape whose entries are the adjoints of entries of
\(\bm{x}\).

If \(\bm{\alpha}\) is a vector of complex numbers satisfying
\(\bm{\alpha}^{*}\cdot\bm{\alpha} = 1\), then for the family of modes
\(\bm{a}\), the operator
\(\hat{a}_{\bm{\alpha}} = \bm{\alpha}\cdot\hat{\bm{a}}\) satisfies the
properties of a mode annihilation operator. We refer to the mode defined
by \(\hat{a}_{\bm{\alpha}}\) as the mode with shape \(\bm{\alpha}\). The
quadratures of this mode are of the form
\begin{align}
 \hat{x}_{\theta} &= -i \qty(e^{i\theta}\hat{a}_{\bm{\alpha}}^{\dagger}- e^{-i\theta}\hat{a}_{\bm{\alpha}}).
\end{align}
Because \(\hat{x}_{\theta}=-i\qty(\qty(\hat{a}_{e^{-i\phi}\bm{\alpha}})^{\dagger}-\hat{a}_{e^{-i\phi}\bm{\alpha}})\),
instead of specifying a quadrature in terms of a mode and a phase,
we directly specify quadratures according to
\begin{align}
  \hat{q}_{\bm{\alpha}} &=
                          -i\qty(\bm{\alpha}\cdot\hat{\bm{a}}^{\dagger}
                          - \bm{\alpha}^{*}\cdot\hat{\bm{a}}),
                          \label{eq:quaddef}
\end{align}
where we do not require \(\bm{\alpha}\) to be normalized. The familiy
of modes is implicit in the notation for \(\hat{q}_{\bm{\alpha}}\).
The commutation relationships are
\([\hat{q}_{\bm{\alpha}},\hat{q}_{\bm{\beta}}] =
\qty(\bm{\alpha}^{*}\cdot \bm{\beta}
-\bm{\beta}^{*}\cdot\bm{\alpha})\).  For this and other operators
associated with a family of modes, we always make it clear which
family of modes they are associated with.  The quadrature
\(\hat{q}_{\bm{\alpha}}\) is normalized if \(|\bm{\alpha}|^{2}=1\).
This differs from the normalization required for canonical commutation
relationships in natural units with \(\hbar=1\). The quadratures
satisfy canonical commutation relationships if
\(|\bm{\alpha}|^{2}=1/2\). For example,
see~\cite{leonhardt1995measuring}, Eq.~(13).  Thus, the canonical
quadrature with respect to vacuum proportional to
\(\hat{q}_{\bm{\alpha}}\) is
\(\hat{q}_{\bm{\alpha}/(\sqrt{2}|\bm{\alpha}|)} =
\frac{1}{\sqrt{2}}\hat{q}_{\bm{\alpha}/|\bm{\alpha}|}\).

For complex vectors \(\bm{\alpha}\), coherent states
\(\ket{\bm{\alpha}}\) of a family of modes \(\bm{a}\) satisfy
\(\hat{a}_{k}\ket{\bm{\alpha}} = \alpha_{k} \ket{\bm{\alpha}}\) for all
\(k\).  The vacuum state of the modes, that is, the state not
containing any photons, is the coherent state \(\ket{\bm{0}}\). We fix
the phases of coherent states relative to the vacuum state by
requiring that \(\braket{\bm{\alpha}}{\bm{0}}\) is real and
positive. For one mode, this is equivalent to requiring that
\(\ket{\beta}\) have amplitude \(e^{-|\beta|^{2}/2}\) at \(\ket{0}\)
when expressed in the number basis.  Coherent states are product
states with respect to every orthogonal mode basis. In particular
\(\ket{\bm{\alpha}} = \otimes_{k=1}^{N}\ket{\alpha_{k}}_{a_{k}}\).

Every unnormalized quadrature \(\hat{q}_{\bm{\alpha}}\) is associated
with a displacement operator
\begin{align}
  \hat{D}_{\bm{\alpha}} &= e^{i\hat{q}_{\bm{\alpha}}}.
\end{align}
This displacement operator can be thought of as being generated by the
Hamiltonian given by the normalized quadrature
\(\hat{q}_{\bm{\alpha}/|\bm{\alpha}|}\) where the evolution time is
\(|\bm{\alpha}|\).   The
displacement operators for one mode satisfy
\begin{align}
  \hat{D}_{\beta}\ket{\alpha}&=e^{i\Im{\alpha^{*}\beta}}\ket{\alpha+\beta},\notag\\
  \hat{D}_{-\beta} \hat{a} \hat{D}_{\beta} &= \hat{a} + \beta,\notag\\
  \hat{D}_{-\beta}\hat{q}_{\alpha}\hat{D}_{\beta} &= \hat{q}_{\alpha}+i\alpha^{*}\beta -i\alpha\beta^{*},
  \notag\\
  \hat{D}_{-\alpha-\beta} \hat{D}_{\beta}\hat{D}_{\alpha}&=e^{i\Im{\alpha^{*}\beta}}.
                                                          \label{eq:dispeqs}
\end{align}

The BBP homodyne configuration is shown in
Fig.~\ref{fig:bbp_homodyne}.  The only difference from the standard
pulsed homodyne configuration are  detectors that measure total energy
or another weighted combination of mode number operators instead of
total photon number.  The measurement operator for modes \(\bm{f}\) is
\(\hat{E}_{\bm{f}} =
\sum_{k}\omega_{k}\hat{f}_{k}^{\dagger}\hat{f}_{k}\), where
\(\omega_{k}\) are energies or other weights for mode \(k\).  The
state to be measured arrives in the signal modes \(\bm{a}\).  The LO
is prepared in the LO modes \(\bm{b}\). We find it mathematically
convenient to make the LO preparation explicit by initializing the LO
modes in vacuum, then applying the appropriate displacement
\(\hat{D}_{R\bm{\beta}}\), where \(R\) is an adjustable positive real
scale factor. In practice, the LO pulse is derived from a laser and
modified by pulse shaping if necessary.  Modes \(\bm{a}\) and
\(\bm{b}\) interfere on a balanced beamsplitter.  The outgoing modes
\(\bm{c}\) and \(\bm{d}\) are measured with the detectors, the
measurement outcome of modes \(\bm{d}\) is subtracted from that of
modes \(\bm{c}\), and the result of the subtraction is multiplied by
\(1/R\).

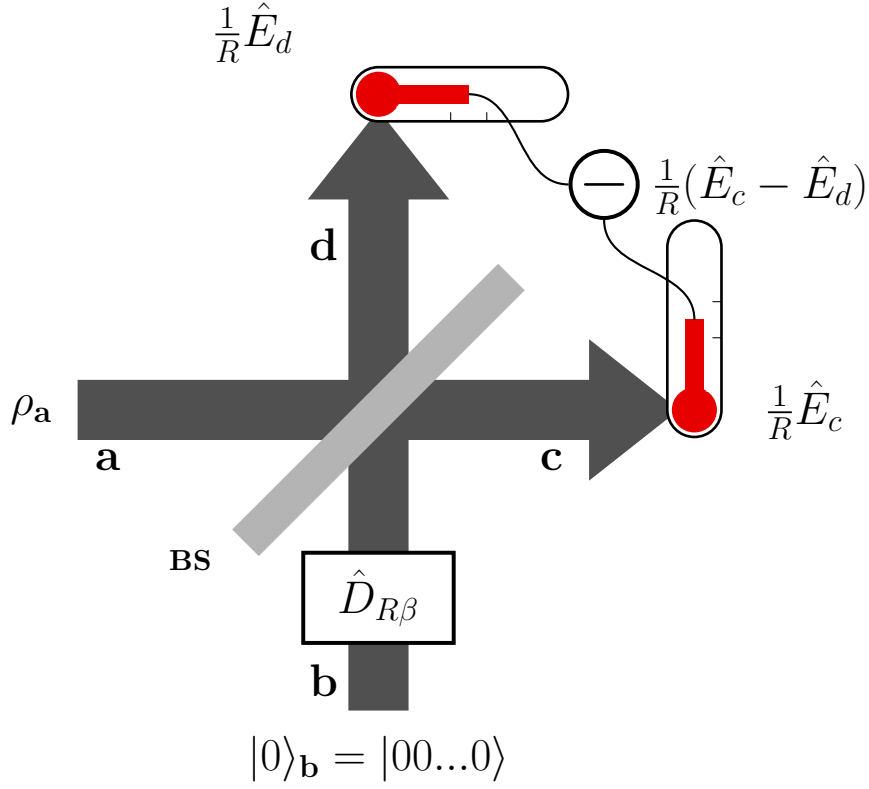
\begin{figure}

\definecolor{modegray}{RGB}{80,80,80}    
\definecolor{bsgray}{RGB}{180,180,180}   
\definecolor{thermored}{RGB}{230,0,0}    

\begin{tikzpicture}[
    mode arrow/.style={
        -{Triangle[width=1.8cm, length=1.2cm]}, 
        line width=0.8cm,                       
        color=modegray
    },
    box/.style={
        draw=black,
        line width=1.5pt,
        fill=white,
        minimum width=2cm,
        minimum height=1.2cm,
        align=center
    },
    subtraction/.style={
        circle,
        draw=black,
        line width=1.5pt,
        fill=white,
        minimum size=0.8cm,
        inner sep=0pt
    }
]

    \coordinate (center) at (0,0);

    
    \draw[mode arrow] (-4,0) -- (4,0);
    
    \draw[mode arrow] (0,-4) -- (0,4);

    
    \node[box] at (0,-2.5) {\Large $\hat{D}_{R\beta}$}; 

    \begin{scope}[rotate=-45]
        \fill[bsgray] (-0.25,-2.5) rectangle (0.25,2.5);
    \end{scope}
    \node at (-2.5, -2) {\textbf{BS}}; 

    
    \node[anchor=east] at (-4.2, 0) {\Large $\rho_{\mathbf{a}}$}; 
    \node[anchor=north] at (-3.6, -0.4) {\Large ${\mathbf{a}}$};     

    \node[anchor=east] at (-0.4, -3.6) {\Large ${\mathbf{b}}$};    
    \node[anchor=north] at (0, -4.2) {\Large $|0\rangle_{\mathbf{b}} = |00...0\rangle$}; 

    \node[anchor=north] at (2.3, -0.4) {\Large ${\mathbf{c}}$};      
    \node[anchor=east] at (-0.4, 2.2) {\Large ${\mathbf{d}}$};       

    
    \tikzset{
        thermometer/.pic={
            \draw[line width=1pt, fill=white] (-0.3,0) arc(180:360:0.3) -- (0.3,1.8) arc(0:180:0.3) -- cycle;
            \fill[thermored] (0,0) circle(0.25); 
            \fill[thermored] (-0.1,0) rectangle (0.1, 1.0); 
            \draw[white, line width=1.5pt, opacity=0.6] (0.15, 0.2) -- (0.15, 1.6);
            \draw[black, thin] (0.2, 0.8) -- (0.3, 0.8);
            \draw[black, thin] (0.2, 1.2) -- (0.3, 1.2);
        }
    }

    \path (0,4.2) pic[rotate=-90, scale=1.2] {thermometer};
    \coordinate (therm_d_connect) at (1.2, 4.2); 

    \path (4.2,0) pic[scale=1.2] {thermometer};
    \coordinate (therm_c_connect) at (4.2, 1.2); 

    
    \node[subtraction] (sub) at (3, 3) {\Huge $-$};
    
    \draw[thick, black] (therm_d_connect) to[out=0, in=180] (sub.west);
    
    \draw[thick, black] (therm_c_connect) to[out=90, in=270] (sub.south);

    
    \node[anchor=south east] at (-1, 4.5) {\Large $\frac{1}{R}\hat{E}_d$}; 
    
    \node[anchor=west] at (5, 0) {\Large $\frac{1}{R}\hat{E}_c$};      
    
    \node[anchor=west] at (3.5, 3) {\Large $\frac{1}{R}(\hat{E}_c - \hat{E}_d)$}; 

\end{tikzpicture}

  \caption{Multi-mode BBP homodyne configuration.  The signal modes
  enter from the left on modes \({\bm{a}}\).  The LO modes enter from
  the bottom on modes \({\bm{b}}\).  For mathematical convenience, we
  represent the LO modes as initially in vacuum, with the LO coherent
  state prepared by the displacement operator
  \(\hat{D}_{R\bm{\beta}}\) . Modes \({\bm{a}}\) and \({\bm{b}}\) are
  combined on a balanced beam splitter (BS). The beamsplitter's
  outgoing modes are \({\bm{c}}\) and \({\bm{d}}\). They are measured
  with detectors such as calorimeters. The observables associated with the detectors are
  of the form
  \(\hat{E}_{\bm{f}}=\sum_{k}\omega_{k}\hat{f}_{k}^{\dagger}
  \hat{f}_{k}\) with \(\bm{f}=\bm{c}\) or \(\bm{d}\), where
  \(\omega_{k}\) are mode energies or other known positive weights.
  The homodyne measurement result is determined by subtracting one
  detector's output from the other and rescaling the result by a
  factor of \(1/R\). The corresponding effective measurement operator
  is shown.  }
\label{fig:bbp_homodyne}
\end{figure}

The effective BBP measurement operator after subtraction and rescaling
can be computed by expressing the outgoing mode operators according to
Heisenberg evolution in terms of the incoming modes. With a specific
sign and phase convention for the balanced beamsplitter, the Heisenberg
forms of the outgoing
mode operators are
\begin{align}
  \hat{\bm{c}}^{H} &= \frac{1}{\sqrt{2}}\qty(\hat{\bm{a}}+\hat{\bm{b}}
                 + R\bm{\beta})
                 ,\notag\\
  \hat{\bm{d}}^{H} &= \frac{1}{\sqrt{2}}\qty(\hat{\bm{a}}-\hat{\bm{b}} - R\bm{\beta}).
\end{align}
The Heisenberg form of the BBP measurement operator is therefore
\begin{align}
  \frac{1}{R}\Delta \hat{E}^{H}
  & = \frac{1}{R}\sum_{k=1}^{N} \omega_{k} \big( {\hat{c}_{k}^{H}}{}^{\dagger} \hat{c}_{k}^{H} - {\hat{d}_{k}^{H}}{}^{\dagger} \hat{d}_{k}^{H} \big) \notag\\
  & = \sum_{k=1}^{N} \omega_{k} ( \hat{a}_{k}^{\dagger} \beta_{k} + \beta_{k}^{*} \hat{a}_{k} ) + \frac{1}{R}\sum_{k=1}^N \omega_{k} ( \hat{a}_{k}^{\dagger} \hat{b}_{k} + \hat{b}_{k}^{\dagger} \hat{a}_{k} )
                            . \label{eq:calop}
\end{align}
If we choose \(\bm{\beta}=-i\bm{\alpha}*(1/\bm{\omega})\), then the
measurement operator is the unnormalized target quadrature
\(\hat{q}_{\bm{\alpha}}\) on the signal modes \(\bm{a}\) plus an
operator proportional to \(1/R\):
\begin{align}
  \frac{1}{R}\Delta\hat{E}^{H}
  &= \hat{q}_{\bm{\alpha}}
    +\frac{1}{R}
    \qty(\qty(\bm{\omega}*\hat{\bm{b}})\cdot\hat{\bm{a}^{\dagger}}
    + \qty(\bm{\omega}*\hat{\bm{b}}^{\dagger})\cdot\hat{\bm{a}}
    ).
    \label{eq:bbpdiffop}
\end{align}
For the analysis of convergence of measurements
and states to those obtained by measuring the target quadrature
directly, we define \(\delta=1/R\) and write
\begin{align}
  \hat{q}_{\bm{\alpha},\delta}
  &=\hat{q}_{\bm{\alpha}} +
    \delta
    \qty(\qty(\bm{\omega}*\hat{\bm{b}})\cdot\hat{\bm{a}^{\dagger}}
    + \qty(\bm{\omega}*\hat{\bm{b}}^{\dagger})\cdot\hat{\bm{a}}
    )
    \label{eq:qdelta}
\end{align}%
for the BBP homodyne measurement operator computed for the LO
displacement with \(\bm{\beta}=-i\bm{\alpha}*(1/\bm{\omega})\).
Throughout this work, \(\hat{q}_{\bm{\alpha},\delta}\) is expressed as
an operator on the incoming mode, which subsumes the displacement and
beamsplitter used in the actual measurement.  The difference between
the BBP homodyne measurement operator and the target quadrature is
\begin{align}
  \hat{q}_{\bm{\alpha},\delta} - \hat{q}_{\bm{\alpha}}
  &= \delta \hat{\Delta}_{q},
    \label{eq:observable}
\end{align}
where we defined the difference operator
\(\hat{\Delta}_{q}=\qty(\qty(\bm{\omega}*\hat{\bm{b}})\cdot\hat{\bm{a}^{\dagger}}
+ \qty(\bm{\omega}*\hat{\bm{b}}^{\dagger})\cdot\hat{\bm{a}})\).
Because \(\hat{q}_{\bm{\alpha},\delta} - \hat{q}_{\bm{\alpha}}\) is an unbounded operator for all \(\delta>0\),
the effect of the difference depends on the incoming state of the
signal mode.

The scale of these weights \(\omega_{k}\) is arbitrary. Changing the
scale also changes the LO amplitude factor \(R=1/\delta\).  Physical
quantities depend on \(\delta\omega_{k}\) and are independent of the
scale used for the \(\omega_{k}\). We therefore normalize the
\(\omega_{k}\) so that their maximum value is \(1\). Standard pulsed homodyne
corresponds to the case where all \(\omega_{k}\) are equal.


\section{Measurement models}
\label{sec:functana}
For the remainder of the paper, we fix the target quadrature
\(\hat{q}_{\bm{\alpha}}\), where the quadrature is normalized,
\(|\bm{\alpha}|=1\). To simplify the notation we write
\(\hat{q}_{\delta}\) for the BBP homodyne measurement operator of Eq.~\eqref{eq:calop} and \(\hat{q}\) for the target quadrature, suppressing the suffix \(\bm{\alpha}\). We identify \(\hat{q}_{0}\) with \(\hat{q}\). Let \(d\Pi(x)\) and \(d\Pi_{\delta}(x)\) symbolize the projection-valued measures of the target quadrature \(\hat{q}\) and of the BBP measurement operator \(\hat{q}_{\delta}\).  We identify
\(d\Pi_{0}(x)\) with \(d\Pi(x)\). See App.~\ref{app:pvm} for a review of the notation and relevant properties for projection-valued measures.

By construction, for \(\delta>0\), \(\hat{q}_{\delta}\) is the
difference of two commuting, displaced and rotated, total energy operators, which
implies that it has a discrete spectrum. The displacement is according
to the initial displacement of the LO in our representation of the
measurement configuration, see Fig.~\ref{fig:bbp_homodyne}.  The
corresponding displaced number states are eigenstates of
\(\hat{q}_{\delta}\). For \(\lambda\) in the spectrum of
\(\hat{q}_{\delta}\), \(\Pi_{\delta}(\lambda)\) can be interpreted as
the projector onto the \(\lambda\)-eigenspace of
\(\hat{q}_{\delta}\). Because the spectrum of \(\hat{q}\) is
continuous and non-degenerate, \(\hat{q}\) has no eigenvectors.  The
putative projector onto the eigenstate with eigenvalue \(y\), which
would be expressed as \(\int_{x\in\rls} \delta_{y}(x) d\Pi(x)\), is
\(0\) because the set \(\{x\}\) has measure zero for continuous
spectra.  See the discussion in App.~\ref{app:pvm}. These
considerations indicate that the analysis of convergence of
\(\hat{q}_{\delta}\) or its measurement to \(\hat{q}\) or its
measurement, respectively, is not straightforward.

We model measurements as unitary processes that couple the system to
be measured to a quantum apparatus, where the apparatus outcomes are
defined by a projection-valued measure.  After this process, the
apparatus may decohere with respect to this measure, and except for such
decoherence, the apparatus no longer interacts with other
systems. This allows us to learn the outcome of the apparatus and
record it elsewhere. To model the usually destructive quadrature
measurements, the measured system is discarded.  For outcome-dependent
actions, the normal experimental order is to record the outcomes, and
then take the actions. This is mathematically equivalent to first
applying the unitary that conditional on the outcome basis of the
apparatus implements the corresponding action, and then decohering the
apparatus and recording outcomes if desired. We define the unitary
measurement model as the combination of apparatus coupling and, if
used, unitarily implemented conditional actions. Conditional actions
do not play a role in the remainder of this section, they are defined
and analyzed in Sect.~\ref{sec:condoperations}.

The measurements depend on how the apparatus is coupled to the
system. For the unitary measurement model, we use a
generalization of the von Neumann measurement
model~\cite{vonneumann1955foundations}. In the von Neumann measurement
model, apparatus couplings are implemented by applying unitaries of
the form
\(\hat{M}_{\delta}=e^{-i\hat{q}_{\delta}\otimes \hat{s}_{M}}\), where
\(M\) refers to the apparatus and \(\hat{s}_{M}\) is a self-adjoint
operator on the apparatus. The initial state of the apparatus is a
pure state \(\ket{g}_{M}\). Let $\ket{\psi}$ be the initial state of
the signal modes \(\bm{a}\), the LO modes \(\bm{b}\), and any
additional relevant systems. Then, we obtain the state
\begin{align}
  \ket{\phi_{\delta}}&=\hat{M}_{\delta}(\ket{\psi}\otimes\ket{g}_{M}).
                       \label{eq:psimdelta_general}
\end{align}
An explicit circuit diagram for the measurement process is shown in Fig.~\ref{fig:cu1}.
The LO modes always start in their vacuum state.
The initial state may be assumed to be pure by adding purifying quantum systems if necessary. In general, to quantify convergence we bound the fidelity from below or the Hilbert-space distance from above as a function of \(\delta\). For this purpose we review the relevant properties of fidelity.

\begin{figure}
  \includegraphics[scale = 0.7]{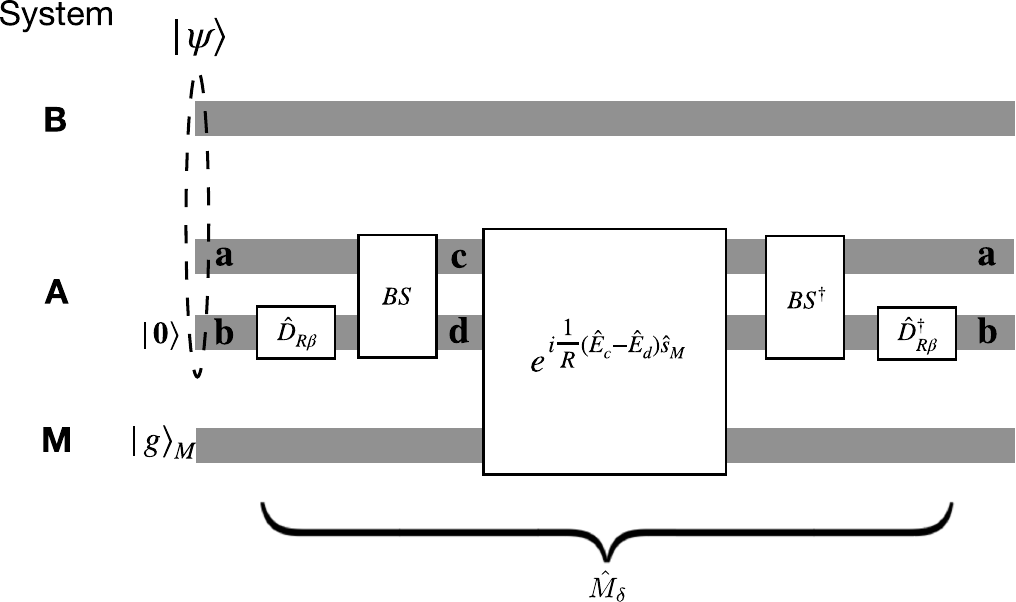}
  \caption{ An explicit circuit diagram for the measurement process
    with effective apparatus coupling
    \(\hat{M}_{\delta}=e^{-i\hat{q}_{\delta}\otimes \hat{s}_{M}}\).
    Here we show one way of representing \(\hat{M}_{\delta}\)
    explicitly in terms of the elements depicted in the schematic for
    BBP homodyne of Fig.~\ref{fig:bbp_homodyne}. The operator
    \(\hat{E}_{c}-\hat{E}_{d}\) in the diagram denotes the
    Schr\"odinger picture operator and is the physical energy
    difference between the two mode lines at the point where it is
    used.  The coupling to the apparatus in this unitary
    measurement model is directly in terms of the scaled energy
    difference operator shown in Fig.~\ref{fig:bbp_homodyne}, which is
    converted to \(\hat{q}_{\delta}\) as defined in
    Eq.~\eqref{eq:qdelta} by conjugating with the beamsplitter and
    displacement. The reversal of the beamsplitter and the
    displacement is used to define \(\hat{M}_{\delta}\) in the
    unitary measurement model. Since they act on subsequently
    ignored, destroyed or discarded systems they and any other
    subsequent action on these systems do not matter for the
    fidelity analysis.  In experiments, the actual measurement may
    be performed as described in Fig.~\ref{fig:bbp_homodyne}.  The
    system \(B\) includes all other relevant systems including any
    purifying systems.}
  \label{fig:cu1}
\end{figure}

The fidelity $F$ between pure states
$\ket{\phi_{1}}$ and $\ket{\phi_{2}}$ is given by
\begin{align}
  F &= |\braket{\phi_{1}}{\phi_{2}}|^{2}=
      \Re(\braket{\phi_{1}}{\phi_{2}})^{2} +
      \Im(\braket{\phi_{1}}{\phi_{2}})^{2} \notag\\
    &\geq
      \Re(\braket{\phi_{1}}{\phi_{2}})^{2}.
      \label{eq:fidRe}
\end{align}
For useful relationships between fidelities, overlaps and distances, see App.~\ref{app:inequalities}. The experimentally relevant fidelities need to account for the apparatus decoherence and outcome recording, which can be treated as a quantum channel and result in mixed states. For mixed states \(\sigma\) and \(\rho\), the fidelity
is given by
\(F(\sigma,\rho) = \qty(\tr(\qty|\sqrt{\sigma}\sqrt{\rho}|))^{2}\), where for bounded operators \(\hat{A}\) and \(\hat{B}\),
\(\tr(\qty|\hat{A}\hat{B}|)
=\tr(\sqrt{\hat{B}^{\dagger}\hat{A}^{\dagger}\hat{A}\hat{B}})\).

We take advantage of monotonicity of fidelity under the application of
quantum channels~\cite{hou2012fidelity} in order to perform our
analysis with pure states only.  According to monotonicity of
fidelity, if \(\cO\) is a quantum channel, then
\(F(\cO(\sigma),\cO(\rho))\geq F(\sigma,\rho)\).  We bound fidelities
between the pure states in the unitary measurement model and treat any
subsequent decoherence, outcome recording and discarding of quantum
systems as quantum channels applied identically to both of the states
being compared. In order to motivate and interpret the use of fidelity
as a measure of the quality of the measurement process, consider the
case where the measurement outcome $x$ is recorded in system $O$,
which may be a computer separate from the quantum apparatus. For
illustration we consider the case where the outcomes $x$ are real and,
given an initial state \(\rho\), the outcome distribution has
probability density \(\nu_{\rho}(x)\). The final state of the complete
measurement process including eventual readout to a classical system
is a quantum-classical state that may be written in the form
\(\cO(\rho)=\int dx \nu_{\rho}(x)\sigma_{\rho}(x)\otimes
\ketbra{x}_{O}\), where for \(x\in\rls\), \(\sigma_{\rho}(x)\) is a
density operator of the quantum systems remaining after the
measurement process has completed. System \(O\) contains classical
outcomes labeled by \(x\), with outcome states written with quantum
notation.  For states of this form, the fidelity is given by
\begin{align}
  F(\cO(\rho_{1}),\cO(\rho_{2}))
  &=
    \qty(\int dx \sqrt{\nu_{\rho_{1}}(x)}\sqrt{\nu_{\rho_{2}}(x)}\tr(\qty|\sqrt{\sigma_{\rho_{1}}(x)}\sqrt{\sigma_{\rho_{2}}(x)}|))^{2}.
\end{align}
This fidelity is a lower bound on the average of the
conditional-on-outcomes fidelities between the
\(\sigma_{\rho_{j}}(x)\), averaged with respect to either
\(\nu_{\rho_{1}}\) or \(\nu_{\rho_{2}}\).  To see this, let
\(\{k,l\}=\{1,2\}\) and use the Cauchy-Schwarz inequality as
follows:\footnote{According to convention, all displayed sequences of
  inequalities and equalities are to be read as if on one line.}
\begin{align}
  F(\cO(\rho_{1}),\cO(\rho_{2}))
  &\leq \left(\int dx \nu_{\rho_{k}}(x)\right)
    \int dx\nu_{\rho_{l}}(x)\tr(\qty|\sqrt{\sigma_{\rho_{1}}(x)}\sqrt{\sigma_{\rho_{2}}(x)}|)^{2}
    \notag\\
  &= \int dx \nu_{\rho_{l}}(x) F(\sigma_{\rho_{1}}(x),\sigma_{\rho_{2}}(x)),
\end{align}
where we used \(\int dx \nu_{\rho_{k}}(x)=1\) since
\(\nu_{\rho_{k}}(x)\) is a probability density. The average
conditional-on-outcome fidelity emphasizes the difference between the states of the quantum systems. The full expression for fidelity also accounts for differences between the observed probability distributions of the outcomes. The latter does not matter if, after taking conditional actions, the outcomes play no further role and are discarded.

To apply the measurement model to BBP homodyne measurements, we elaborate the unitary measurement model. Consider first a BBP measurement that is ideal aside from having $\delta>0$.  In this case one would wish for the apparatus outcome basis to consist of \(\ket{x}_{M}\) for \(x\) real, and for the state after the unitary measurement
process to be \(\ket{\phi_{\delta}}=\int_{x\in\rls}
d\Pi_{\delta}(x)\ket{\psi}\otimes\ket{x}_{M}\).  This makes sense for \(\delta>0\) because the spectrum of \(\hat{q}_{\delta}\) is discrete. But for \(\delta=0\), the state so expressed is not a proper Hilbert-space state. See the discussion in App.~\ref{app:pvm}. If one nevertheless attempts to assess the fidelity between 
\(\ket{\phi_{\delta}}\) and \(\ket{\phi_{0}}\) according to the above expression, it appears that the fidelity, to the extent it can be defined, is \(0\).

The ideal measurements considered in the previous paragraph are not
realizable by real-world measurements. For this reason, it is
necessary to consider realistic measurements. Realistic measurements
do not perfectly resolve the eigenvalues of the operator being
measured. In terms of the eventually recorded classical information,
this lack of resolution can be described by a stochastic process that,
conditional on the true eigenvalue \(x\), produces a measured value
\(z\) with transition probability density \(f(z|x)\).  In the unitary
measurement model, this effect is achieved with a coupling unitary
\(\hat{M}_{\delta}\) for which 
\begin{align}
  \ket{\phi_{\delta}}
  &=   \hat{M}_{\delta}\qty(\ket{\psi}\otimes\ket{g}_M )
    \notag\\
  &= \int_{x\in\rls}d\Pi_{\delta}(x)\qty(\ket{\psi}\otimes \int_{z\in\rls} dz f_{1/2}(z|x)\ket{z}_M)
    \label{eq:udelta}
\end{align}
where \(f_{1/2}(z|x)\) is a complex-valued function satisfying
\(|f_{1/2}(z|x)|^{2}=f(z|x)\).  Intuitively, after apparatus
decoherence, the value recorded in the apparatus behaves as if it had
been obtained by the stochastic process above. The coupling unitary
can in principle be implemented as a \(\hat{q}_{\delta}\)-conditional
unitary. Translation invariant transition probability densities
  can be represesented with coupling unitaries \(\hat{M}_{\delta}\) of
  the form given before Eq.~\eqref{eq:psimdelta_general}. To represent
  other transition probability densities, we define such unitaries
as unitaries that commute with \(\hat{q}_{\delta}\) and can be
expressed in the form
\begin{align}
  \hat{M}_{\delta} &= \int_{x\in\rls}d\Pi_{\delta}(x)\otimes \hat{U}(x)_{M}
                     \label{eq:udeltaform}
\end{align}
with \(\hat{U}(x)_{M}\) independent of \(\delta\). See App.~\ref{app:pvm} for requirements on \(\hat{U}(x)_{M}\) that ensure that \(\hat{M}_{\delta}\) is a well-defined unitary operator.


\section{Quantified convergence of post-measurement states}
\label{sec:sconv}
We first analyze the convergence of the post-measurement state without
outcome conditional actions, using the form of the unitary measurement
process given in Eq.~\eqref{eq:psimdelta_general}.  The main result is
Thm.~\ref{thm:qmeas}, which gives a bound on the fidelity between
\(\ket{\phi_{\delta}}\) and \(\ket{\phi_{0}}\) that depends on the
coefficients of the target quadrature, the expectation of an
energy-related operator on the signal modes, and the second and fourth
moment of \(\hat{s}_{M}\). The bound on fidelity yields quantified
convergence of \(f(\hat q_{\delta})\ket{\psi}\) for continuous bounded
functions \(f\) given that \(\ket{\psi}\) has bounded expectation of
the energy-related operator. We end the section with a lemma that
makes it possible to apply the quantified convergence results to
unbounded measurable functions by regularizing these functions.  The
bounds obtained in this section for BBP homodyne can be improved at
the cost of more involved expressions with a more detailed analysis of
the inequalities in App.~\ref{app:appendix3}. For standard pulsed
homodyne the expressions simplify and yield better bounds, as we show
in App.~\ref{app:standardhomodyne}.

We begin with a fundamental bound that is used both for estimating the fidelity of the measurement process and the fidelity of states obtained after conditional operations.  For this bound we consider a general system-apparatus state that is vacuum on the LO modes. We normally express the bounds in terms of distances rather than fidelities,
relying on the relationships in App.~\ref{app:fidelities} to convert between them.

\begin{theorem}\label{thm:mainbnd}
  Let \(\ket{\psi}\) be a joint state of the signal, LO modes, the
  apparatus \(M\) and any other relevant systems including those
  needed for purifying the state. Assume that \(\ket{\psi}\) is vacuum
  on the LO modes.  For each mode \(a_{k}\), let \(\hat{n}_{k}\) be
  its number operator. Define
  \begin{align}
    \hat{\Omega}&=\sum_{k}\omega_{k}^{2}\hat{n}_{k},
                  \notag\\
    \overline{\omega^{2}}
    &= \sum_{k}|\alpha_{k}|^{2}\omega_{k}^{2}.
  \end{align}
  Then
  \begin{align}
    \qty|e^{-i\hat q_{\delta}s}\ket{\psi} - e^{-i\hat{q}s}\ket{\psi}|^{2}
    &\leq
       4(\delta\,s)^{2}\bra{\psi}
    \qty(\qty(1+s^{2})\hat{\Omega} + s^{2}\overline{\omega^{2}})
      \ket{\psi}.
      \label{thm:eq:mains}
  \end{align}
  For every self-adjoint apparatus operator \(\hat{s}_{M}\),
  \begin{align}
    \qty|e^{-i\hat q_{\delta}\hat{s}_{M}}\ket{\psi}
    - e^{-i\hat{q}\hat{s}_{M}}\ket{\psi}|^{2}
    &\leq
       4\delta^{2}\bra{\psi}
    \hat{s}_{M}^{2}\qty(\qty(1+\hat{s}_{M}^{2})\hat{\Omega} + \hat{s}_{M}^{2}\overline{\omega^{2}})
      \ket{\psi}.
      \label{thm:eq:mainhats}
  \end{align}
\end{theorem}

\begin{proof}
  The inequality of Eq.~\eqref{thm:eq:mains} is proven in
  App.~\ref{app:appendix3} and is identical to the last inequality of
  Eq.~\eqref{eq:simplifiedbnd0}.
  For more detailed and tighter bounds, see the derivation of this
  inequality. In particular, if the contributions from
  \(s^{2}(\delta\,s)^{2}\hat{\Omega}\) dominate the bound, the bounds
  may be improvable by applying Eq.~\eqref{eq:simplifiedbnd1} instead
  of Eq.~\eqref{eq:simplifiedbnd0}.  This comes at the cost of a
  degree \(6\) contribution in \(s\). See
  App.~\ref{app:standardhomodyne} for the necessary adjustments
  for such contributions in standard pulsed
  homodyne.  Up to constants, these tighter bounds cannot be improved
  as shown in the last part of App.~\ref{app:appendix3} by estimates for coherent
  states. For unitary operators \(\hat{U}\) and \(\hat{V}\),
  \(|\hat{U}\ket{\psi}+\hat{V}\ket{\psi}|^{2} \leq 2\), so the bounds
  in the theorem can be replaced by the minimum of \(2\) and the
  bounds shown, if desired.

  The inequality of Eq.~\eqref{thm:eq:mainhats} follows from
  Eq.~\eqref{thm:eq:mains} by expanding the state with respect to the
  projection-valued measure \(d\Pi_{M}(s)\) of \(\hat{s}_{M}\)
  according to
  \(\int_{s\in\rls}d\Pi_{M}(s)\ket{\psi} =
  \int_{s\in\rls}^{\oplus}d\mu(s)\ket{\psi(s)}\), where \(d\mu(s)\) is
  the measure in the direct-integral representation of \(\ket{\psi}\)
  with respect to \(d\Pi_{M}(s)\) (App.~\ref{app:pvm}). The
  conditional states \(\ket{\psi(s)}\) satisfy the conditions on
  \(\ket{\psi}\) in the theorem statement and can be substituted for
  \(\ket{\psi}\) in Eq.~\eqref{thm:eq:mainhats}.  Abbreviate
  \(\hat{A}= e^{-i\hat q_{\delta}\hat{s}_{M}} -
  e^{-i\hat{q}\hat{s}_{M}}\), which commutes with \(\hat{s}_{M}\).The
  unnormalized states \(\ket{\psi(s)}\) naturally belong to Hilbert
  spaces of the form \(\cH_{l}\otimes\cH_{R}\), where \(\cH_{R}\) is
  the Hilbert space of the systems not including the apparatus and
  \(\cH_{l}\) is one of the Hilbert spaces of the direct-integral
  representation for \(\hat{s}_{M}\) on the apparatus Hilbert space.
  Because \(\hat{s}_{M}\) has eigenvalue \(s\) on the Hilbert space
  associated with \(s\in\rls\) in the direct integral, the operator
  \(\hat{A}\) acts as \(e^{-i\hat{q}_{\delta}s}-e^{i\hat{q}s}\) on
  this Hilbert space, and only on the factor \(\cH_{R}\) .
  Consequently,
  \(\int_{s\in\rls}\hat{A}^{\dagger}\hat{A}\,d\Pi_{M}(s)\ket{\psi} =
  \int_{s\in\rls}^{\oplus}d\mu(s)\hat{A}^{\dagger}\hat{A}\ket{\psi_{s}}\).
  We have
  \begin{align}
        \qty|e^{-i\hat q_{\delta}\hat{s}_{M}}\ket{\psi}
    - e^{-i\hat{q}\hat{s}_{M}}\ket{\psi}|^{2}
    &=
      \bra{\psi}\hat{A}^{\dagger}\hat{A}\ket{\psi}
      \notag\\
    &=
      \int_{s\in\rls}\bra{\psi}\hat{A}^{\dagger}\hat{A}d\Pi_{M}(s)\ket{\psi}
      \notag\\
    &=\int_{s\in\rls}d\mu(s)\bra{\psi_{s}}\hat{A}^{\dagger}\hat{A}\ket{\psi_{s}}
      \notag\\
    &=\int_{s\in\rls}d\mu(s)
              \qty|e^{-i\hat q_{\delta}\hat{s}_{M}}\ket{\psi_{s}}
      - e^{-i\hat{q}\hat{s}_{M}}\ket{\psi_{s}}|^{2}
      \notag\\
    &\leq
      \int_{s\in\rls}d\mu(s)
             4\delta^{2}\bra{\psi_{s}}s^{2}
    \qty(\qty(1+s^{2})\hat{\Omega} + s^{2}\overline{\omega^{2}})
      \ket{\psi_{s}}
      \notag\\
    &=
      \int_{s\in\rls}
             4\delta^{2}\bra{\psi}s^{2}
    \qty(\qty(1+s^{2})\hat{\Omega} + s^{2}\overline{\omega^{2}})
      d\Pi(s)\ket{\psi}
      \notag\\
    &=
       4\delta^{2}\bra{\psi}
    \hat{s}_{M}^{2}\qty(\qty(1+\hat{s}_{M}^{2})\hat{\Omega} + \hat{s}_{M}^{2}\overline{\omega^{2}})
      \ket{\psi},
  \end{align}
  which is the right-hand side of the desired inequality.
\end{proof}

The next theorem is a special case of Thm.~\ref{thm:mainbnd} and
determines the distance between the states after the unitary
measurement process, which determines the fidelity between the
states. As discussed in Sect.~\ref{sec:functana}, when expressed in
terms of fidelity, the bounds are valid for the quantum-classical
states obtained after any further processing independent of
\(\delta\), including decoherence, classical readouts and discarding
of systems. For a wavefunction \(g(s)\) in \(L^{2}(\rls)\), we write
the pure quantum state corresponding to \(g(s)\) as
\(\ket{g(\bm{\cdot})}\).

\begin{theorem} \label{thm:qmeas} Let \(\hat{\Omega}\) and
  \(\overline{\omega^{2}}\) be as defined in
  Thm.~\ref{thm:mainbnd}. Let \(\hat{s}_{M}\) have non-degenerate
  spectrum so that the apparatus Hilbert space has a representation
  \(L^{2}(\rls, d\mu(s))\), and let \(\ket{g(\bm{\cdot})}\) be the
  initial state of the apparatus. Define
  \begin{align}
    b_{g,2l} &= \int_{s\in\rls} d\mu(s) s^{2l}|g(s)|^{2}.
  \end{align}
  Let \(\ket{\psi}\) be the initial state of the signal, LO modes and
  systems other than the apparatus, where \(\ket{\psi}\) is vacuum on
  the LO modes. Then, with \(\ket{\phi_{\delta}}\) as in
  Eq.~\eqref{eq:psimdelta_general}, we have
  \begin{align}
    \qty|\vphantom{f^{2}}\ket{\phi_{\delta}}-\ket{\phi_{0}}|^{2}
    &\leq
      4\delta^{2}\qty(
      b_{ g,2}\bra{\psi}\hat{\Omega}\ket{\psi}
        + b_{ g,4}\qty(\bra{\psi}\hat{\Omega}\ket{\psi}+ \overline{\omega^{2}})
      ).
      \label{eq:qmeastheorem1}
  \end{align}
\end{theorem}

\begin{proof}
  We substitute \(\ket{\psi}\otimes \ket{g(\bm{\cdot})}\) for \(\ket{\psi}\) 
  in Thm.~\ref{thm:mainbnd}. With this substitution, we have
  \begin{align}
    b_{g,2l}&= (\bra{\psi}\otimes\bra{g(\bm{\cdot})})\hat{s}^{2l}
              (\ket{\psi}\otimes\ket{g(\bm{\cdot})}).
  \end{align}
  The inequality of Eq.~\eqref{eq:qmeastheorem1} then follows by substitution
  into the inequality of Eq.~\eqref{thm:eq:mainhats}.
\end{proof}

The following proposition quantifies the convergence of functions of \(\hat{q}_{\delta}\) in a state-dependent way.  

\begin{proposition}\label{prop:fbndedmeas}
  Let \(\hat{\Omega}\) and
  \(\overline{\omega^{2}}\) be as defined in
  Thm.~\ref{thm:mainbnd}.
  Suppose that \(f(x)\) is the Fourier transform of a bounded
  complex measure \(d\tilde\mu(\kappa)\). Write \(d\tilde\mu(\kappa)=e^{i\theta(\kappa)}d\mu(\kappa)\) for a bounded positive measure \(d\mu(\kappa)\).
  Define \(\tilde{f}_{l}=\int_{\kappa\in\rls} d\mu(\kappa) |\kappa|^{l}\). Then for all states \(\ket{\psi}\) of the signal and LO modes and other relevant systems,
  where \(\ket{\psi}\) is vacuum on the LO modes, we have the following bound:
  \begin{align}
    \qty|\vphantom{f^{2}}f(\hat{q}_{\delta})\ket{\psi}-f(\hat{q}_{0})\ket{\psi}|^{2}
    &\leq
      4\delta^{2}(\tilde f_{0}+\tilde f_{2})
      \tilde f_{2}\bra{\psi}(\hat{\Omega}+\overline{\omega^{2}})\ket{\psi}.
                                                              \label{eq:cor:cfunq}
  \end{align}
\end{proposition}

See App.~\ref{app:pvm} for a brief review of bounded complex measures
and our convention for Fourier transforms.

\begin{proof}
  According to the spectral theorem in projection-valued measure form, for every normalized state \(\ket{\phi}\), we have
  \begin{align}
    \bra{\phi}f(\hat{q}_{\delta})\ket{\psi}
    &= \int_{x\in\rls} \bra{\phi} d\Pi_{\delta}(x)\ket{\psi} f(x)\notag\\
    &=\int_{x\in{\rls}} \bra{\phi}d\Pi_{\delta}(x)\ket{\psi} \int_{\kappa\in\rls} d\tilde\mu(\kappa) e^{ix\kappa}.
      \label{eq:overlapmeasure}
  \end{align}
  The expression \(\bra{\phi}d\Pi_{\delta}(x)\ket{\psi}\) is a bounded
  complex measure on \(\rls\), so we can write
  \(\bra{\phi}d\Pi_{\delta}(x)\ket{\psi} =e^{i\theta'(x)}d\nu(x)\)
  with \(d\nu(x)\) a bounded positive measure. The expression
  \(d\nu(x) d\tilde\mu(\kappa)\)  defines a bounded complex measure on
  \(\rls^{2}\). Because the integrands and measures are bounded, we can apply the Fubini-Tonelli theorem to write the last expression in
  Eq.~\eqref{eq:overlapmeasure} as
  \begin{align}
    \int_{x\in\rls} \bra{\phi}d\Pi_{\delta}(x)\ket{\psi} \int_{\kappa\in\rls} d\tilde\mu(\kappa) e^{ix\kappa}
    &=
      \int_{x\in\rls} d\nu(x) e^{i\theta'(x)} \int_{\kappa\in\rls} d\mu(\kappa)e^{ix\kappa}e^{i\theta(\kappa)}  \notag\\
    &=
      \int_{\kappa\in\rls}  d\mu(\kappa) e^{i\theta(\kappa)}\int_{x\in\rls} d\nu(x)e^{ix\kappa}e^{i\theta'(x)}\notag\\
    &=
      \int_{\kappa\in\rls} d\tilde\mu(\kappa) \int_{x\in\rls} \bra{\phi}d\Pi_{\delta}(x)\ket{\psi} e^{ix\kappa}\notag\\
    &=
      \int_{\kappa\in\rls} d\tilde\mu(\kappa)\bra{\phi} e^{i\kappa\hat{q}_{\delta}}\ket{\psi}.
  \end{align}
  We can therefore use the Cauchy-Schwarz inequality
  and apply Eq.~\eqref{thm:eq:mains} to bound
  \begin{align}
    \qty|\vphantom{f^{2}}\bra{\phi}f(\hat{q}_{\delta}) - f(\hat{q}_{0})\ket{\psi}|^{2}
    &=
      \qty|\int_{\kappa\in\rls} d\tilde\mu(\kappa)
      \bra{\phi}e^{i\kappa\hat{q}_{\delta}}-
      e^{i\kappa\hat{q}_{0}}\ket{\psi}|^{2}
      \nonumber\\
    &=
      \qty|\int_{\kappa\in\rls} d\mu(\kappa)e^{i\theta(\kappa)}
      \bra{\phi}e^{i\kappa\hat{q}_{\delta}}-
      e^{i\kappa\hat{q}_{0}}\ket{\psi}|^{2}
      \nonumber\\
    &\leq \qty(\int_{\kappa\in\rls} d\mu(\kappa)
      \qty|\bra{\phi}\qty(e^{i\kappa\hat{q}_{\delta}}-
      e^{i\kappa\hat{q}_{0}})\ket{\psi}|)^{2}
      \nonumber\\
    &\leq \qty(\int_{\kappa\in\rls} d\mu(\kappa)
      \qty|\qty(e^{i\kappa\hat{q}_{\delta}}-
      e^{i\kappa\hat{q}_{0}})\ket{\psi}|)^{2}
      \notag\\
    &=
      \qty(\int_{\kappa\in\rls}
      d\mu(\kappa)\sqrt{1+\kappa^{2}}\frac{1}{\sqrt{1+\kappa^{2}}}
      \qty|\qty(e^{i\kappa\hat{q}_{\delta}}-
      e^{i\kappa\hat{q}_{0}})\ket{\psi}|)^{2}
      \notag\\
    &\leq
      \int_{\kappa\in\rls}d\mu(\kappa)(1+\kappa^{2})
      \int_{\kappa\in\rls}d\mu(\kappa)\frac{1}{1+\kappa^{2}}
      \qty|\qty(e^{i\kappa\hat{q}_{\delta}}-
      e^{i\kappa\hat{q}_{0}})\ket{\psi}|^{2}
      \notag\\
    &\leq 4\delta^{2} (\tilde f_{0}+\tilde f_{2})
      \int_{\kappa\in\rls}d\mu(\kappa)\frac{1}{1+\kappa^{2}}
      \kappa^{2}
      \qty(\bra{\psi}\qty(1+\kappa^{2})\hat{\Omega} +
      \kappa^{2}\overline{\omega^{2}})
      \ket{\psi}
      \notag\\
    &\leq 4\delta^{2}(\tilde f_{0}+\tilde f_{2})
      \int_{\kappa\in\rls}d\mu(\kappa)
      \kappa^{2}\bra{\psi}(\hat{\Omega}+\overline{\omega^{2}})\ket{\psi}
      \notag\\
    &\leq 4\delta^{2}(\tilde f_{0}+\tilde f_{2}) \tilde
      f_{2}\bra{\psi}(\hat{\Omega}+\overline{\omega^{2}})\ket{\psi}.
\label{eq:propfbndedmeas}
  \end{align}
  Here we used \(\kappa^{2}/(1+\kappa^{2}) \leq 1\) to obtain the
  second to last line.  Since this inequality holds for all normalized
  states $\ket{\phi}$, the inequality of Eq.~\eqref{eq:cor:cfunq}
  follows.
\end{proof}

For a generalization of Prop.~\ref{prop:fbndedmeas}, see
App.~\ref{app:generalize_fbndedmeas}. One can use regularization to apply Prop.~\ref{prop:fbndedmeas} to
functions that are not the Fourier transform of a bounded complex measure or for which \(\tilde f_{0}\) or \(\tilde f_{2}\) are too large or infinite. The regularization error can be determined by means of the next proposition. In this proposition, \(f_{0}(x)\) is the function of interest, and \(f(x)\) is its regularization. 

\begin{proposition}\label{prop:regfexp}
  Let \(\hat{\Omega}\) and \(\overline{\omega^{2}}\) be as defined in
  Thm.~\ref{thm:mainbnd}.  Let \(f_{0}(x), f(x)\) be real-valued functions,
  and \(h(x)\) a non-negative function.  Suppose that \(f(x)\) is the
  Fourier transform of a bounded complex measure as in
  Prop.~\ref{prop:fbndedmeas} and satisfies
  \(|f_{0}(x) - f(x)|\leq h(x)\).  Then for the states \(\ket{\psi}\)
  of Prop.~\ref{prop:fbndedmeas} 
  \begin{align}
    |f(\hat{q}_{\delta})\ket{\psi} -  f_{0}(\hat{q})\ket{\psi}|^{2}
    &\leq
            \qty(2\delta\sqrt{(\tilde f_{0}+\tilde f_{2})
      \tilde f_{2}}\sqrt{\bra{\psi}(\hat{\Omega}+\overline{\omega^{2}})\ket{\psi}}
      +\sqrt{\bra{\psi} h(\hat{q})^{2}\ket{\psi}})^{2}.
      \\
    &\leq
      8\delta^{2}(\tilde f_{0}+\tilde f_{2})
      \tilde f_{2}\bra{\psi}(\hat{\Omega}+\overline{\omega^{2}})\ket{\psi}
      +2\bra{\psi} h(\hat{q})^{2}\ket{\psi}.
  \end{align}
\end{proposition}

\begin{proof}
  The inequality \(|f_{0}(x) - f(x)|\leq h(x)\) implies that
  \((f_{0}(\hat{q})-f(\hat{q}))^{2}\leq h(\hat{q})^{2}\). We compute
  \begin{align} 
    | \qty(f(\hat{q}_{\delta}) -  f_{0}(\hat{q}))\ket{\psi}|^{2}
    &\leq
      \qty(| \qty(f(\hat{q}_{\delta})- f(\hat{q}))\ket{\psi}|
      + | \qty(f(\hat{q}) -  f_{0}(\hat{q}))\ket{\psi}|)^{2}
      \notag\\
    &=
      \qty(| \qty(f(\hat{q}_{\delta})- f(\hat{q}))\ket{\psi}|
      + \sqrt{\bra{\psi}\qty(f(\hat{q}) -  f_{0}(\hat{q}))^{2}\ket{\psi}})^{2}
      \notag\\
    &\leq 
      \qty(| \qty(f(\hat{q}_{\delta})- f(\hat{q}))\ket{\psi}|
      + \sqrt{\bra{\psi}\qty(h(\hat{q}))^{2}\ket{\psi}})^{2}
      \notag\\
  \end{align}
  The first inequality now follows by applying
  Prop.~\ref{prop:fbndedmeas}.  For the second inequality apply
  \((a+b)^{2}\leq 2(a^{2}+b^{2})\).
\end{proof}

Props.~\ref{prop:fbndedmeas} and~\ref{prop:regfexp} can be used to
bound the difference between the expectations of
\(f(\hat{q}_{\delta})\) and \(f(\hat{q})\).  
By applying Eq.~\eqref{app:eq:a-bexp} we get
\(|\bra{\psi} f(\hat{q_{\delta}})\ket{\psi}-\bra{\psi}f(\hat{q})\ket{\psi}|^{2} \leq
|(f(\hat{q}_{\delta}) - f(\hat{q}_{0}))\ket{\psi}|^{2}\).


\section{Quantified convergence for outcome-conditional actions}
\label{sec:condoperations}

We next consider unitary operations \(\hat{U}_{B}(x)\) applied
conditionally to another quantum system $B$ based on measurement
outcomes \(x\) from a BBP homodyne measurement of \(\hat{q}\).  To
simplify the notation and avoid confusion with tensor product
orderings, we omit the tensor product symbol $\otimes$ and use the
system labels $\bm{a}$, \(\bm{b}\), $B$ and $M$ when needed. We use
\(A\) for the joint system consisting of modes \(\bm{a}\) and
\(\bm{b}\).  The ideal conditional operation on the joint system of
the signal modes, the LO modes and system \(B\) may be expressed as
\begin{align}
\hat{\cC}_{0}&=\int_{x\in\rls} d\Pi(x) \hat{U}_{B}(x).
\end{align}
In this section we compare the ideal conditional operation
to the operation that applies the unitaries conditional on a BBP measurement outcome modeled as described in Sect.~\ref{sec:functana}.
We do this in two steps. In the first step we compare \(\hat{\cC}_{0}\) to the conditional operation
\(\hat{\cC}_{\delta}\) that applies the unitaries conditional on \(\hat{q}_{\delta}\) instead of \(\hat{q}\)
\begin{align}
\hat{\cC}_{\delta}&=\int_{x\in\rls} d\Pi_{\delta}(x) \hat{U}_{B}(x).
\end{align}
In the second step we compare \(\hat{\cC}_{\delta}\) to the operation that first implements the unitary measurement process \(\hat{M}_{\delta}\) for \(\hat{q}_{\delta}\), then conditions the unitaries on the apparatus \(M\) instead of the signal modes. The conditional unitary
after the unitary measurement process is
\begin{align}
  \hat{\cC}_{M} &=
                         \int_{z\in\rls}d\Pi_{M}(z)\hat{U}_{B}(z),
\end{align}
where \(d\Pi_{M}(z)\) is the projection valued measure associated with the apparatus outcomes. We assume that the unitary measurement process \(\hat{M}_{\delta}\) is of the form given in Eq.~\eqref{eq:udeltaform} with the action given in Eq.~\eqref{eq:udelta}, so that \(\hat{M}_{\delta}\) commutes with \(\hat{\cC}_{\delta}\).

In applications such as quantum teleportation, the measured modes and the measurement outcomes are discarded after the conditional operations have been performed. In any case, decoherence in the outcome basis of the apparatus does not affect the performance of the operation. If decoherence in the control basis of the measured system also does not matter, then we can improve the comparison by applying
additional unitaries diagonal in the outcome and control bases to the apparatus and the measured modes after the conditional unitaries. 
Here we need such additional unitaries only for the case where the conditional unitaries involve non-commuting displacements. The additional unitaries are of the form
\begin{align}
  \hat{W}_{AM} &=
                 \int_{z\in\rls,x\in\rls} d\Pi_{M}(z) d\Pi_{\delta}(x) e^{iw(z,x)},
\end{align}
for a real-valued function \(w(z,x)\)

Let \(\ket{\psi}\) be the joint initial state of \(\bm{a}\), \(B\) and the LO modes \(\bm{b}\), where \(\ket{\psi}\) is vacuum on the LO modes. We define the following states:
\begin{align}
  \ket{\psi_{c,\delta}}
  &=
    \hat{\cC}_{\delta}\ket{\psi}\ket{g}_{M},
    \notag\\
  \ket{\psi_{m,\delta}}
  &=
    \hat{M}_{\delta}\hat{\cC}_{\delta}\ket{\psi}\ket{g}_{M},
    \notag\\
  \ket{\tilde\psi_{m,\delta}}
  &=
    \hat{W}_{AM}\hat{\cC}_{M}\hat{M}_{\delta}\ket{\psi}\ket{g}_{M}.
    \label{eq:condstatedef}
\end{align}
Here, \(\ket{\psi_{c,0}}\) is the state obtained after applying the
ideal conditional unitary \(\hat{\cC}_{0}\) without any measurement
processes. The state after implementing the conditional unitary
conditional on BPP measurement outcomes with the additional unitary is
\(\ket{\tilde\psi_{m,\delta}}\). The goal is to compare the two,
taking into account that decoherence of the measured system and
apparatus do not affect the performance.  We first bound the distance
from \(\ket{\psi_{c,\delta}}\) to \(\ket{\psi_{c,0}}\) and then the
distance from \(\ket{\psi_{m,\delta}}\) to
\(\ket{\tilde\psi_{m,\delta}}\). The distance between
\(\ket{\psi_{c,\delta}}\) and \(\ket{\psi_{c,0}}\) is the same as the
distance between \(\hat{M}_{\delta}\ket{\psi_{c,\delta}}\) and
\(\hat{M}_{\delta}\ket{\psi_{c,0}}\). We can use Thm.~\ref{thm:qmeas}
to bound the distance between \(\hat{M}_{0}\ket{\psi_{c,0}}\) and
\(\hat{M}_{\delta}\ket{\psi_{c,0}}\).  The state
\(\hat{M}_{0}\ket{\psi_{c,0}}\) is what would have been obtained after
the ideal conditional unitary with a subsequent unitary measurement
process for \(\hat{q}\). The bounds obtained give distance bounds for every
step in the chain
\begin{align}
  \hat{M}_{0}\ket{\psi_{c,0}} \rightarrow \hat{M}_{\delta}\ket{\psi_{c,0}} \rightarrow \hat{M}_{\delta}\ket{\psi_{c,\delta}}
  =\ket{\psi_{m,\delta}}  \rightarrow \ket{\tilde{\psi}_{m,\delta}}.
  \label{eq:condchain}
\end{align}
By adding the bounds along the chain, we obtain a bound on the
distance from \(\hat{M}_{0}\ket{\psi_{c,0}}\) to
\(\ket{\tilde{\psi}_{m,\delta}}\), which yields a lower bound on the fidelity between the two states. By monotonicity of fidelity, for implementations of measurement-conditional unitaries with outcome and control decoherence, the overall fidelity is at least as large.
Circuit diagrams for the states in each step of the chain are shown in Fig.~\ref{fig:cchain}.

\begin{figure}
  \begin{tabular}{l}
    \includegraphics[scale = 0.7]{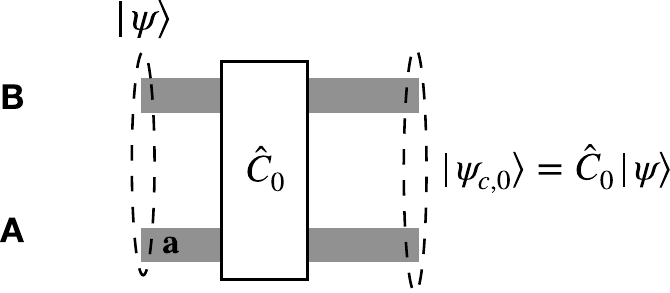}\\
    \includegraphics[scale = 0.7]{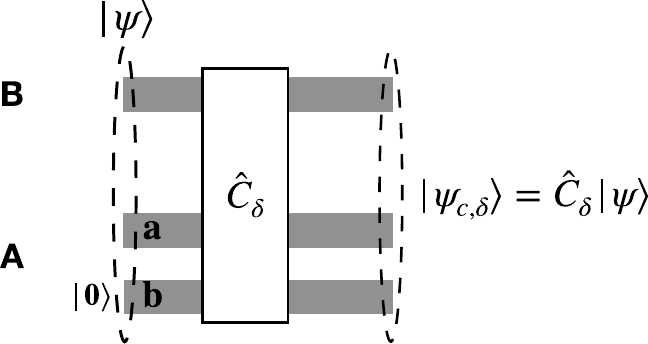}\\
    \includegraphics[scale = 0.7]{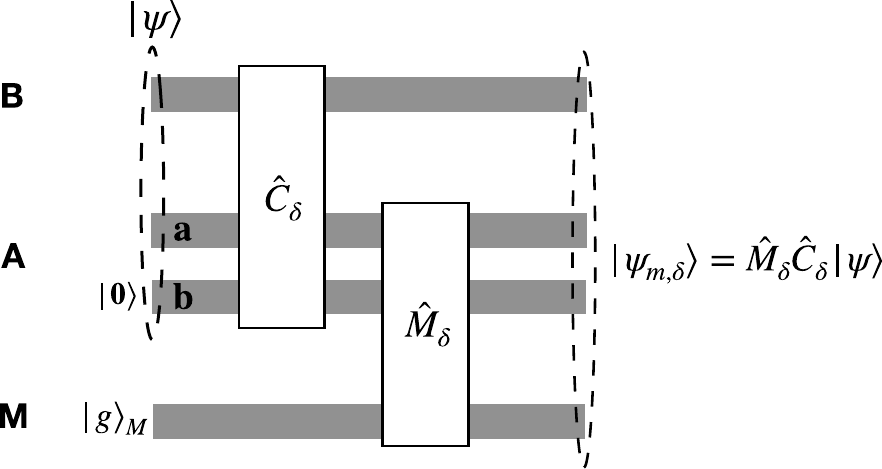}\\
    \includegraphics[scale = 0.7]{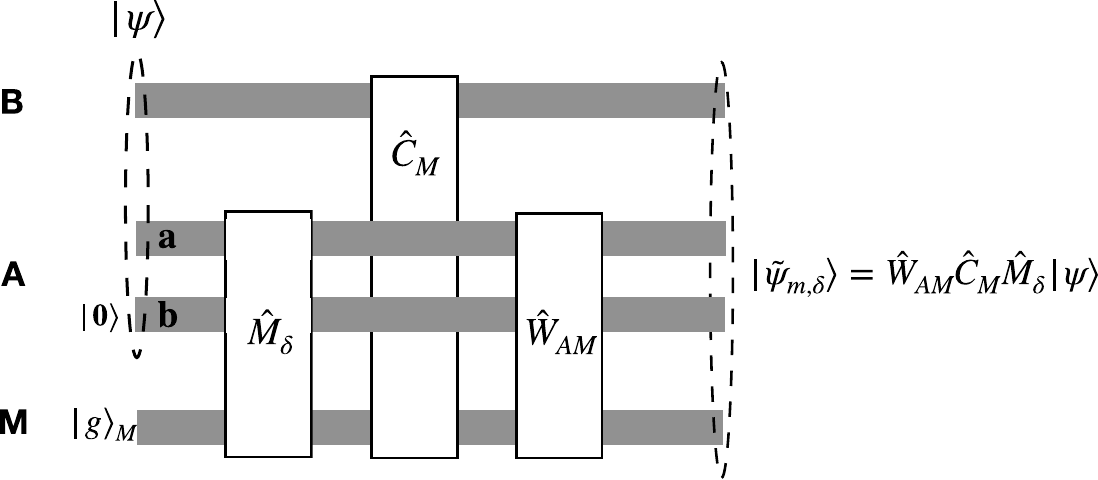}
  \end{tabular}
  \caption{Circuit diagrams for the main states in the distance bounding chain
  of Eq.~\eqref{eq:condchain}. }
  \label{fig:cchain}
\end{figure}

To bound the distance between \(\ket{\psi_{c,\delta}}\) to
\(\ket{\psi_{c,0}}\), we assume that the \(\hat{U}_{B}(x)\) mutually commute for different values of \(x\). We first prove a ``master'' proposition that generalizes
Prop.~\ref{prop:fbndedmeas}. We then apply this proposition to the case where the conditional unitaries are displacements.

\begin{proposition}
  \label{prop:condphase}
  Let \(d\Pi_{B}(y)\) be a projection-valued measure of system \(B\) with \(y\in\cY\). 
  Let \(f(x,y)\) be a function on \(\rls \times \cY\) and define
  \begin{align}
    \hat{f}_{\delta} & = \int_{x\in\rls,y\in\cY}f(x,y)d\Pi_{\delta}(x)d\Pi_{B}(y)
                       = \int_{y\in\cY}f(\hat{q}_{\delta},y)d\Pi_{B}(y).
  \end{align}
  Suppose that for each \(y\), the function \(x\mapsto f(x,y)\) is the
  Fourier transform of a bounded complex measure
  \(d\tilde\nu_{y}(\omega) = e^{i\phi(\omega)}d\nu_{y}(\omega)\). Let
  \(\tilde f_{l}(y)\) be functions on \(\cY\) satisfying
  \(\tilde f_{l}(y) \geq \int |\omega|^{l}d\nu_{y}(\omega)\) for all
  \(y\). Let \(\hat{\Omega}\) and \(\overline{\omega^{2}}\) be as
  defined in Thm.~\ref{thm:mainbnd}. Define the positive operator
  \begin{align}
    \hat{Q} &= \int_{y\in\cY} d\Pi_{B}(y)\qty(\tilde f_{0}(y)+\tilde f_{2}(y))\tilde f_{2}(y)
              \qty(\hat\Omega +
              \overline{\omega^{2}}).
              \label{eq:prop:condphase1}
  \end{align}
  Then for all states \(\ket{\psi}\) of the signal and LO modes,
  system \(B\) and other relevant systems, where \(\ket{\psi}\) is
  vacuum on the LO modes, we have
  \begin{align}
    \qty| \hat{f}_{\delta}\ket{\psi}-\hat{f}_{0}\ket{\psi}|^{2}
    &\leq 4\delta^{2}\bra{\psi}\hat{Q}\ket{\psi}.
  \end{align}
\end{proposition}

\begin{proof}
  We have
  \begin{align}
    \hat{f}_{\delta}\ket{\psi}-\hat{f}_{0}\ket{\psi}
    &=\int d\Pi_{B}(y)\qty(f(\hat{q}_{\delta},y) - f(\hat{q},y))\ket{\psi}.
  \end{align}
  Let \(d\mu(y)\) be the probability measure defined by
  \(d\mu(y)=\bra{\psi}d\Pi_{B}(y)\ket{\psi}\). Let \(\ket{\psi(y)}\)
  be a square-integrable family of states with respect to \(d\mu(y)\)
  that represents \(\ket{\psi}\) in the direct-integral representation
  for \(d\Pi_{B}(y)\). The Hilbert space \(\cH(y)\) of
  \(\ket{\psi(y)}\) factors naturally into the Hilbert space of
  systems other than \(B\) and a \(y\)-dependent Hilbert space
  associated with the direct integral representation of
  \(d\Pi_{B}(y)\) on \(B\). With this factorization, in the direct
  integral representation, \(\hat{f}_{\delta}\) acts as
  \(f(\hat{q}_{\delta},y)\) on \(\ket{\psi(y)}\).  We apply
  Prop.~\ref{prop:fbndedmeas} for each \(y\in\cY\) to get
  \begin{align}
    \qty|\hat{f}_{\delta}\ket{\psi}-\hat{f}_{0}\ket{\psi}|^{2}
    &=
      \int_{y\in\cY} \bra{\psi}d\Pi_{B}(y)(f(\hat{q}_{\delta},y)-f(\hat{q},y))^{\dagger}
      (f(\hat{q}_{\delta},y)-f(\hat{q},y))\ket{\psi}
      \notag\\
    &=
      \int_{y\in\cY}  d\mu(y) \bra{\psi(y)}(f(\hat{q}_{\delta},y)-f(\hat{q},y))^{\dagger}
      (f(\hat{q}_{\delta},y)-f(\hat{q},y))\ket{\psi(y)}
      \notag\\
    &=
      \int_{y\in\cY} d\mu(y)\qty|\vphantom{A^{2}}f(\hat{q}_{\delta},y)\ket{\psi(y)}-f(\hat{q},y)\ket{\psi(y)}|^{2}
      \notag\\
    &\leq
      4\delta^{2}\int_{y\in\cY} d\mu(y)\bra{\psi(y)}\qty(\tilde f_{0}(y) + \tilde f_{2}(y))\tilde f_{2}(y)\qty(\hat\Omega +
      \overline{\omega^{2}})\ket{\psi(y)}
      \notag\\
    &= 4\delta^{2}
      \int_{y\in\cY}\bra{\psi}d\Pi_B(y)
      \qty(\tilde f_{0}(y) + \tilde f_{2}(y))\tilde f_{2}(y)\qty(\hat\Omega +
      \overline{\omega^{2}})\ket{\psi}
      \notag\\
    & = 4\delta^{2}\bra{\psi}\hat{Q}\ket{\psi}.
  \end{align}
\end{proof}

The following proposition provides one way to apply
Prop.~\ref{prop:condphase} to the case of conditional
displacements.  The proposition is formulated for an arbitrary self-adjoint
operator on system \(B\). We use the notation \(f^{(l)}(x)\) for the \(l\)'th
derivative of \(f(x)\).

\begin{proposition}
  \label{prop:conddisp}
  Let \(\hat{p}_{B}\) be a self-adjoint operator of \(B\), and consider conditional unitaries given by
  \(\hat{U}_{B}(x)=e^{iv(x)\hat{p}_{B}}\) for a real-valued
  function \(v(x)\). Let \(\xi\in\rls\) and define
  \(w(x)=v(x)-\xi x\). Suppose that \(w(x)\) is twice differentiable and the Fourier transform
  of the integrable function \(\tilde w(\kappa)\).
  In addition, assume that \(w_{l}=\int dx |w^{(l)}(x)|\) and
  \(\tilde w_{l}= \int d\kappa |\kappa|^{l}|\tilde w(\kappa)|\) are finite for \(l=0,1,2\). 
  Define
  \begin{align}
    \hat{P}_{B} &= \qty(\frac{1}{\sqrt{\pi}}\qty(\sqrt{w_{0}\tilde w_{0}}
                  +|\hat{p}_{B}|\sqrt{w_{1}\tilde w_{1}})+1),
                  \notag\\
    \hat{Q}_{B} &=
                  |\hat{p}_{B}|\qty(\tilde w_{2} + |\hat{p}_{B}|
                  \qty(\tilde w_{1}^{2} + 2\xi^{2})).
  \end{align}
  Then with \(\ket{\psi_{c,\delta}}\) as defined in Eq.~\eqref{eq:condstatedef},
  \begin{align}
        \qty| \ket{\psi_{c,\delta}}-\ket{\psi_{c,0}}|^{2}
    &\leq 4\delta^{2}
      \bra{\psi}      \qty(1+\hat{Q}_{B})\hat{Q}_{B}
      \hat{P}_{B}^{2}\qty(\hat{\Omega}+\overline{\omega^{2}})
      \ket{\psi}.
  \end{align}
\end{proposition}

\begin{proof}
  The constraints on \(w(x)\) and \(\tilde w(\kappa)\) are in terms of
  \(L^{1}\) norms of derivatives of \(w(x)\) and their Fourier
  transforms.  The \(L^{1}\) norms imply bounds on the sup-norm
  according to
  \begin{align}
    \|w^{(l)}(x)\|_{\infty} &\leq \tilde w_{l}.
  \end{align}
  These bounds imply the following bounds on the \(L^{2}\)-norms of
  \(w^{(l)}(x)\) and \(\kappa^{l}\tilde w(\kappa)\):
  \begin{align}
    2\pi\|\kappa^{l}\tilde w(\kappa)\|_{2}^{2}
    &= \|w^{(l)}(x)\|_{2}^{2}
      \notag\\
    &= \int_{x\in\rls}dx |w^{(l)}(x)|^{2}
      \notag\\
    &\leq \|w^{(l)}(x)\|_{\infty}\int_{x}dx |w^{(l)}(x)|
      \notag\\
    &\leq \tilde w_{l} w_{l}.
  \end{align}
  The factor of \(2\pi\) compensates
  for our use of non-unitary Fourier transforms.
  
  Let \(d\Pi_{B}(y)\) be the projection-valued measure of
  \(\hat{p}_{B}\), so that
  \(d\Pi_{B}(y)\hat{U}_{B}(x) = e^{iv(x)y}d\Pi_{B}(y)\). If we set
  \(f(x,y)=e^{iv(x)y}\) in Prop.~\ref{prop:condphase}, then
  \(\hat{f}_{\delta}=\hat{C}_{\delta}\).
  Let
  \(k_{y}(x)=f(x,y)\).  To apply the proposition, for each \(y\) we need
  \(k_{y}(x)\) to be the Fourier transform of a bounded complex
  measure \(d\tilde\nu_{y}(\kappa)=e^{i\theta(\kappa)}d\nu_{y}(\kappa)\),
  and it is necessary to determine bounds on
  \(\tilde f_{l}(y) = \int_{\kappa\in\rls}d\nu_{y}(\kappa)|\kappa^{l}|\) for
  \(l=0, 2\).

  We have that \(k_{y}(x)\) is the Fourier transform of
  \(d\tilde\nu_{y}(\kappa)\) iff
  \(g_{y}(x) = k_{y}(x)e^{-i\xi xy}= e^{iw(x)y}\) is the Fourier
  transform of the displaced measure \(d\tilde\nu_{y}(\kappa +\xi y)\).
  The needed bounds can be determined from those for the displaced measure.
  We have
  \begin{align}
    \int_{\kappa\in\rls} d\nu_{y}(\kappa)
    &= \int_{\kappa\in\rls}d\nu_{y}(\kappa + \xi y),
      \notag\\
    \int_{\kappa\in\rls} d\nu_{y}(\kappa)|\kappa|^{2}
    &=
      \int_{\kappa\in\rls} d\nu_{y}(\kappa+\xi y) |\kappa+\xi y|^{2}
      \notag\\
    &\leq       \int_{\kappa\in\rls} d\nu_{y}(\kappa+\xi y) 2\qty(|\kappa|^{2}+|\xi y|^{2})
      \notag\\
    &= 2\int_{\kappa\in\rls} d\nu_{y}(\kappa+\xi y) |\kappa|^{2}
      + 2|\xi y|^{2}\int_{\kappa\in\rls} d\nu_{y}(\kappa).
      \label{eq:disp1}
  \end{align}
  With these inequalities, we can obtain bounds for \(k_{y}(x)\) from
  bounds for \(g_{y}(x)\).

  Define \(h_{y}(x)=g_{y}(x)-1\). Then \(g_{y}(x)\) is the Fourier
  transform of a bounded complex measure iff \(h_{y}(x)\) is such a
  Fourier transform, and the two measures differ by the point measure
  \(\delta(\kappa)\) at \(0\).  We first show that \(h_{y}(x)\) is the
  Fourier transform of a function \(\tilde h_{y}(\kappa)\) in
  \(L^{1}(\rls)\) and obtain a bound on
  \(\|\tilde h_{y}(\kappa)\|_{1}\).  We claim that \(\tilde h_{y}(\kappa)\)
  is a well-defined function in \(L^{2}(\rls)\).  This
  follows from
  \begin{align}
    |h_{y}(x)|^{2}
    &= 2(1-\Re\qty(e^{iw(x)y}))
      \notag\\
    &\leq y^{2}w(x)^{2},
  \end{align}
  so that \(\| h_{y}(x)\|_{2} \leq |y|\sqrt{w_{0}\tilde w_{0}}\).
  Therefore we can write
  \(h_{y}(x)=\int_{\kappa\in\rls} d\kappa e^{i\kappa x}\tilde h_{y}(\kappa)\) for a
  function \(\tilde h_{y}(\kappa)\in L^{2}(\rls)\). To obtain a bound
  on this function's \(L^{1}\)-norm, we use second differentiability
  of \(h_{y}(x)\), where the derivatives are in
  \(L^{2}(\rls)\). The derivatives are given by
  \begin{align}
    h^{(1)}_{y}(x) &= iyw^{(1)}(x)e^{iw(x)y},
                \notag\\
    h^{(2)}_{y}(x) &= \qty(iyw^{(2)}(x)- y^{2}w^{(1)}(x)^{2})e^{iw(x)y}.
  \end{align}
  Consequently, the \(L^{2}\) norms of the derivatives are
  \(\|h^{(1)}_{y}(x)\|_{2} \leq |y| \sqrt{w_{1}\tilde w_{1}}\) and
  \(\|h^{(2)}_{y}(x)\|_{2}\leq |y| \sqrt{w_{2}\tilde w_{2}}+
  y^{2}\tilde w_{1}\sqrt{w_{1}\tilde w_{1}}\).  For the latter we used
  \(\|w^{(1)}(x)^{2}\|_{2}\leq\|w^{(1)}(x)\|_{\infty}\|w^{(1)}(x)\|_{2}\).
  The derivatives \(h^{(1)}_{y}(x)\) and \(h^{(2)}_{y}(x)\) are the Fourier
  transforms of the \(L^{2}(\rls)\) functions
  \(i\kappa \tilde h_{y}(\kappa)\) and
  \(-\kappa^{2}\tilde h_{y}(\kappa)\).  We apply the Cauchy-Schwarz
  inequality to bound the \(L^{1}\) norm of the Fourier transform of
  \(h_{y}(x)\):
  \begin{align}
    \int_{\kappa\in\rls} d\kappa |\tilde h_{y}(\kappa)|
    & =  \int_{\kappa\in\rls} d\kappa \frac{1}{1+|\kappa|} (1+|\kappa|)|\tilde h_{y}(\kappa)|
      \notag\\
    &\leq \left\|\frac{1}{1+|\kappa|}\right\|_{2}
      \left\|(1+|\kappa|)|\tilde h_{y}(\kappa)|\right\|_{2}
      \notag\\
    &\leq\frac{1}{\sqrt{2\pi}}\left\|\frac{1}{1+|\kappa|}\right\|_{2}
      \qty(\left\|h_{y}(x)\|_{2}
      +
      \|h^{(1)}_{y}(x)\right\|_{2})
      \notag\\
    &\leq \frac{1}{\sqrt{\pi}}(\sqrt{w_{0}\tilde w_{0}}+|y|\sqrt{w_{1}\tilde w_{1}}),
  \end{align}
  where we used the identity
  \(\left\|\frac{1}{1+|\kappa|}\right\|_{2}=\int_{-\infty}^{\infty}d\kappa/(1+|\kappa|)^{2} =
  2\int_{0}^{\infty}d\kappa/(1+\kappa)^{2}=-2
  /(1+\kappa)|_{0}^{\infty}=2\). 

  We also need a bound on the \(L^{1}\)-norm of the inverse Fourier
  transform \(-\kappa^{2}\tilde h_{y}(\kappa)\) of \(h^{(2)}_{y}(x)\).
  Since \(e^{iw(x)y}= h_{y}(x)+1\), we have
  \begin{align}
    h^{(2)}_{y}(x) &= \qty(iyw^{(2)}(x)- y^{2}w^{(1)}(x)^{2})
                     (h_{y}(x)+1)
                     \label{eq:bndh2y}
  \end{align}
  We apply the facts that \(L^{1}(\rls)\) with convolution is a Banach
  algebra, and the Fourier transform of a convolution is the product
  of the Fourier transforms.  Consequently, the \(L^{1}\) norm of the
  inverse Fourier transform of \(w^{(1)}(x)^{2}\) is bounded by
  \(\tilde w_{1}^{2}\), from which it follows that the \(L^{1}\)-norm
  of the inverse Fourier transform of the factor
  \(iyw^{(2)}(x)- y^{2}w^{(1)}(x)^{2}\) on the right-hand side of
  Eq.~\eqref{eq:bndh2y} is bounded by
  \(|y|\tilde w_{2}+|y|^{2} \tilde w_{1}^{2}\).  Distributing the
  product over \(h_{y}(x)+1\) and applying the Banach property again
  to the first summand gives
  \begin{align}
    \int_{\kappa\in\rls}d\kappa |\kappa|^{2}|\tilde h_{y}(\kappa)|
    &\leq
      |y|(\tilde w_{2}+  |y|\tilde w_{1}^{2})\qty(
      \frac{1}{\sqrt{\pi}}\qty(\sqrt{w_{0}\tilde w_{0}}+|y|\sqrt{w_{1}\tilde w_{1}})+1).
  \end{align}
  Since \(d\tilde \nu_{y}(\kappa+\xi y) = d\kappa \tilde h_{y}(\kappa)+\delta(\kappa)\),
  the inequalities of Eq.~\eqref{eq:disp1} can be continued for
  \begin{align}
    \int_{\kappa\in\rls} d\nu_{y}(\kappa)
    &\leq
      \int_{\kappa\in\rls}d\kappa |\tilde h_{y}(\kappa)| + 1
      \notag\\
    &\leq \frac{1}{\sqrt{\pi}}\qty(\sqrt{w_{0}\tilde w_{0}}+|y|\sqrt{w_{1}\tilde w_{1}})+1
      ,\notag\\
    \int_{\kappa\in\rls} d\nu_{y}(\kappa)|\kappa|^{2}
    &\leq
      \int_{\kappa\in\rls}d\kappa |\kappa|^{2}\tilde h_{y}(\kappa)
      + 2|\xi y|^{2}\qty(\frac{1}{\sqrt{\pi}}\qty(\sqrt{w_{0}\tilde w_{0}}+|y|\sqrt{w_{1}\tilde w_{1}})+1)
      \notag\\
    &\leq
      |y|\qty(\frac{1}{\sqrt{\pi}}\qty(\sqrt{w_{0}\tilde w_{0}}+|y|\sqrt{w_{1}\tilde w_{1}})+1)
      \qty(\tilde w_{2}+  |y|\qty(\tilde w_{1}^{2} + 2\xi^{2}))
      .
      \label{eq:disp2}
  \end{align}
  The proposition follows by substitution into the conclusion of
  Prop.~\ref{prop:condphase}.
\end{proof}

Applying Prop.~\ref{prop:conddisp} requires that \(v(x)\) is
sufficiently well behaved. This may require regularizing \(v(x)\). For
this it helps to first compare the intended \(v(x)\) to a regularized
one by applying the following lemma.

\begin{lemma}
  \label{lem:condreg}
  Let \(\hat{p}_{B}\) be a self-adjoint operator of \(B\) with projection valued
  measure \(d\Pi_{B}(x)\), and
  \(d\Pi(x)\) be a projection-valued measure on the modes \(\bm{a}\) with \(x\in\rls\). Let \(\hat{x}=\int_{x\in\rls}d\Pi(x)x\). For real-valued functions \(u(x)\) and \(v(x)\), consider the two families of unitaries \(\hat{U}_{B}(x)=e^{iu(x)\hat{p}_{B}}\) and
  \(\hat{V}_{B}(x)=e^{iv(x)\hat{p}_{B}}\) defining conditional
  unitaries \(\cC_{\bm{a}B}=\int_{x\in\rls} d\Pi(x) \hat{U}_{B}(x)\) and
  \(\cD_{\bm{a}B}=\int_{x\in\rls} d\Pi(x) \hat{V}_{B}(x)\). Then
  \begin{align}
    \qty|\cC_{\bm{a}B}\ket{\psi}-\cD_{\bm{a}B}\ket{\psi}|^{2}
    &\leq \bra{\psi}\hat{p}_{B}^{2} (u(\hat{x})-v(\hat{x}))^{2}\ket{\psi}.
  \end{align}
\end{lemma}

\begin{proof}
  It suffices to show that
    \begin{align}
    \Re\qty(\bra{\psi}\cD_{\bm{a}B}^{\dagger}\cC_{\bm{a}B}\ket{\psi})
    &\geq
      \one- \frac{1}{2}\bra{\psi}\hat{p}_{B}^{2} (u(\hat{x})-v(\hat{x}))^{2}\ket{\psi}.
  \end{align}
  We have 
  \begin{align}
    \Re\qty(\cD_{\bm{a}B}^{\dagger}\cC_{\bm{a}B})
    &=\int_{x\in\rls,y\in\rls} d\Pi(x) d\Pi_{B}(y)
      \Re\qty(e^{i(u(x)-v(x))y})
      \notag\\
    &\geq
      \int_{x\in\rls,y\in\rls} d\Pi(x) d\Pi_{B}(y)
      \qty(1- \frac{1}{2}(u(x)-v(x))^{2}y^{2})
      \notag\\
     &= 1-\frac{1}{2}\hat{p}_{B}^{2}(u(\hat{x})-v(\hat{x}))^{2},
  \end{align}
  where we applied Eq.~\eqref{eq:reeith-ineq-op} for the second
    line.  The inequality of the lemma follows. 
\end{proof}

Next, we bound the distance between \(\ket{\psi_{m,\delta}}\) and
\(\ket{\tilde\psi_{m,\delta}}\). We first give a proposition that
provides a general strategy for obtaining lower bounds on
\(\Re\qty(\braket{\tilde\psi_{m,\delta}}{\psi_{m,\delta}})\). For this proposition, the self-adjoint operator \(\hat{q}_{\delta}\) need
not be a quadrature of the modes \(\bm{a}\), and the initial state need not be vacuum on the LO modes. More generally, let
\(d\Pi_{A}(x)\) be a projection-valued measure of system \(A\) with
\(x\in\cB\).  Let \(\hat{U}(x)_{B}\) a family of unitary operators
defining the conditional unitary
\(\hat{C}_{AB}=\int_{x\in\cB}d\Pi_{A}(x)\hat{U}(x)_{B}\).  Let the
apparatus \(M\) have a projection-valued measure \(d\Pi_{M}(z)\) for
\(z\in\cB\), and let \(\hat{M}_{AM}\) be a measurement unitary of the
form \(\hat{M}_{AM} = \int_{x\in\cB}d\Pi_{A}(x)\hat{V}(x)_{M}\) for a
family of unitary operators \(\hat{V}(x)_{M}\) on the apparatus.  The
initial state of the apparatus \(M\) is a fixed pure state
\(\ket{g}_{M}\). Write \(\ket{g_{x}}_{M}=\hat{V}(x)_{M}\ket{g}\).  The
states \(\ket{g_{x}}_{M}\) have a direct-integral representation
\(\int_{z\in\cB}d\Pi_{M}(z)\ket{g_{x}} = \int_{z\in\cB}^{\oplus}d\mu(z)\ket{g_{x}(z)}_{M}\) with
respect to a \(\sigma\)-finite measure \(d\mu(z)\) on \(\cB\).  Then
\begin{align}
  \hat{M}_{AM}\ket{\psi}_{AB}\ket{g}_{M}
  &=
    \int_{x\in\cB}d\Pi_{A}(x)\ket{\psi}_{AB}\ket{g_{x}}_{M}
    \notag\\
  &=
    \int_{x\in\cB,z\in\cB}d\Pi_{A}(x)d\Pi_{M}(z)\ket{\psi}_{AB}\ket{g_{x}}_{M}
    \notag\\
  &=
    \int_{x\in\cB}\int_{z\in\cB}^{\oplus}d\Pi_{A}(x)d\mu(z)\ket{\psi}_{AB}\ket{g_{x}(z)}_{M}
   .
\end{align}
For each \(x\), define the probability distribution
\(d\mu(z)h(z|x)=d\mu(z)\braket{g_{x}(z)}\).
We define the apparatus conditional unitary \(\hat{C}_{MB}\)
as \(\hat{C}_{MB} = \int_{z\in\cB}d\Pi_{M}(z) \hat{U}(z)_{B}\).  Let
\(\hat{W}_{AM}=\int_{zx}d\Pi_{M}(z)d\Pi_{A}(x) e^{i\omega(z,x)}\) for
a real-valued function \(\omega(z,x)\).  Define
\begin{align}
  \ket{\psi_{m}} &= \hat{M}_{AM}\hat{C}_{AB}\ket{\psi}_{AB}\ket{g}_{M}
                   \notag\\
            &=\hat{C}_{AB}\hat{M}_{AM}\ket{\psi}_{AB}\ket{g}_{M},
                   \notag\\
  \ket{\tilde\psi_{m}}
                 &= \hat{W}_{AM}\hat{C}_{MB}\hat{M}_{AM}\ket{\psi}_{AB}\ket{g}_{M}.
\end{align}

\begin{proposition}\label{prop:conditionalU}
With the definitions introduced in the previous paragraph, consider
a positive semi-definite self-adjoint operator $\hat{Q}_{B}$ and a non-negative
function $f(z,x)$ on \(\cB^{2}\) such that for all $z$ and $x$
\begin{align}
f(z,x)\hat{Q}_{B} \geq    
\Re\qty(1-e^{-iw(z,x)}\hat{U}^{\dagger}(z)_{B}\hat{U}(x)_{B}).
\label{eq:prop:conditionalU:1}
\end{align}
Define  \( \hat{F}_{A}=\int d\Pi_{A}(x) f(z,x)h(z|x)d\mu_{M}(z)\). Let $\rho_{A}$ and $\rho_{B}$ be the reduced density matrices on $A$ and $B$ of the state $\ketbra{\psi}$. Then
\begin{align}
  \qty|\ket{\smash{\tilde\psi_{m}}}-\ket{\psi_{m}}|^{2}
  &\leq 2\bra{\psi}_{AB}\hat{F}_{A}\hat{Q}_{B}\ket{\psi}_{AB}
    \label{eq:prop:conditionalU:2}\\
  &\leq
    2\tr(\rho_{A} \hat{F}_{A}^{2})^{1/2} \tr(\rho_{B} \hat{Q}_{B}^{2})^{1/2}.
    \label{eq:prop:conditionalU:3}
\end{align}
\end{proposition}

\begin{proof}
  The  operator \(\hat{Q}_{B}\)
  need not be bounded. Let \(\hat{R}\) be the operator on the
  right-hand side of the inequality in
  Eq.~\eqref{eq:prop:conditionalU:1}. Since \(\hat{R}\) is bounded,
  \(\hat{Q}_{B}-\hat{R}\) is selfadjoint with the same domain as
  \(\hat{Q}_{B}\). The inequality is equivalent to the statement
  that \(\hat{Q}_{B}-\hat{R}\) is positive semi-definite.

  The expression on the right-hand side of
  Eq.~\eqref{eq:prop:conditionalU:3} is \(\infty\) when \(\ket{\psi}\)
  is not in the domains of \(\hat{Q}_{B}\) or \(\hat{F}_{A}\), in
  which case there is nothing to prove.  Assume that \(\ket{\psi}\) is
  in the domains of these operators.  We have
  \begin{align}
    \ket{\psi_{m}}
    & = \hat{\cC}_{AB}\int_{x} d\Pi_{A}({x})
      \ket{\psi}_{AB} \ket{g_{{x}}}_{M}
      \notag\\
    &=  \int_{x\in\cB} d\Pi_{A}({x})\hat{U}_{B}({x})\ket{\psi}_{AB}\ket{g_{{x}}}_{M},
      \notag\\
    \ket{\tilde\psi_{m}}
    &=
      \hat{W}_{AM} \hat{\cC}_{MB}\int_{x\in\cB} d\Pi_{A}({x})
      \ket{\psi}_{AB}\ket{g_{{x}}}_{M}
      \notag\\
    & = \int_{z\in\cB,x\in\cB} d\Pi_{M}({z})d\Pi_{A}({x})e^{iw({z},{x})} \hat{U}_{B}({z})\ket{\psi}_{AB}\ket{g_{{x}}}_{M}
    .
  \end{align}
  Therefore
  \begin{align}
    \Re\qty(\braket{\tilde\psi_{m}}{\psi_{m}})
    &=
      \int_{x\in\cB,z\in\cB} \bra{\psi}_{AB}d\Pi_{A}({x})e^{-iw({z},{x})} \hat{U}_{B}({z})^{\dagger}
      \hat{U}_{B}({x})\ket{\psi}_{AB}
      \bra{g_{{x}}}_{M}d\Pi_{M}({z})\ket{g_{{x}}}
      \notag\\
    &=
      \int_{z\in\cB, x\in\cB}d\mu_{M}({z})  \bra{\psi}d\Pi_{A}({x})e^{-iw({z},{x})}\hat{U}_{B}({z})^{\dagger}
      \hat{U}_{B}({x})\ket{\psi}
      h(z|x).
  \end{align}
  From this we get
  \begin{align}
    1 - \Re\qty(\bra{\tilde\psi_{m}} \ket{\psi_{m}})
    &=
      \Re\Bigg(\int_{z\in\cB,x\in\cB}d\mu(z)\bra{\psi}_{AB} d\Pi_{A}({x})\ket{\psi}_{AB} h({z}|{x})
      \notag\\
    &\hphantom{{}=\Re\Bigg(}
      {} -
      \int_{z\in\cB,x\in\cB} d\mu({z})\bra{\psi}_{AB}d\Pi_{A}({x})e^{-iw(z,x)} \hat{U}_{B}({z})^{\dagger}
      \hat{U}_{B}({x})\ket{\psi}_{AB}
      h({z}|{x})\Bigg)
      \notag\\
    &=\Re\qty(\int_{z\in\cB,x\in\cB} d\mu({z})\bra{\psi}_{AB} d\Pi_{A}({x}) \qty(1-e^{-iw(z,x)} \hat{U}_{B}({z})^{\dagger}
      \hat{U}_{B}({x}))\ket{\psi}_{AB} h({z}|{x}))
      \notag\\
    &=\int_{z\in\cB,x\in\cB}d\mu({z}) \bra{\psi}_{AB} d\Pi_{A}({x}) \Re\qty(1-e^{-iw(z,x)} \hat{U}_{B}({z})^{\dagger}
      \hat{U}_{B}({x}))\ket{\psi}_{AB} h({z}|{x})
      \notag\\
    &\leq
      \int_{z\in\cB,x\in\cB} d\mu({z})\bra{\psi}_{AB} d\Pi_{A}({x})\hat{Q}_{B}\ket{\psi}_{AB}
      f({z},{x})h({z}|{x})
      \notag\\
    &=
      \bra{\psi}_{AB}\hat{Q}_{B}\int_{z\in\cB,x\in\cB}  d\Pi_{A}({x})
      f({z},{x})h({z}|{x})d\mu({z})\ket{\psi}_{AB}
      \notag\\
    &=\bra{\psi}_{AB}\hat{F}_{A}\hat{Q}_{B}\ket{\psi}_{AB}
      \notag\\
    & \leq \qty(\bra{\psi}_{AB} \hat{F}_{A}^{2} \ket{\psi}_{AB})^{1/2}
      \qty(\bra{\psi}_{AB} \hat{Q}_{B}^{2} \ket{\psi}_{AB})^{1/2}
      \notag\\
    &=
      \tr(\rho_{A} \hat{F}_{A}^{2})^{1/2} \tr(\rho_{B} \hat{Q}_{B}^{2})^{1/2}.
  \end{align}
  The last inequality is the Cauchy-Schwarz inequality. The
  inequalities in the proposition follow by expanding
  \(\qty|\ket{\smash{\tilde\psi_{m}}}-\ket{\psi_{m}}|^{2}\).
\end{proof}

For the case where the conditional unitaries are displacements, we
can obtain explicit expressions for the operators in Prop.~\ref{prop:conditionalU}.
For this purpose, we let \(\cB=\rls^{k}\) and denote members of \(\cB\)
by \(\bm{x}\) and \(\bm{z}\), representing vectors of real values.

\begin{proposition} \label{prop:conditionalD}
  In Prop.~\ref{prop:conditionalU}, let
  $\hat{U}_{B}(\bm{x}) = D_{\bm{\gamma}({\bm{x}})}$ be displacements
  acting on modes of \(B\). Then we can choose $w(\bm{z},\bm{x})$ such that
  the conclusions of Prop.~\ref{prop:conditionalU} hold with
  \begin{align}
    f(\bm{z},\bm{x}) &= \frac{1}{2}
                       |\bm{\gamma}(\bm{z})-\bm{\gamma}(\bm{x})|^{2}, \\
    \hat{Q}_{B} & = 4(\hat{n}_{\tot}+1/2), \\
  \end{align}
  where $\hat{n}_{\tot}$ is the total photon number operator on the
  modes involved in the displacements.
  Suppose that in addition, $\bm{\gamma}(\bm{x})$ satisfies that for all
  \(\bm{x},\bm{z}\) we have
  $|\bm{\gamma}(\bm{z})-\bm{\gamma}(\bm{x})| \leq K_{1} |\bm{z}-\bm{x}|$,
  and for all \(\bm{x}\),
  $\int d\bm{z}|\bm{x}-\bm{z}|^{2}h(\bm{z}|\bm{x}) \leq K_{2}$.  Then
  \begin{align}
    \qty|\ket{\smash{\tilde\psi_{m}}}-\ket{\psi_{m}}|^{2}
    \leq 2K_{1}^{2}K_{2}
    \qty(\bra{\psi}_{AB}(\hat{n}_{\tot}+1/2)^{2}\ket{\psi}_{AB})^{1/2}.
  \end{align}
\end{proposition}

\begin{proof}
  We have 
\begin{align}
\hat{U}_{B}(\bm{z})^{\dagger}\hat{U}_{B}(\bm{x}) & = e^{-i \Im (\bm{\gamma}(\bm{x})^{\dagger} \bm{\gamma}({\bm{z}}))} D_{\bm{\gamma}(\bm{x}) - \bm{\gamma}(\bm{z})}.
\end{align}
In view of this identity, we let
\begin{align}
w(\bm{z},\bm{x}) = -\Im(\bm{\gamma}(\bm{x})^{\dagger}\bm{\gamma}(\bm{z})).
\end{align}

We determine a bound on displacements. 
Consider a general displacement operator $D_{\bm{\beta}}$. This is unitarily equivalent to the displacement $D_{r}=e^{ir\hat{p}}$ with $\hat{p}$ a normalized quadrature and \(r=|\bm{\beta}|\). For an arbitrary state $\ket{\phi}$,
let \(d\mu(z) = \bra{\phi} d\Pi_{p}(z) \ket{\phi}\), where \(d\Pi_{p}(z)\)
is the spectral measure of \(\hat{p}\). Then
\begin{align}
\Re(\bra{\phi} D_{r}\ket{\phi})
&= \int d\mu(z)  \Re\qty(e^{irz}) \nonumber \\
& \geq 1 - \int d\mu(z)  \frac{r^{2} z^{2}}{2} \nonumber \\
& = 1 - \frac{r^{2}}{2}\bra{\phi} \hat{p}^{2}\ket{\phi}.
\end{align}
By arbitrariness of \(\ket{\phi}\), it follows that the operator inequality
\begin{align}
\frac{r^{2}}{2}\hat{p}^{2} &\geq \Re(1- D_{r})
\end{align}
holds. 
Since $\hat{p}^{2} \leq 4(\hat{n}+1/2)\leq 4(\hat{n}_{\tot}+1/2)$, we can replace $\hat{p}^{2}$ with $4(\hat{n}_{\tot}+1/2)$ in this inequality. 

Next, we apply the bound of the previous paragraph to verify
Eq.~\eqref{eq:prop:conditionalU:1} of Prop.~\ref{prop:conditionalU}
for \(\hat{Q}_{B}=4(\hat{n}_{\tot}+1/2)\) and
\(f(\bm{z},\bm{x})=\frac{1}{2}|\bm{\gamma}(\bm{z})-
\bm{\gamma}(\bm{x})|^{2} \leq
\frac{K_{1}^{2}}{2}|\bm{z}-\bm{x}|^{2}\).
\begin{align}
\Re\qty(1-e^{-iw(\bm{z},\bm{x})}\hat{U}_{B}^{\dagger}(\bm{z}) \hat{U}_{B}(\bm{x}))
& = \Re\qty(1- e^{-iw(\bm{z},\bm{x})} D_{-\bm{\gamma}(\bm{z})}
D_{\bm{\gamma}(\bm{x})}) \nonumber \\
& = \Re\qty(1- D_{\bm{\gamma}(\bm{x})-\bm{\gamma}(\bm{z})}) \nonumber \\
& \leq \frac{1}{2}|\bm{\gamma}(\bm{z})-\bm{\gamma}(\bm{x})|^{2}
                                                                           4(\hat{n}_{\tot}+1/2)\nonumber\\
  &\leq 2K_{1}^{2}|\bm{z}-\bm{x}|^{2}(\hat{n}_{\tot}+1/2).
\end{align}
The second-last line is the bound needed for the first part of the
proposition.  The last line is needed for the second part. To finish
establishing the second part, we bound $\hat{F}_{\bm{a}}$ by bounding
the term multiplying \(d\Pi_{\bm{a}}(\bm{x})\) in the definition of
\(\hat{F}_{\bm{a}}\) in Prop.~\ref{prop:conditionalU}:
\begin{align}
  \int d\mu(z) f(\bm{z},\bm{x})h(\bm{z}|\bm{x})
  &\leq \int d\mu(z)\frac{K_{1}^{2}}{2}|\bm{x}-\bm{z}|^{2}h(\bm{z}|\bm{x})
    \notag\\
  &\leq \frac{K_{1}^{2}K_{2}}{2}.
\end{align}
consequently, \(\hat{F}_{\bm{a}}^{2}\leq K_{1}^{4}K_{2}^{2}/4\).
Substituting in the inequality
of Eq.~\eqref{eq:prop:conditionalU:2} gives the bound in the proposition.
\end{proof}


\section{Applications to standard pulsed homodyne}
\label{sec:onemode}
Our bounds can be applied to standard pulsed homodyne,
which corresponds to having \(\omega_{k}=\omega\) for \(k=1,\ldots,N\).  With our conventions, \(\omega=\omega_{1}=1\). The
LO displacement is then
\(\frac{\bm{\beta}}{\delta}=-\frac{i\bm{\alpha}}{\delta}\).  Because the target quadrature is normalized, the expected number of photons in the LO after the displacement is \(1/\delta^{2}\). For the quantities introduced in Thm.~\ref{thm:mainbnd}, we get
\(\hat{\Omega}=\hat{n}_{\tot}\) and \(\overline{\omega^{2}} = 1\).

\newcommand{\lowdetnoise}{730} 
\newcommand{\highdetnoise}{8250} 

A standard implementation of pulsed homodyne uses a matched pair of
photodiodes as the detectors. A detailed description of such an
implementation with low added noise is given in
Ref.~\cite{hansen2001ultrasensitive}. The subtracted photodiode
measurement has an effective added noise of
\(\sigma\approx \lowdetnoise\) electrons before rescaling. Another
example of such an implementation is given in
Ref.~\cite{gerrits2011balanced} with an added noise of
\(\sigma\approx\highdetnoise\) electrons. Because the added noise
involves a large number of electrons, this can to good approximation
be modeled by Gaussian added noise in the measurement. In the von
Neumann measurement scheme as used in Thm.~\ref{thm:qmeas}, this
effect is obtained with the apparatus operator \(\hat{s}_{M}\) the
momentum operator of a one-dimensional system with Hilbert space
\(L^{2}(\rls,d\mu(s))\), and with initial real positive Gaussian wave
function \(\tilde g(x)\) in position space satisfying
\begin{align}
  \tilde g(x)^{2} &= \frac{1}{\sqrt{2\pi}\delta\,\sigma}e^{-x^{2}/
             (2(\delta\,\sigma)^{2})}.
             \label{eq:gaussian_added_noisex}
\end{align}
The operator \(\hat{s}_{M}\) generates translations of position space
so that \(e^{-i t\hat{s}_{M}}\) translates wavefunctions by a distance of \(t\).
The final measurement is in the position basis of the apparatus so
that the added Gaussian measurement noise has variance
\((\delta\,\sigma)^{2}\) for a resolution of \(r=\delta\,\sigma\).
This accounts for the rescaling of the
original homodyne difference operator by \(\delta\) and the added
noise of standard deviation \(\sigma\) after the subtraction of detector outcomes.

For applying Thm.~\ref{thm:qmeas}, we need the wavefunction \(g(s)\)
in momentum space \(L^{2}(\rls,d\mu(s))\).  This wave function is
given by the unitary Fourier transform of \(\tilde g(x)\), which is
\begin{align}
  g(s) &= \frac{\sqrt{2\delta\,\sigma}}{(2\pi)^{1/4}} e^{- (\delta\,\sigma)^{2}s^{2}}.
             \label{eq:gaussian_added_noise}
\end{align}
The square \(|g(s)|^{2}\) is a standard Gaussian with variance
\(1/(2\delta\,\sigma)^{2}=1/(2r)^{2}\).

We refer
to the measurement model with \(g(s)\) satisfying
Eq.~\eqref{eq:gaussian_added_noise} as Gaussian added noise at
resolution \(r\). The LO pulses used for the measurements in
Ref.~\cite{hansen2001ultrasensitive} have \(10^{6}\) to
\(3\times 10^{8}\) photons, which corresponds to
\(\delta\in [10^{-3},\sqrt{3}\times
10^{-4}]\). In~\cite{gerrits2011balanced}, the number of photons in
the LO pulses is \(0.5\times 10^{9}\) to \(2\times 10^{9}\).

\subsection{Fidelity of measurement outcomes and remaining state.}
For illustration, we first apply the general bounds for BBP homodyne
from Sect.~\ref{sec:sconv}. We then improve the bounds with the
results for standard pulsed homodyne from
App.~\ref{app:standardhomodyne}.  The bounds obtained show that while
the general bounds are practically relevant and can be used regardless
of the distribution of the \(\omega_{k}\) without complications from
higher-degree terms, the specific bounds yield significant
improvements.

Thm.~\ref{thm:qmeas} implies a bound on the fidelity between the states obtained after coupling to the apparatus according to the ideal target quadrature compared to coupling according to the pulsed homodyne observable. The bound depends on the apparatus initial state
through the quantities \(b_{g,2l}\) defined in the theorem. In view of the relationship between distance and fidelity in
App.~\ref{app:fidelities}, the bound is
\begin{align}
  F_{\delta} &\geq
               1-4\delta^{2}\bra{\psi}\qty(b_{g,2}\hat{n}_{\tot}
               + b_{g,4}\qty(\hat{n}_{\tot}+1))\ket{\psi},
               \label{eq:outstate_fid}
\end{align}
where \(\ket{\psi}\) is the purified initial state.  This is a lower bound on the classical-quantum fidelity  of the joint state of the unmeasured quantum modes and the classical measurement outcomes, which is the relevant fidelity if the classical outcome distribution matters. 
For the homodyne measurement with Gaussian added noise at resolution \(r\), the quantities
\(b_{g,2}\) and \(b_{g,4}\) can be computed to be
\(b_{{g},2} = 1/(2r)^{2}\) and \(b_{{g},4}= 3/(2r)^{4}\).
This gives the bound
\begin{align}
  F_{\delta} &\geq                1-4\frac{\delta^{2}}{(2r)^{4}}
               \bra{\psi}\qty(\qty((2r)^{2}+{3})\hat{n}_{\tot}
               + {3})\ket{\psi}.
               \label{eq:outstate_fid_gauss}
\end{align}
High fidelity requires that
\({\delta^{2}\langle\hat{n}_{\tot}\rangle}/{r^{4}}\) is small,
which corresponds to requiring that the average number of photons in the LO be larger than the average number of photons in the signal by a factor determined by the resolution. This refines and quantifies the conventional recommendation that the average number of photons in the
LO be much larger than the average number of photons in the signal modes~\cite{braunstein1990homodyne}.

If, as in pulsed homodyne measurements with a matched pair of
photodiodes, the added noise is approximately constant before
rescaling, the resolution improves with increasing LO amplitude
according to \(r=\delta\,\sigma\).  Higher resolution implies that the
measurement of \(\hat{q}\) becomes sensitive to finer details of the
quadrature probability distribution. At fixed high resolution, it is
necessary to have correspondingly high LO amplitude for the homodyne
measurement to preserve these fine details. But with an LO amplitude
dependent resolution of \(r=\delta\,\sigma\), this resolution improves
too rapidly with LO amplitude for the LO amplitude to
compensate. Expressed in terms of \(\delta\) and \(\sigma\), the
fidelity tradeoff between detector-added noise \(\sigma\) and LO
amplitude is expressed with our bounds by
\begin{align}
  F_{\delta} &\geq                1-\frac{4}{(2\sigma)^{4}}
               \bra{\psi}\qty((2\sigma)^{2}\hat{n}_{\tot}+{3\delta^{-2}}(\hat{n}_{\tot}
               + 1))\ket{\psi}.
\end{align}
In practice, applications require a particular resolution \(r_{0}\).
At a given LO amplitude \(1/\delta\), if \(\delta\,\sigma < r_{0}\),
one can artificially add noise to the homodyne measurement outcomes to
increase \(\sigma\) to \(\sigma'=r_{0}/\delta\) and satisfy
\(\delta\,\sigma=r_{0}\).  The lower bound on fidelity in terms of
resolution then applies, and increasing the LO amplitude improves the
fidelity according to Eq.~\eqref{eq:outstate_fid_gauss} with \(r\)
replaced by \(r_{0}\).
\edef\sphfidexdelta{\fpeval{10^-4}}%
\edef\sphfidexr{\fpeval{\sphfidexdelta * \lowdetnoise}}%
Consider the
values for photodiode added noise and LO average photon numbers from
Ref.~\cite{hansen2001ultrasensitive} and assume an upper bound for
expected photon number of \newcommand{\sphfidexn}{5.0}
\(\langle \hat{n}_{\tot}\rangle =
\num[round-mode=places,round-precision=1]{\sphfidexn}\).  For
\(\delta=\num[exponent-mode=scientific,round-mode=figures,round-precision=1,
print-unity-mantissa=false]{\sphfidexdelta}\) corresponding to an LO
expected photon number of
\(\num[exponent-mode=scientific,round-mode=figures,round-precision=1,
print-unity-mantissa=false]{\fpeval{\sphfidexdelta^(-2)}}\), we get a
respectable resolution and fidelity of
\edef\sphfidex{\fpeval{%
    1-4*(\sphfidexdelta^2/((2*\sphfidexr)^4))* (((2*\sphfidexr)^2+%
    3)*\sphfidexn + 3) }}%
\begin{align}
  r&=\num[round-mode=places,round-precision=3]{\sphfidexr},
     \notag\\
  F_{\num[exponent-mode=scientific,round-mode=figures,round-precision=1, print-unity-mantissa=false]{\sphfidexdelta}} & \geq \num[round-mode=places,round-precision=3]{\sphfidex}.
\end{align}

Next we update the bounds obtained according to the results of App.~\ref{app:standardhomodyne}. With Thm.~\ref{thm:qmeas_sph},
Eq.~\eqref{eq:outstate_fid} becomes
\begin{align}
  F_{\delta} &\geq
               1-\delta^{2}\bra{\psi}\qty(2 b_{g,2}\hat{n}_{\tot}
               + \frac{1}{9}\delta^{2}b_{g,6}\hat{q}^{2} + \frac{1}{2} b_{g,4})\ket{\psi}.
               \label{eq:outstate_fid_sph}
\end{align}
We have \(b_{g,6}=15/(2r)^{6}\) and, conservatively, \(\bra{\psi}\hat{q}^{2}\ket{\psi} \leq 4\hat{n}_{\tot}+2\). Therefore
\begin{align}
  F_{\delta}
  &\geq
    1-\delta^{2}\bra{\psi}\qty(\qty(\frac{2}{(2r)^{2}} +\frac{20
    \delta^{2}}{3 (2r)^{6}})\hat{n}_{\tot}
    + \frac{10\delta^{2}}{3(2r)^{6}}+\frac{3}{2 (2r)^{4}})\ket{\psi}
    \notag\\
  &=    1-\frac{\delta^{2}}{(2r)^{4}}\bra{\psi}\qty(\qty(2(2r)^{2} +\frac{20
    \delta^{2}}{3 (2r)^{2}})\hat{n}_{\tot}
    + \frac{10\delta^{2}}{3(2r)^{2}}+\frac{3}{2})\ket{\psi}.
    \label{eq:outstate_fid_gauss_sph}
\end{align}
\edef\sphfidexx{\fpeval{
    1-(\sphfidexdelta^2/((2*\sphfidexr)^4))*
    (
      (2*(2*\sphfidexr)^2+ (20*\sphfidexdelta^2)/(3*(2*\sphfidexr)^2))*\sphfidexn
      + (10*\sphfidexdelta^2)/(3*(2*\sphfidexr)^2) + 3/2
    )
  }}
\edef\sphfideydelta{\fpeval{10^(-3)}}
\edef\sphfidexy{\fpeval{
    1-(\sphfideydelta^2/((2*\sphfidexr)^4))*
    (
      (2*(2*\sphfidexr)^2+ (20*\sphfideydelta^2)/(3*(2*\sphfidexr)^2))*\sphfidexn
      + (10*\sphfideydelta^2)/(3*(2*\sphfidexr)^2) + 3/2
    )
  }}
With the same parameters as before we get \(F_{\num[exponent-mode=scientific,round-mode=figures,round-precision=1, print-unity-mantissa=false]{\sphfidexdelta}} \geq \num[round-mode=places,round-precision=5]{\sphfidexx}\),
which is significantly better than what we obtained with the simplified
general bounds for BBP homodyne. The expressions derived in
App.~\ref{app:appendix3} can in principle be used to obtain improved
general bounds for the relevant regimes.

\subsection{Estimating values of the characteristic function.}
Next we consider measuring the characteristic function of the Wigner
distribution of a mode at a given point. The characteristic function
of the Wigner function of the mode of quadrature \(\hat{q}\) is
\(\chi(\gamma)=\langle D_{\gamma\bm{\alpha}}\rangle\) for complex
\(\gamma\). For \(\gamma=0\), \(\chi(\gamma)=1\).  If it were possible
to implement an ideal quadrature measurement, we would estimate
\(\chi(\gamma)\) by repeatedly measuring
\(\hat{q}_{(\gamma/|\gamma|)\bm{\alpha}}\) and averaging
\(e^{-i|\gamma|\lambda}\) over the observed eigenvalues \(\lambda\).
Instead, we perform this average with the pulsed homodyne measurement
outcomes.  To determine the difference between \(\chi(\gamma)\) and
the computed average with pulsed homodyne measurements in the limit of
a large number of measurements, we apply Thm.~\ref{thm:mainbnd}
Eq.~\eqref{thm:eq:mains}. In this equation, replace \(s\) and
\(\bm{\alpha}\) there by \(|\gamma|\) and
\((\gamma/|\gamma|)\bm{\alpha}\) here. Because the pulsed homodyne
observable \(\hat{q}_{\delta}\) for \(\delta>0\) has discrete
spectrum, it can in principle be implemented directly without added
noise.  Unbiased added noise does not affect the error in the mean. In
particular, for determining the difference between the expectations of
the target observable and the measured observable, we do not have to
account for the added noise from, for example, subtracted photodiode
measurements. Here, we do not consider the reduction in
signal-to-noise from added noise.  The difference between
\(\chi(\gamma)\) and the corresponding pulsed homodyne average is
\begin{align}
  |\langle e^{-i\hat{q}_{\delta}|\gamma|}\rangle- \chi(\gamma)|
  &=
    |\langle e^{-i\hat{q}_{\delta}|\gamma|}\rangle - \langle e^{-i\hat{q}|\gamma|}\rangle|^{2}
    \notag\\
  &\leq
    |\qty(e^{-i\hat{q}_{\delta}|\gamma|} - e^{-i\hat{q}|\gamma|})\ket{\psi}|^{2}
    \notag\\
  &\leq
    4\delta^{2} |\gamma|^{2}
    \qty((1+|\gamma|^{2})\langle \hat{n}_{\tot}\rangle + |\gamma|^{2})
    ,
\end{align}
where we applied Eq.~\eqref{app:eq:a-bexp}.  Thus, to ensure a good estimate of the characteristic function at \(\gamma\) it suffices to choose \(\delta\) so that
\(\delta^{2}|\gamma|^{4}(\langle \hat{n}_{\tot}\rangle+1)\) and
\(\delta^{2}|\gamma|^{2}\langle \hat{n}_{\tot}\rangle\) are small. 
As a
function of \(\gamma\), the LO photon number required by this inequality scales as \(|\gamma|^{4}\).
\newcommand{\chfgamma}{20}
For \(\delta=\num[exponent-mode=scientific,round-mode=figures,round-precision=1, print-unity-mantissa=false]{\sphfidexdelta}\)
and \(\langle\hat{n}_{\tot}\rangle\leq \num[round-mode=places,round-precision=1]{\sphfidexn}\), we can measure the characteristic
function with error less than
\(
\num[round-mode=places,round-precision=3]{
  \fpeval{4*\sphfidexdelta^2*\chfgamma^2*((1+\chfgamma^2)*\sphfidexn + \chfgamma^2)}
  }
\)
for
\(|\gamma|\leq \num[round-mode=places,round-precision=1]{\chfgamma}\),
not accounting
for added noise in the measurement.

If we use the bound from Eq.~\eqref{thm:eq:mains_sph} for standard pulsed
homodyne, we get
\begin{align}
  |\langle e^{-i\hat{q}_{\delta}|\gamma|}\rangle- \chi(\gamma)|
  &\leq
    \delta^{2} |\gamma|^{2}
    \qty(\qty(2+\frac{4}{9}\delta^{2}|\gamma|^{4})\langle \hat{n}_{\tot}\rangle +
    \frac{2}{9}\delta^{2}|\gamma|^{4}+ \frac{1}{2}|\gamma|^{2}).
\end{align}
For the same bound of \num[round-mode=places,round-precision=1]{\chfgamma} on \(|\gamma|\), we get a significantly smaller error bound of
\(
\num[round-mode=places,round-precision=4]{
  \fpeval{\sphfidexdelta^2*\chfgamma^2*((2+(4/9)*\sphfidexdelta^{2}\chfgamma^4)*\sphfidexn
    + (2/9)*\sphfidexdelta^{2}\chfgamma^4+ (1/2)*\chfgamma^{2})}
}
\). This is likely much lower than the signal-to-noise in the measurement.
\edef\chfgammah{40}
For \(|\gamma|\leq\num[round-mode=places,round-precision=1]{\chfgammah}\),
the error bound is
\(\num[round-mode=places,round-precision=4]{
  \fpeval{
    \sphfidexdelta^2*\chfgammah^2*
    ((2+(4/9)*\sphfidexdelta^{2}\chfgammah^4)*\sphfidexn
    + (2/9)*\sphfidexdelta^{2}\chfgammah^4+ (1/2)*\chfgammah^{2})}
}
\).

\subsection{Estimating moments.}
We can also perform the above exercise for measuring the \(k\)'th
moment of the quadrature. The homodyne measurement has the correct
expectation for \(k=1\). For \(k=2\), the difference between the
expectations can be determined directly and is given by
\(\delta^{2}\langle \hat{n}_{\tot}\rangle\), see for example
Ref.~\cite{bbp1}.  Here we consider even \(k\) with \(k\geq 4\), so
that for real \(x\), \(|x|^{k}=x^{k}\).  Since the function \(f_{0}(x)=x^{k}\) is not
bounded, it is necessary to regularize it by applying
Prop.~\ref{prop:regfexp}. There are many options for
\(f(x)\). Consider a \(k\)-times differentiable function \(g(x)\) such
that \(\frac{d}{dx^{k}}g(x)\big|_{x=0} = k!\) and \(g(x)\) is the
Fourier transform of \(\tilde g(\kappa)\), where \(\tilde g(\kappa)\)
is the sum of an integrable function and a multiple of the delta
function at \(0\). Equivalently, \(d\kappa\,\tilde g(\kappa)\) is the
sum of a point measure at \(0\) and a bounded complex measure that is
absolutely continuous with respect to Lebesgue measure.  Then
\(f_{\lambda}(x)=g(\lambda x)/\lambda^{k}\) converges to \(x^{k}\)
uniformly on bounded intervals as \(\lambda\) goes to \(0\). Let
\(\Delta(x) = x^{k}-g(x)\). Then
\begin{align}
f_{0}(x)-f_{\lambda}(x)&= \frac{1}{\lambda^{k}}\Delta(\lambda x),
                      \notag\\
  d\kappa\tilde f_{\lambda}(\kappa)
  &= d\kappa\frac{1}{\lambda^{k+1}}\tilde g(\kappa/\lambda),
    \notag\\
  \int_{\kappa}d\kappa |\kappa|^{l}|\tilde f_{\lambda}(\kappa)|
  &=\frac{1}{\lambda^{k+1}}\int_{\kappa}d\kappa |\kappa|^{l}|\tilde g(\kappa/\lambda)|
    \notag\\
  &=\frac{1}{\lambda^{k-l}}\int_{\kappa}d\kappa |\kappa|^{l}|\tilde g(\kappa)|.
\end{align}
To determine the constants needed to apply Prop.~\ref{prop:regfexp} it
suffices to determine bounds
\(\tilde g_{l}\geq \int_{\kappa}d\kappa |\kappa|^{l}|\tilde g(\kappa)|\)
for \(l=0,2\). With these bounds, applying the proposition gives
\begin{align}
  \qty|\langle f(\hat{q}_{\delta})\rangle - \langle\hat{q}^{k}\rangle|^{2}
  &\leq
    \frac{8}{\lambda^{2(k-1)}}
    \delta^{2}(\tilde g_{0}+\lambda^{2}\tilde g_{2})
    \tilde g_{2}(\langle\hat{n}_{\tot}\rangle+1)
    + \frac{2}{\lambda^{2k}}\langle \Delta(\lambda \hat{q})^{2}\rangle
    .
    \label{eq:momenterror1}
\end{align}
It remains to choose \(g(x)\) and \(\lambda\).    Consider
\(g(x)=x^{k}/(1+x^{k})\). Then, because \(k\) is even,
\(\Delta(x)=x^{2k}/(1+x^{k}) \leq x^{2k}\). Therefore, the second summand in
the bound of Eq.~\eqref{eq:momenterror1} is bounded by
\begin{align}
  \frac{2}{\lambda^{2k}}\langle \Delta(\lambda \hat{q})^{2}\rangle
  &\leq
    2\lambda^{2k}\langle \hat{q}^{4k}\rangle.
\end{align}
We have \(g(x)=1-1/(1+x^{k})\) so that \(\tilde g(\kappa)\) is a delta function
at \(0\) minus the inverse Fourier transform \(\tilde h(\kappa)\) of \(h(x)=1/(1+x^{k})\). We have
\begin{align}
  \int_{\kappa} d\kappa |\kappa|^{l}|\tilde g(\kappa)|
  &=
  \delta_{0,l}+\int_{\kappa}d\kappa|\kappa|^{l}|\tilde h(\kappa)|,
\end{align}
where \(\delta_{0,l}\) is the Kronecker
delta.  The inverse Fourier transform of \(h(x)\) can be obtained from
its partial fractions expansion. The function \(h(x)\) has poles at
the non-degenerate roots of \(p(x)=1+x^{k}\), which are at
\(x_{j}=e^{i\pi2j/k}\) for
\(j=-(k-1)/2,-(k-3)/2,\ldots,-1/2,1/2,\ldots,(k-1)/2\), since \(k\) is
even. We have \(x_{j}=x_{-j}^{*}\).  Thus
\(h(x)=\sum_{j=0}^{k-1}\frac{c_{j}}{x-x_{j}}\), where the coefficients
\(c_{j}\) are given by
\(c_{j}=\frac{1}{p'(x_{j})} = \frac{1}{k
  x_{j}^{k-1}}=\frac{-x_{j}}{k}\).  For \(1/2\leq j\leq (k-1)/2\), the
absolute value of the inverse Fourier transform of \(x_{j}/(x-x_{j})\)
is \(\theta(-\kappa)e^{\kappa \Im(x_{j})}\) where \(\theta(y)\) is
\(0\) for \(y\leq 0\) and \(1\) otherwise. Similarly, for
\(-(k-1)/2\leq j\leq -1/2\), this absolute value is
\(\theta(\kappa)e^{-\kappa \Im(x_{j})}\).  Summing these absolute
values over \(j\) and dividing by \(k\) gives
\begin{align}
  |\tilde h(\kappa)|
  &\leq \frac{1}{k}\sum_{j=1/2}^{(k-1)/2}
    e^{-|\kappa| \Im(x_{j})}.
\end{align}
For \(y>0\), we have \(\int_{\kappa\in\rls}e^{-|\kappa| y} = 2/y\)
and \(\int_{\kappa\in\rls}\kappa^{2}e^{-|\kappa|y} = 4/y^{3}\).
For \(1/2\leq j\leq (k-1)/2\), \(\Im(x_{j}) = \sin(2\pi j/k)\geq \sin(\pi/k)\).
Consequently
\begin{align}
  \int_{\rls}|\tilde h(\kappa)|
  &\leq \frac{2}{\sin(\pi/k)},
    \notag\\
  \int_{\rls}\kappa^{2}|\tilde h(\kappa)|
  &\leq \frac{4}{\sin(\pi/k)^{3}}.
\end{align}
Substituting back into Eq.~\eqref{eq:momenterror1}
\begin{align}
  \qty|\langle f(\hat{q}_{\delta})\rangle - \langle\hat{q}^{k}\rangle|^{2}
  &\leq
        \frac{32}{\sin(2\pi/k)^{3}\lambda^{2(k-1)}}
    \delta^{2}\qty(1+\frac{2}{\sin(\pi/k)}+\lambda^{2}\frac{4}{\sin(\pi/k)^{3}})(\langle\hat{n}_{\tot}\rangle+1)
    \notag\\
  &\hphantom{{}\leq{}\;\;\;}
        + 2\lambda^{2k}\langle \hat{q}^{4k}\rangle
    .
    \label{eq:momenterror2}
\end{align}
When applying this inequality, we do not know the value of
\(\langle \hat{q}^{4k}\rangle\), but we can replace this value by a
conservatively estimated upper bound
\(q_{4k}\geq \langle\hat{q}^{4k}\rangle\).  We assume that
\(q_{4k}\geq 1/2\) and minimize the right-hand side of the inequality
over positive \(\lambda\). 
For the estimated minimum to be less than \(1\) requires
that \(2\lambda^{2k}q_{4k}< 1\), which implies \(\lambda<1\).
Therefore \(2\lambda^{2k}q_{4k}\leq 2\lambda^{2(k-1)}q_{4k}\).
Since \(k\geq 4\), \(\sin(\pi/k)\leq 1/\sqrt{2}\), so
\begin{align}
  1+\frac{2}{\sin(\pi/k)}+\lambda^{2}\frac{4}{\sin(\pi/k)^{3}}
  &\leq \frac{1/(2\sqrt{2}) + 1 + 4\lambda^{2}}{\sin(\pi/k)^{3}}
    \notag\\
  &\leq \frac{6}{\sin(\pi/k)^{3}},
\end{align}
where the denominator of the last fraction was chosen to reduce
expression complexity at the cost of a more conservative bound.  We
can set
\begin{align}
  \lambda^{2(k-1)}
  &=\frac{4\sqrt{6}\;\delta\sqrt{\langle\hat{n}_{\tot}\rangle+1}}
    {\sin(\pi/k)^{3}\sqrt{q_{4k}}}
\end{align}
to obtain
\begin{align}
  \qty|\langle f(\hat{q}_{\delta})\rangle - \langle\hat{q}^{k}\rangle|^{2}
  &\leq 16\sqrt{6}\;\delta 
    \frac{\sqrt{\langle\hat{n}_{\tot}\rangle+1}\sqrt{q_{4k}}}
    {\sin(\pi/k)^{3}}.
    \label{eq:momenterror4}
\end{align}
The square of the moment estimation error according to
Eq.~\eqref{eq:momenterror4} scales as \(\delta\), while the explicit
operator expansion for the moment of the pulsed homodyne observable
indicates a scaling of order \(\delta^{2}\) for
\(|\langle \hat{q}_{\delta}^{k}\rangle - \langle\hat{q}^{k}\rangle|\),
for example, see Ref.~\cite{bbp1}. For sufficiently small \(\delta\),
the estimate is therefore poor. The explicit operator expansion
involves large sums of operators of high degree in the mode operators
whose expectations can be very large. As a result it is difficult to
apply the explicit expansion directly, and for moderately small
\(\delta\), the higher order terms can contribute significantly.  The
bounds above, while very conservative, have the advantage of being
simple to apply, requiring only known bounds on
\(\langle \hat n\rangle\) and \(\langle \hat{q}^{4k}\rangle\). For
example, with \(\delta=10^{-4}\) and
\(\langle \hat{n}_{\tot}\rangle=5\), for \(k=4\) and \(k=6\) the error
according to the bound of Eq.~\eqref{eq:momenterror4} becomes
\(0.027\sqrt{q_{16}}\) and \(0.077\sqrt{q_{24}}\), respectively.
\ignore{
mb4 = 16*sqrt(6)*10^-4*sqrt(5+1)/(sin(pi/4)^3)
mb6 = 16*sqrt(6)*10^-4*sqrt(5+1)/(sin(pi/6)^3)
}

We remark that because
the fundamental inequality leading to the bound of
Eq.~\eqref{eq:momenterror4} is based on Prop.~\ref{prop:regfexp}, it
is not just a bound on the difference between expectations, but a
bound on the difference between the states obtained by applying
\(f(\hat{q}_{\delta})\) and \(f_{0}(\hat{q})\). This is one reason for
the conservative power of \(\delta\).  It may also be possible to
choose better approximations \(f(x)\) of \(f_{0}(x)\) to improve the
scaling of the bound.

We do not perform the exercise of computing the error for moment
measurements based on the bounds for standard pulsed homodyne in
App.~\ref{app:standardhomodyne}. This would require bounding the
fourth moment of \(\tilde g(\kappa)\) and involve a more difficult
optimization over \(\lambda\), but would likely significantly improve
the error bounds.

\subsection{CV teleportation.}
For an example involving conditional operations, we consider
continuous variable (CV) teleportation as described
in~\cite{braunstein1998teleportation,furusawa1998teleportation}.  We
describe how to implement the results for conditional operations of
Sect.~\ref{sec:condoperations} to obtain bounds on the three
contributions to infidelity of measurement-conditional operations
given in the chain of Eq.~\eqref{eq:condchain}.
We do not instantiate the bounds numerically here.  These
bounds apply for generic measurement noise models.  For measurement
noise that can be treated as displacement noise before the conditional
operations, it is possible to apply the simpler measurement fidelity
results of Sect.~\ref{sec:sconv} as described and numerically
implemented in the discussion of CV error correction below.

We first consider the middle step in the chain of
Eq.~\eqref{eq:condchain}.  In CV teleportation, there are two
quadrature measurements on independent modes.  For each quadrature
measurement, a conditional displacement of the form
\(\hat{U}_{B}(x)=e^{i\xi x \hat{p}_{B}}\) is applied, where
\(\hat{p}_{B}\) is one of two conjugate canonical quadratures of a
mode of \(B\), which are normalized as described after
Eq.~\eqref{eq:quaddef}. The variable \(x\) is in the spectrum of
\(\hat{q}_{\delta}\), which is normalized according to our conventions
and therefore a factor of \(\sqrt{2}\) larger than the corresponding
variable for canonical quadratures.  These conditional displacements
have \(v(x)=\xi x\) with \(\xi>0\) and \(w(x)=0\) when applying
Prop.~\ref{prop:conddisp}. The ideal scale factor \(\xi\) is \(\xi=1\)
because of the normalization of \(\hat{q}\).  In the references, this
factor is \(\sqrt{2}\) for canonical quadratures.  The factor \(\xi\)
may be modified by a gain to compensate for experimental effects.
Because \(w(x)=0\), we have \(w_{l}=\tilde w_{l}=0\). Therefore the
operators defined in Prop.~\ref{prop:conddisp} are \(\hat{P}_{B}=1\)
and \(\hat{Q}_{B}=2\xi^{2}\hat{p}_{B}^{2}\).  The difference between
the state \(\ket{\psi_{c,0}}\) obtained by applying the ideal
quadrature-conditional displacement to initial state \(\ket{\psi}\)
and the state \(\ket{\psi_{c,\delta}}\) obtained by applying the
displacement conditional on the operator \(\hat{q}_{\delta}\) for
homodyne measurement with LO amplitude \(1/\delta\) is therefore given
by
\begin{align}
          \qty| \ket{\psi_{c,\delta}}-\ket{\psi_{c,0}}|^{2}
  &\leq 8\delta^{2}\xi^{2}
    \bra{\psi}(1+2\hat{p}_{B}^{2}\xi^{2})\hat{p}_{B}^{2}
    (\hat{n}_{\tot}+1)
    \ket{\psi}
    .
\end{align}
This bound applies to one of the two required conditional
displacements.  To obtain a bound on both of them, let \(\hat{x}_{B}\)
and \(\hat{p}_{B}\) be the two conjugate canonical quadratures satisfying
\([\hat{x}_{B},\hat{p}_{B}] = i\) of the mode of \(B\). The first
displacement applies \(\hat{U}_{B}(x)\) conditioned on
\(\hat{q}_{a,\delta}\), and the second applies
\(\hat{V}_{B}(x)=e^{i\xi\hat{x}_{B}}\) conditioned on
\(\hat{q}_{b,\delta}\), where \(\hat{q}_{a,\delta}\) and
\(\hat{q}_{b,\delta}\) are the observables for homodyne measurement
with LO amplitude \(1/\delta\) associated with quadratures
\(\hat{q}_{a}=\hat{q}_{a,0}\) and \(\hat{q}_{b}=\hat{q}_{b,0}\)
of orthogonal modes.  Let \(\ket{\psi_{D,\delta}}\) denote
the state after applying the two displacements conditioned on
\(\hat{q}_{a,\delta}\) and \(\hat{q}_{b,\delta}\). We wish to bound
\begin{align}
  \qty| \ket{\psi_{D,\delta}}-\ket{\psi_{D,0}}|
  &=
    \qty| e^{i\xi \hat{q}_{b,\delta}\hat{x}_{B}}e^{i\xi\hat{q}_{a,\delta}\hat{p}_{B}}\ket{\psi} - e^{i\xi \hat{q}_{b}\hat{x}_{B}}e^{i\xi\hat{q}_{a}\hat{p}_{B}}\ket{\psi}|
    \notag\\
  &\leq
    \qty|e^{i\xi\hat{q}_{b,\delta}\hat{x}_{B}}\qty(e^{i\xi\hat{q}_{a}\hat{p}_{B}}\ket{\psi})
    - e^{i\xi\hat{q}_{b}\hat{x}_{B}}\qty(e^{i\xi\hat{q}_{a}\hat{p}_{B}}\ket{\psi})|
    \notag\\
  &\hphantom{\leq\,|\,{}}
    + \qty| e^{i\xi\hat{q}_{b,\delta}\hat{x}_{B}}\qty(e^{i\xi\hat{q}_{a,\delta}\hat{p}_{B}}\ket{\psi}
    -e^{i\xi\hat{q}_{a}\hat{p}_{B}}\ket{\psi})|.
\end{align}
For the second summand, the previously obtained bound applies as follows:
\begin{align}
  \qty| e^{i\xi\hat{q}_{b,\delta}\hat{x}_{B}}\qty(e^{i\xi\hat{q}_{a,\delta}\hat{p}_{B}}\ket{\psi}
  -e^{i\xi\hat{q}_{a}\hat{p}_{B}}\ket{\psi})|^{2}
  &=\qty|e^{i\xi\hat{q}_{a,\delta}\hat{p}_{B}}\ket{\psi}
    -e^{i\xi\hat{q}_{a}\hat{p}_{B}}\ket{\psi}|^{2}
    \notag\\
  &\leq 8\delta^{2}\xi^{2}
    \bra{\psi}(1+2\hat{p}_{B}^{2}\xi^{2})\hat{p}_{B}^{2}
    (\hat{n}_{a,\tot}+1)
    \ket{\psi},
\end{align}
where \(\hat{n}_{a,\tot}\) is the total photon number of the mode of \(\hat{q}_{a}\).
For the first summand, the operator
\(e^{i\xi\hat{q}_{a}\hat{p}_{B}}\) does not affect the LO modes, so
preserves vacuum on the LO modes. The bound of
Prop.~\ref{prop:conddisp} therefore applies, giving 
\begin{align}
  \qty|e^{i\xi\hat{q}_{b,\delta}\hat{x}_{B}}\qty(e^{i\xi\hat{q}_{a}\hat{p}_{B}}\ket{\psi})
  - e^{i\xi\hat{q}_{b}\hat{x}_{B}}\qty(e^{i\xi\hat{q}_{a}\hat{p}_{B}}\ket{\psi})|^{2}
  \hspace*{-2in}&\notag\\
  &\leq 
    8\delta^{2}\xi^{2}
    \bra{\psi}e^{-i\xi\hat{q}_{a}\hat{p}_{B}}
    (1+2\hat{x}_{B}^{2}\xi^{2})\hat{x}_{B}^{2}
    (\hat{n}_{b,\tot}+1)
    e^{i\xi\hat{q}_{a}\hat{p}_{B}}\ket{\psi}.
\end{align}
Since \(e^{-i\xi\hat{q}_{a}\hat{p}_{B}}\hat{x}_{B}e^{i\xi\hat{q}_{a}\hat{p}_{B}}
= \hat{x}_{B}- \xi\hat{q}_{a}\), we now have
\begin{align}
  \qty| \ket{\psi_{D,\delta}}-\ket{\psi_{D,0}}|^{2}
  &\leq 16\delta^{2}\xi^{2}\bra{\psi}\Big(
    (1+2(\hat{x}_{B}- \xi\hat{q}_{a})^{2}\xi^{2})
    \notag\\
    &\hphantom{\leq 16\delta^{2}\xi^{2}\bra{\psi}\Big((1}
    \times(\hat{x}_{B}- \xi\hat{q}_{a})^{2}
    (\hat{n}_{b,\tot}+1)
    \notag\\
  &\hphantom{\leq 8\delta^{2}\xi^{2}}
    + 
    (1+2\hat{p}_{B}^{2}\xi^{2})\hat{p}_{B}^{2}
    (\hat{n}_{a,\tot}+1)
    )\Big)\ket{\psi}
\end{align}
with the extra factor of two accounting for the cross terms from
squaring that are dropped. To apply the bound it suffices to use
conservative upper bounds for the expectations of the products of
quadrature and number operators that are obtained when multiplying out
the terms in the bound. The bound can be expressed in terms of
correlated moments of number operators by bounding \(\hat{x}_{B}^{2}\)
and \(\hat{p}_{B}^{2}\) above by \(2\hat{n}_{B}+1\), applying
Eq.~\eqref{app:eq:quadsq-hatn} with the normalization for conjugate
quadratures.

Next we consider the last step of Eq.~\eqref{eq:condchain}.  To obtain
its contribution to infidelity, we can apply
Prop.~\ref{prop:conditionalD}.  We can model homodyne measurements
with photodiodes such as in Ref.~\cite{hansen2001ultrasensitive} with
\(h(z|x)=e^{-(z-x)^{2}/(2\delta^{2}\sigma^{2})}/(\sqrt{2\pi}\delta\,\sigma)\)
in Prop.~\ref{prop:conditionalD}. Accordingly, in this proposition,
\(K_{1}=\xi\) and \(K_{2}=\delta^{2}\sigma^{2}\).
\begin{align}
  \qty| \ket{\tilde\psi_{m,\delta}}-\ket{\psi_{m,\delta}}|^{2}
  &\leq
    2\xi^{2}\delta^{2}\sigma^{2}\langle (\hat{n}_{\tot}+1/2 )^{2}\rangle^{1/2}.
\end{align}
Because of the two measurements, the distance from this expression is
added twice to the bound on the distance in the previous paragraph, once for
\(\hat{n}_{\tot}=\hat{n}_{a,\tot}\) and once for
\(\hat{n}_{\tot}=\hat{n}_{b,\tot}\).

For the contribution from the first step of the chain of
Eq.~\eqref{eq:condchain} we can apply Thm.~\ref{thm:qmeas} to the two
measurements with the conditional displacements applied perfectly and
unitarily, conditional on the target quadratures before coupling to
the apparatus. For applying Thm.~\ref{thm:qmeas}, the state being
measured corresponds to the state after ``teleporting'' unitarily
and does not actually occur in a teleportation experiment.  For
Gaussian added noise, we can apply Eq.~\eqref{eq:outstate_fid} twice,
where \(\langle \hat{n}_{\tot}\rangle\) must be estimated for the
state obtained after the unitary conditional displacements and each of
the two measurement modes.

We do not instantiate the above bounds numerically here because for
the example of a balanced pair of photodiodes adding Gaussian noise,
this added noise is equivalent to displacement noise before the
measurement.  If such displacement noise has already been accounted
for in the error budget, it is possible to bypass the above bounds and
directly apply the results of Sect.~\ref{sec:sconv} or
App.~\ref{app:standardhomodyne} for post-measurement states. This
technique is used below to analyze CV error correction with a closely
related measurement configuration, and the bounds obtained there can
be applied to teleportation with few modifications.

\subsection{CV error correction.}\label{sect:cvec}
We consider Gottesman-Kitaev-Preskill (GKP) error correction for
qubits encoded in a quantum mode.  There are two standard ways to
implement GKP error correction, Steane error correction and teleported
error correction.  Versions of these protocols are shown in Fig.~10 of
Ref.~\cite{grimsmo2021quantum}.  The protocol for Steane error
correction involves the same measurements and shifts as for
teleportation, so the analysis for teleportation given above can be
applied given bounds on expectations of correlated photon number
moments computed for the appropriate input states. Teleported
error correction does not require explicitly correcting
shifts. However, the logical Pauli frame is adjusted according to the
measurement outcomes, which affects future operations. For the purpose
of applying the fidelity analysis for conditional operations, it is
necessary to apply outcome-conditional logical operations to reset the
frame and remove the dependence of future operations on the
measurement outcomes. Because these logical operations are not simple
linear shifts, a direct application of our bounds requires
regularization with Lem.~\ref{lem:condreg}.  We do not perform this
exercise here.  Instead we take advantage of the equivalence of added
noise in the measurement outcomes to displacement noise on the
incoming GKP qubit.  In particular, for GKP error correction based on
quadrature measurements, added Gaussian noise arising from the use of
balanced photodiodes as detectors according to
Eq.~\eqref{eq:gaussian_added_noise} corresponds to added Gaussian
displacements noise.  Such noise can be included in the analysis of
the quality of the error-correction process with ideal
quadrature-conditional operations. Therefore, the additional
infidelity from using homodyne measurements with finite LO amplitude
can be accounted for by the measurement fidelity analysis of
Sect.~\ref{sec:sconv}. By monotonicity of fidelity, subsequent
operations conditional on the measurement outcomes do not increase
this infidelity.  In the remainder of this section we formalize this
approach to bounding the added infidelity in GKP error correction due
to finite LO amplitude when implementing the homodyne measurements.

Following the notation introduced at the beginning of
Sect.~\ref{sec:condoperations}, let \(\ket{\psi}\) be the initial
state of modes \(\bm{a}\), \(B\) and the LO modes \(\bm{b}\), where
\(\ket{\psi}\) is vacuum on the LO modes .  In the case of GKP error
correction, modes of system \(B\) carry encoded information, while
modes \(\bm{a}\) contain syndrome information to be used for error
correction or frame tracking, where the currently relevant syndrome
information is carried by the quadrature \(\hat{q}\).  In principle,
the ideal error correction procedure or future frame-dependent
unitaries implement unitaries conditional on \(\hat{q}\). Let
\(\hat{p}\) be a quadrature of the mode of \(\hat{q}\) such that the
displacement \(\hat{D}_{x}=e^{i x\hat{p}}\) satisfies
\(\hat{D}_{-x}\hat{q}\hat{D}_{x}=\hat{q}-x\).  For
\(\hat{q}=\hat{q}_{\alpha}\) as defined in
Sect.~\ref{sec:preliminaries}, we have
\(\hat{p}=\hat{q}_{i\alpha/(2|\alpha|^{2})}\).  In the next two
paragraphs, we show that additive measurement noise with probability
distribution \(\mu(x)dx\) is equivalent to random displacements
\(\hat{D}_{x}\) acting on the mode of \(\hat{q}\), where \(x\) has
probability distribution \(\mu(x)dx\).  We can assume without loss of
generality that \(\bm{a}\) consists of only one mode, mode \(a\), with
quadratures \(\hat{q}\) and \(\hat{p}\).

In general, a dilated, unitary model for the displacement noise is
obtained by adding the quantum system \(M\) consisting of a
one-dimensional system with position operator \(\hat{x}_{M}\) and
initial wavefunction \(\ket{\sqrt{\mu(\bm{\cdot})}}\) in the position
basis for \(\hat{x}_{M}\). The system \(M\) is coupled to the mode of
\(\hat{q}\) by \(e^{i\hat{x}_{M}\hat{p}}\).  The random displacement
noise is obtained after tracing out \(M\). Future conditional
operations applied after displacement noise are of the form
\(\hat{U}(\hat{q})\) and applied after the coupling to \(M\).  Since
\(\hat{q}e^{i\hat{x}_{M}\hat{p}}=
e^{i\hat{x}_{M}\hat{p}}(\hat{q}-\hat{x}_{M})\), the conditional
operations \(\hat{U}(\hat{q})\) after the coupling are equivalent to
the corresponding operations \(\hat{U}(\hat{q}-\hat{x}_{M})\)
conditioned on \(\hat{q}-\hat{x}_{M}\) before the coupling.  Given
that the mode \(a\) and the system \(M\) play no further role
besides these conditional operations, the dilated model is equivalent
to initializing \(M\) in state \(\ket{\sqrt{\mu(\bm{\cdot})}}\), and
applying the needed conditional operations conditional on
\(\hat{q}-\hat{x}_{M}\), without any other operations applied to mode
\(a\) or apparatus \(M\).
  
We compare random displacement noise before \(\hat{q}\)-conditional
operations to an instance of the unitary measurement model of
Sect.~\ref{sec:functana} with operations conditioned on the apparatus,
whose quantum system is also \(M\).  For the unitary measurement
model, we prepare \(M\) with the same initial wavefunction
\(\ket{\sqrt{\mu(\bm{\cdot})}}\). But \(M\) is coupled to the mode of
\(\hat{q}\) by \(e^{i\hat{p}_{M}\hat{q}}\), where \(\hat{p}_{M}\) is
the momentum operator for \(M\) scaled so that the displacement
\(\hat{D}_{y}=e^{iy\hat{p}_{M}}\) satisfies
\(\hat{D}_{-y}\hat{x}_{M}\hat{D}_{y}= \hat{x}_{M}-y\).  Future
operations are of the form \(\hat{U}(-\hat{x}_{M})\), conditioned on
\(-\hat{x}_{M}\) rather than \(\hat{q}\). Here, to simplify the
desired correspondence, we chose the apparatus coupling so that the
measured quantity is encoded in the eigenvalues of \(-\hat{x}_{M}\)
rather than \(\hat{x}_{M}\).  Since
\((-\hat{x}_{M})e^{i\hat{p}_{M}\hat{q}} =
e^{i\hat{p}_{M}\hat{q}}(-\hat{x}_{M}+\hat{q})\), these conditional
operations are equivalent to operations
\(\hat{U}(\hat{q}-\hat{x}_{M})\) conditioned on
\(\hat{q}-\hat{x}_{M}\) before the apparatus coupling. The situation
is therefore equivalent to the dilated random displacement model,
except for the irrelevant different processing of mode \(a\) and
apparatus \(M\) after the conditional operations acting before the
coupling.

The above shows that added displacement noise in measurement according
to the unitary measurement model with subsequent apparatus-conditional
unitaries is equivalent to added displacement noise before the
corresponding quadrature-conditional unitaries. This added noise can
be accounted for in the analysis of the performance of GKP error
correction. Added infidelity due to the difference between
\(\hat{q}_{\delta}\) and \(\hat{q}\) at finite LO amplitude can then
be attributed to the infidelity between the global state after
apparatus coupling with the ideal coupling \(e^{i\hat{p}_{M}\hat{q}}\)
and the state obtained with the coupling
\(e^{i\hat{p}_{M}\hat{q}_{\delta}}\), where the apparatus initial
state reflects the detector added noise. The strategy is then to first
consider GKP error correction when the incoming state has been
affected by physical displacement noise with a given variance.  We add
to the physical noise the effective displacement noise contributed by
noisy measurements of the ideal quadratures and bound the infidelity
of the logical information for ideal error correction, where in the
ideal error correction, operations are conditional unitaries
conditioned on the ideal quadratures.  Second, we bound the infidelity
between the state obtained with the ideal apparatus coupling, and the
state obtained with the apparatus coupling with the homodyne
observable at LO amplitude \(1/\delta\), where the apparatus initial
state reflects the detector added noise. For this bound, we can use
Eqs.~\eqref{eq:outstate_fid}. For Gaussian added noise described by
\(\mu(x)=\frac{1}{r\sqrt{2\pi}}e^{-x^{2}/(2r^{2})}\), we can apply
Eq.~\eqref{eq:outstate_fid_gauss}.  For a worst-case bound on the
overall infidelity, one can then combine the two infidelity bounds.
Instead, we determine the LO amplitude and resolution required to meet
separate infidelity budgets for the two contributions to infidelity.
The resolution-amplitude tradeoffs for constant-added-noise
photodiodes discussed after Eq.~\eqref{eq:outstate_fid_gauss} apply
here, but now become a tradeoff between added displacement noise that
must be handled by the GKP codes, and infidelity in the
post-measurement state due to finite LO amplitude.

\begin{figure}

\begin{tikzpicture}[
    every path/.style={very thick},
    cnot/.style={draw, circle, inner sep=0pt, minimum size=12pt, fill=white, very thick},
    label node above/.style={inner sep=2pt, font=\small, fill=none, above=6pt},
    label node below/.style={inner sep=2pt, font=\small, fill=none, below=6pt}
]

    \def\rowsep{1.8} 
    \def\colsep{2.5}
    \def\wirelength{8.5}

    \foreach \i in {1,2,3} {
        \coordinate (L\i) at (0, -\i*\rowsep);
        \coordinate (R\i) at (\wirelength, -\i*\rowsep);
    }

    \draw (L1) -- (R1);
    \node[left=5pt] at (L1) {$\ket{\psi_{L,\text{in}}}_{a}$};
    \node[label node above] at ($(L1)!0.1!(R1)$) {$\hat{q}_{a}$};
    \node[label node below] at ($(L1)!0.1!(R1)$) {$\hat{p}_{a}$};
    
    \node[label node above] at ($(L1)!0.65!(R1)$) {$\hat{q}_{a}'=\hat{q}_{a}$};
    \node[label node below] at ($(L1)!0.65!(R1)$) {$\hat{p}_{a}'=\hat{p}_{a}- \hat{p}_{b}$};

    \draw (L2) -- (R2);
    \node[left=5pt] at (L2) {$\ket{0_L}_{b}$};
    \node[label node above] at ($(L2)!0.1!(R2)$) {$\hat{q}_{b}$};
    \node[label node below] at ($(L2)!0.1!(R2)$) {$\hat{p}_{b}$};

    \node[label node above] at ($(L2)!0.7!(R2)$) {$\hat{q}_{b}'=\hat{q}_{a}+\hat{q}_{b}+\hat{q}_{c}$};
    \node[label node below] at ($(L2)!0.7!(R2)$) {$\hat{p}_{b}'=\hat{p}_{b}$};

    \draw (L3) -- (R3);
    \node[left=5pt] at (L3) {$\ket{+_L}_{c}$};
    \node[right=5pt] at (R3) {$\ket{\psi_{L,\text{out}}}_{c}$};
    \node[label node above] at ($(L3)!0.1!(R3)$) {$\hat{q}_{c}$};
    \node[label node below] at ($(L3)!0.1!(R3)$) {$\hat{p}_{c}$};

    \node[label node above] at ($(L3)!0.5!(R3)$) {$\hat{q}_{c}'=\hat{q}_{c}$};
    \node[label node below] at ($(L3)!0.5!(R3)$) {$\hat{p}_{c}'=\hat{p}_{c}-\hat{p}_{b}$};

    
    \draw ($(L3)+(1.5,0)$) -- ($(L2)+(1.5,0)$);
    \fill ($(L3)+(1.5,0)$) circle (2.8pt);
    \node[cnot] at ($(L2)+(1.5,0)$) {};
    \draw ($(L2)+(1.5,-0.2)$) -- ($(L2)+(1.5,0.2)$);
    \draw ($(L2)+(1.5-0.2,0)$) -- ($(L2)+(1.5+0.2,0)$);

    \draw ($(L1)+(3.5,0)$) -- ($(L2)+(3.5,0)$);
    \fill ($(L1)+(3.5,0)$) circle (2.8pt);
    \node[cnot] at ($(L2)+(3.5,0)$) {};
    \draw ($(L2)+(3.5,-0.2)$) -- ($(L2)+(3.5,0.2)$);
    \draw ($(L2)+(3.5-0.2,0)$) -- ($(L2)+(3.5+0.2,0)$);

    \foreach \y/\txt in {1/\hat{p}_{a}', 2/\hat{q}_b'} {
        \begin{scope}[shift={(R\y)}]
            \draw[fill=white] (0, 0.6) -- (0, -0.6) -- (0.6, -0.6) 
                arc (-90:90:0.5 and 0.6) -- cycle;
            \node at (0.45, 0) {$\txt$};
        \end{scope}
    }

  \end{tikzpicture}
  
  \caption{GKP teleported error correction. The top line carries mode
    \(a\) which contains the incoming GKP qubit, the bottom line
    carries mode \(c\) which contains the outgoing GKP qubit. The mode
    operators for the incoming modes and for the modes after coupling
    are shown with their relationships. The couplings are denoted by
    CNOT symbolic gates. These are active linear optical. For example,
    the one that couples mode \(c\) to mode \(b\) commutes with
    \(\hat{q}_{c}\) and \(\hat{p}_{b}\) and the quadratures
    \(\hat{q}_{b}''\) and \(\hat{p}_{c}''=\hat{p}_{c}'\) are given by
    \(\hat{q}_{b}+\hat{q}_{c}\) and \(\hat{p}_{c}-\hat{p}_{b}\) in
    terms of the incoming quadratures.  The diagram is based on
    Fig.~10 of Ref.~\cite{grimsmo2021quantum}.}
  \label{fig:gkptelerr}
\end{figure}

To estimate the measurement fidelities and account for the effect of
the displacement noise on GKP encoded information we consider
teleported error correction with the couplings shown in
Fig.~\ref{fig:gkptelerr}, see also Fig.~10 of
Ref.~\cite{grimsmo2021quantum}. Teleported error correction requires
coupling the mode containing a GKP encoded qubit to two ancillas
prepared independently in GKP encoded states, followed by measuring
quadratures. The measured quadratures and their associated
displacement noise can be expressed in terms of the quadratures of the
three modes \(a,b,c\) consisting of the incoming GKP qubit in mode
\(a\) and two ancilla GKP qubits in modes \(b\) and \(c\).  Let
\(\hat{q}_{u},\hat{p}_{u}\) be the incoming GKP conjugate quadratures
of mode \(u\) with \(u\in\{a,b,c\}\), where we use the normalization
conventions from Sect.~\ref{sec:preliminaries} for these
quadratures. The teleported error correction circuit is shown in
Fig.~\ref{fig:gkptelerr}.  The quadratures \(\hat{q}_{a}'\) and
\(\hat{p}_{b}'\) are intended to be measured and subsequently used for
error correction or Pauli frame tracking. We describe such actions as
conditional correction unitaries.  We analyze the effect of displacement noise
generated by the relevant quadratures, propagating the effect to the
incoming and outgoing GKP qubits, where we add the effect on the
outgoing GKP qubit to the error on incoming GKP qubit in the next
error-correction cycle.  From this analysis, we can determine the
performance of the error-correction circuit in terms of noise acting
only on the incoming GKP qubit.

To propagate the displacement noise to the incoming GKP qubit, we
consider each displacement-generating quadrature. For this purpose we
express the quadratures acting just before measurements in terms of
incoming quadratures as shown in Fig.~\ref{fig:gkptelerr}. The identities
represent equivalent actions in terms of the overall effect on the circuit.
\begin{itemize}
\item Displacements generated by \(\hat{q}_{b}\): We can write
  \begin{align}
    \hat{q}_{b}=\hat{q}_{b}' -
    \hat{q}_{a}-\hat{q}_{c}=\hat{q}_{b}'-\hat{q}_{a}-\hat{q}_{c}'.
  \end{align}
  Displacements generated by \(\hat{q}_{b}'\) have no effect on the
  circuit since \(\hat{q}_{b}'\) is measured. Thus, displacements
  generated by \(\hat{q}_{b}\) are equivalent to simultaneous
  displacements of the incoming GKP qubit according to \(\hat{q}_{a}\)
  and the outgoing GKP qubit by \(\hat{q}_{c}'\).
\item Displacements generated by \(\hat{p}_{b}\): We can write
  \begin{align}
    \hat{p}_{b}=-\hat{p}_{a}' + \hat{p}_{a}.
  \end{align}
  Displacements generated by \(\hat{p}_{a}'\) have no effect since \(\hat{p}_{a}'\)
  is measured. Thus, displacements generated by \(\hat{p}_{b}\) are equivalent
  to displacements of the incoming GKP qubit according to \(\hat{p}_{a}\).
\item Displacements generated by \(\hat{q}_{c}\):
  Since \(\hat{q}_{c}=\hat{q}_{c}'\), these displacements act directly on
  the outgoing GKP qubit.
\item Displacements generated by \(\hat{p}_{c}\):
  We can write
  \begin{align}
    \hat{p}_{c}=\hat{p}_{c}'-\hat{p}_{a}'+\hat{p}_{a}.
  \end{align}
  The contribution from \(\hat{p}_{a}'\) is invisible, that from \(\hat{p}_{a}\)
  acts on the incoming GKP qubit, and that from \(\hat{p}_{c}'\) on the outgoing
  GKP qubit.
\item Displacements generated by \(\hat{q}_{a}'\): This is measurement-added
  noise in the measurement of \(\hat{p}_{a}'\). Since \(\hat{q}_{a}'=\hat{q}_{a}\),
  this is a displacement of the incoming GKP qubit according to \(\hat{q}_{a}\).
\item Displacements generated by \(\hat{p}_{b}'\): This is measurement-added
  noise in the measurement of \(\hat{q}_{b}'\). We can write
  \begin{align}
    \hat{p}_{b}' = \hat{p}_{b}= -\hat{p}_{a}' + \hat{p}_{a},
  \end{align}
  which generates displacements equivalent to displacements of the
  incoming GKP qubit according to \(\hat{p}_{a}\).
\end{itemize}

We consider square GKP qubits. With respect to conjugate canonical
quadratures \(\hat{x}\) and \(\hat{p}\) normalized as described after
Eq.~\eqref{eq:quaddef} and satisfying \([\hat{x},\hat{p}] = i\), the
ideal but improper logical states \(\ket{s}\) for \(s=0,1\) are joint
improper eigenstates of the stabilizers \(e^{i2\sqrt{\pi}\hat{x}}\)
and \(e^{i2\sqrt{\pi}\hat{p}}\) that in the position space for
\(\hat{x}\) can be written as
\(\sum_{j=-\infty}^{\infty}\ket{(2j+s)\sqrt{\pi}}\).  Our bounds are
expressed in terms of quadrature operators \(\hat{q}\) that are
normalized according to \(\hat{q}=\sqrt{2}\hat{x}\) or
\(\hat{q}=\sqrt{2}\hat{p}\), depending on which quadrature is being
measured.  The measurement added noise is expressed in terms of
measurements of quadratures with the latter normalization.  Thus noise
with variance \(r^{2}\) in measuring \(\hat{q}\) corresponds to noise
with variance \(r^{2}/2\) in measuring \(\hat{x}\) or \(\hat{p}\).
Below we refer to numerical calculations of GKP error correction
fidelities with respect to displacement noise. For these calculations,
displacement noise is with respect to canonical quadratures. So
Gaussian noise with variance \(\sigma_{g}^{2}\) for a canonical
quadrature is equivalent to noise with variance \(2\sigma_{g}^{2}\)
for the quadrature normalized according to our conventions. In
general, with the quadrature definitions of
Sect.~\ref{sec:preliminaries}, the variances are related by
\(\langle (\hat{q}_{\alpha}-\langle\hat{q}_{\alpha}\rangle)^{2}\rangle
=
2|\alpha|^{2}\langle(\hat{q}_{1/\sqrt{2}}-\langle\hat{q}_{1/\sqrt{2}}\rangle)^{2}\rangle\),
where \(\hat{q}_{1/\sqrt{2}}\) is canonical.  For the remainder of
this analysis, variances with the Greek letter \(\sigma\) refer to
displacement noise variances for canonical quadratures. We convert
measurement added noise with variance \(r^{2}\) for measuring
\(\hat{q}\) normalized according to our conventions to canonical
displacement noise with variance \(r^{2}/2\) before the measurement.

We assume that the incoming GKP qubit and ancilla qubits both
experience physical Gaussian isotropic displacement noise with
marginal variance \(\sigma_{0}^{2}\) for each canonical quadrature.
In consideration of the equivalences given above, we can redistribute
it and attribute it to noise on the incoming GKP qubit with no noise
acting on the ancillas. For this, we add the variances of the noisy
displacement for each quadrature from redistributing noise in the
previous error-correction cycle acting on the then outgoing and now
incoming GKP qubit to the physical noise and the redistributed noise
variances from the current error-correction cycle. The resulting noisy
displacements add independently from each source and after accounting
for redistributed displacement noise from both the previous and the
present error-correction cycle the noise remains isotropic.  Adding
the variances from the redistributed physical noise, we get a total
variance of \(3\sigma_{0}^{2}\) in each canonical quadrature acting
only on the incoming GKP qubit. For displacement noise generated by
canonical \(\hat{q}\) in Fig.~\ref{fig:gkptelerr}, this variance is
obtained by adding the initial noise on the GKP qubit coming in from
the previous error correction cycle, the noise on this qubit
propagated forward from the previous cycle's \(\hat{q}_{b}\)
displacement noise, and the noise propagated from the current cycle's
\(\hat{q}_{b}\) displacement noise. The total displacement noise
generated by \(\hat{p}\) can be determined similarly.  Let \(r\) be
the resolution of homodyne measurement, accounting for
measurement-added noise and the LO amplitude.  The measurement-added
noise contributes a variance \(r^{2}/2\) in the canonical quadratures.
The total noise variance affecting the fidelity of the
error-correction cycle with ideal quadrature-conditional correction
unitaries is therefore
\(\sigma_{\text{noise}}^{2}=3\sigma_{0}^{2}+r^{2}/2\).

It is necessary to estimate the expected photon numbers in the
measured modes. For the contribution to infidelity from the
implemented measurement, it is necessary to calculate the expected
photon numbers with respect to the physical incoming states, without
redistribution of the physical noise or applying the equivalence
between added measurement noise and displacement noise to the current
error-correction cycle. For the incoming GKP qubit, it is necessary to
take into account the effects from the previous cycle and possible
actions. We assume that these effects and actions are well represented
by the physical noise and a displacement that might be used to
implement a logical Pauli gate with displacements to one or both of
the quadratures of the GKP qubit.

We write \(\hat{n}'_{a}\) and \(\hat{n}'_{b}\) for
the number operators of the modes measured for \(\hat{p}_{a}'\) and
\(\hat{q}_{b}'\). For applying our bounds, these quadratures are
normalized according to our conventions.  From the equivalences in
Fig.~\ref{fig:gkptelerr} and applying the identity in
Eq.~\eqref{app:eq:quadsq-hatn},
\begin{align}
  \langle \hat{n}'_{a}\rangle
  &=
    \frac{1}{4}\left\langle\hat{q}_{a}^{2}+
    (\hat{p}_{a}-\hat{p}_{b})^{2}-2\right\rangle \notag
  \\
  &=
    \frac{1}{4}\left\langle\hat{q}_{a}^{2}+
    \hat{p}_{a}^{2}+\hat{p}_{b}^{2} - 2\hat{p}_{a}\hat{p}_{b}-2\right\rangle
    ,
    \ignore{    
    \notag\\
  &\leq
    \frac{1}{4}\qty(\langle\hat{q}_{a}^{2}\rangle+
    \langle\hat{p}_{a}^{2}\rangle
    +\langle\hat{q}_{b}^{2}\rangle+\langle\hat{p}_{b}^{2}\rangle -
    2\langle\hat{p}_{a}\hat{p}_{b}\rangle-2) \notag
  \\
  & \leq \langle
    \hat{n}_{a}\rangle+\langle\hat{n}_{b}\rangle
    -
    \frac{1}{2}
    \langle\hat{p}_{a}\hat{p}_{b}\rangle +
    \frac{1}{2},
    }
    \label{eq:gkp_p}
\end{align} and similarly
\begin{align}
  \langle \hat{n}'_{b}\rangle
  &=
    \frac{1}{4}\left\langle
    (\hat{q}_{a}+\hat{q}_{b}+\hat{q}_{c})^{2} +
    \hat{p}_{b}^{2} -2\right\rangle \notag
  \\
  &=
    \frac{1}{4}\qty(
    \langle\hat{q}_{a}^{2}\rangle+\langle\hat{q}_{b}^{2}\rangle
    +\langle\hat{p}_{b}^{2}\rangle+\langle\hat{q}_{c}^{2}\rangle +
    2\qty(
       \langle\hat{q}_{a}\hat{q}_{b}\rangle 
       +\langle\hat{q}_{a}\hat{q}_{c}\rangle 
       +\langle\hat{q}_{b}\hat{q}_{c}\rangle)
    -2)
    .
    \ignore{
    \notag\\
  &\leq
       \langle \hat{n}_{a}\rangle+\langle\hat{n}_{b}\rangle+
       \langle\hat{n}_{c}\rangle + \frac{1}{2}\qty(
       \langle\hat{q}_{a}\hat{q}_{b}\rangle 
       +\langle\hat{q}_{a}\hat{q}_{c}\rangle 
       +\langle\hat{q}_{b}\hat{q}_{c}\rangle 
    ) + 1.
    }
      \label{eq:gkp_q}
\end{align}
The three input modes are independent. 

For reference, we use the performance data for finite-energy square
GKP qubits shown in App.~\ref{app:gkp_numerics}
Fig.~\ref{fig:gkp_numerics}, where we use the curves that show an
analytic approximation for the entanglement fidelity of teleported
error-correction as a function of the total redistributed displacement
noise variance acting on the incoming GKP qubit.  The error-free
states for these square GKP qubits have close to zero-mean
quadratures.  Let \(\bar n\) be the expected number of photons for a
particular such qubit. Since, with our conventions,
\(4(\bar n+1/2)=\langle\hat{q}^{2}\rangle+\langle\hat{p}^{2}\rangle\),
and the square GKP qubits are approximately symmetric, we use the
estimate
\(v=\langle\hat{q}^{2}\rangle \approx \langle\hat{p}^{2}\rangle \approx
2\bar n+1\).
\newcommand{\gkpnbnd}{4.8}
\newcommand{\sigmagkp}{0.229}
\newcommand{\fsqueeze}[1]{-10 * ln ( 2 * (#1)^2 ) / ln ( 10  )}
\edef\gkpsqz{\fpeval{\fsqueeze{\sigmagkp}}}
\newcommand{\fquadsq}[1]{2 * (#1) +  1}%
\edef\quadsq{\fpeval{\fquadsq{\gkpnbnd}}}%
\newcommand{\infidbnd}{0.06}
\newcommand{\minfidbnd}{0.02}
\newcommand{\sigmabnd}{0.1}
\newcommand{\sigmaphys}{0.05}%
\newcommand{\fmeasres}[2]{sqrt( (2*(#1)^2 - 6 * (#2)^2))}%
\edef\measres{\fpeval{\fmeasres{\sigmabnd}{\sigmaphys}}}%
\newcommand{\fsigmanoise}[2]{sqrt(3*(#1)^2+(#2)^2/2)}%
\edef\sigmanoise{\fpeval{\fsigmanoise{\sigmaphys}{\measres}}}%
We choose the square GKP qubit with
\(\bar n = \num[round-mode=places,round-precision=1]{\gkpnbnd}\) for
illustration.  The squeezing associated with
this GKP qubit is
\(\num[round-mode=places,round-precision=1]{\gkpsqz}\). We fix an
infidelity bound of \(\epsilon_{\text{ec}}=\infidbnd\) for the
entanglement infidelity according to curves in
Fig.~\ref{fig:gkp_numerics}.  We expect our bounds to be quite
conservative and allow \(\epsilon_{m}=\minfidbnd\) for the sum of the
two infidelities from the two quadrature measurements.  We further
expect the three infidelities to add incoherently for an overall
infidelity bounded by
\(\num[round-mode=places,round-precision=3]{\fpeval{\infidbnd+\minfidbnd}}\).
To determine \(\sigma_{\text{noise}}\) we consult the entanglement
fidelity curves for the GKP qubit with
\(\bar n=\num[round-mode=places,round-precision=1]{\gkpnbnd}\) and let
\(\sigma_{\text{noise}}\) be a lower bound on the noise standard
deviation at which the curve crosses the chosen entanglement fidelity
bound of \(\infidbnd\).  Thus we set
\(\sigma_{\text{noise}}=\sigmabnd\).  Assuming a bound of
\(\sigmaphys\) on the physical displacement noise \(\sigma_{0}\), we
solve for \(r\) in the expression for \(\sigma_{\text{noise}}\) to get
\(r=\num[round-mode=places,round-precision=3]{\measres}\).

Based on the assumed values of \(\bar n\), we obtain
\(v=\num[round-mode=places,round-precision=1]{\quadsq}\) for the
expectations of the square quadratures in Eqs.~\eqref{eq:gkp_p}
and~\eqref{eq:gkp_q} before accounting for displacements from logical
operations and noise. This value is obtained from the expected photon
number by applying Eq.~\ref{app:eq:quadsq-hatn}, assuming the
expectations of the quadrature are \(0\) and the variances of the two
quadratures are identical.  The expectations of the square
quadratures of the incoming GKP qubit and the ancillas each increase
by the added noise variance of \(2\sigma_{0}^{2}\), accounting for the
normalization used here.  The incoming GKP qubit may have been
displaced by a logical operation Pauli gate, which can be implemented
by displacements of magnitude\(\sqrt{\pi}\) of the canonical
variables.  This can add \(2\pi\) to the expectation of each square
quadratures of the incoming GKP qubit in Eqs.~\eqref{eq:gkp_q}
and~\eqref{eq:gkp_p}, since the mean of the quadratures is close to
zero before the operation.  By independence of the modes, we obtain
the bounds
\newcommand{\fnprimep}[2]{(3*(#1) - 2 + 3(#2)+4*pi)/4}%
\newcommand{\fnprimeq}[2]{(4*(#1) - 2 + 4(#2) + 2*pi)/4}
\edef\nprimep{\fpeval{\fnprimep{\quadsq}{\sigmaphys}}}%
\edef\nprimeq{\fpeval{\fnprimeq{\quadsq}{\sigmaphys}}}%
\begin{align}
  \langle\hat{n}'_{a}\rangle
  &\leq \frac{1}{4}\qty(3v - 2
    +3\sigma_{0}^{2}+ 4\pi)
    = \num[round-mode=places,round-precision=1]{\nprimep},
    \notag\\
  \langle\hat{n}'_{b}\rangle
  &\leq
    \frac{1}{4}\qty(4v-2
    +4\sigma_{0}^{2}+2\pi)
    )
    = \num[round-mode=places,round-precision=1]{\nprimeq}.
\end{align}

The bound on the  measurement infidelity for one quadrature measurement is expressed
in terms of the expected photon number, the resolution, and the LO
photon number \(\delta^{2}\) in Eq.~\eqref{eq:outstate_fid_gauss_sph}.
We require that the sum of the bounds on the quadrature measurement infidelities
is less than the allowance \(\epsilon_{m}\). This gives the inequality
\begin{align}
  \epsilon_{m} &\geq
                 \frac{\delta^{2}}{(2r)^{4}}\bra{\psi}\qty(\qty(4(2r)^{2} +\frac{40
    \delta^{2}}{3 (2r)^{2}})(\langle\hat{n}'_{q}\rangle+\langle\hat{n}'_{p}\rangle)
    + \frac{20\delta^{2}}{3(2r)^{2}}+3)\ket{\psi}.
\end{align}
Let \(\delta_{m}\) be the value of \(\delta\) that achieves equality
in the inequality above. Provided that \(\delta\leq\delta_{m}\), the
sum of the measurement infidelities is less than our target
\(\epsilon_{m}\).  To solve for \(\delta_{m}\), we set the resolution
and expected photon numbers according to the previously obtained
values.
The identity to be solved is
\newcommand{\fgkptwo}[3]{(4/((2*(#1))^2)*((#2)+(#3)) + 3/((2*(#1))^4))}%
\newcommand{\fgkpfour}[3]{((40/3)/((2*(#1))^6)*((#2)+(#3)+1/2))}%
\edef\gkpctwo{\fpeval{\fgkptwo{\measres}{\nprimep}{\nprimeq}}}%
\edef\gkpcfour{\fpeval{\fgkpfour{\measres}{\nprimep}{\nprimeq}}}%
\begin{align}
  \num[exponent-mode=scientific,round-mode=figures,round-precision=4]{\minfidbnd}
  &=
    \num[exponent-mode=scientific,round-mode=figures,round-precision=4]{\gkpctwo}\delta_{m}^{2}
    + \num[exponent-mode=scientific,round-mode=figures,round-precision=4]{\gkpcfour}\delta_{m}^{4}.
    \label{eq:gkpeqdeltasq}
\end{align}
\newcommand{\fgkpideltasq}[2]{(#1)/(#2)}%
\edef\gkpideltasq{\fpeval{\fgkpideltasq{\minfidbnd}{\gkpctwo}}}%
\newcommand{\fequivdetnoise}[2]{(#1)/sqrt((#2))}%
\edef\equivdetnoise{\fpeval{\fequivdetnoise{\measres}{\gkpideltasq}}}%
\newcommand{\fainfidbnd}[3]{(#1)*(#2) + (#3)*(#2)^2}%
\edef\ainfidbnd{\fpeval{\fainfidbnd{\gkpctwo}{\gkpideltasq}{\gkpcfour}}}%
The degree four term in \(\delta_{m}\) contributes little. We
neglect it for
\( \delta_{m}^{2} =
\num[exponent-mode=scientific,round-mode=figures,round-precision=2]{\gkpideltasq}
\) .  If we don't neglect the degree four term, substituting this
value of \(\delta_{m}^{2}\) results in a value of
\(\num[exponent-mode=scientific,round-mode=figures,round-precision=7]{\ainfidbnd}\)
for the right-hand sides of Eq.~\eqref{eq:gkpeqdeltasq}.  The value of
\(\delta_{m}^{2}\) found corresponds to a minimum required LO photon
number of
\newcommand{\fnlo}[1]{1/(#1)}%
\edef\nlo{\fpeval{\fnlo{\gkpideltasq}}}%
\(n_{\text{LO}}=\num[exponent-mode=scientific,round-mode=figures,round-precision=2]{\nlo}\).
In terms of equivalent photodiode added noise, the resolution and LO
amplitude corresponds to an average added noise of
\(\num[round-mode=places,round-precision=1]{\equivdetnoise}\)
electrons. The bounds remain satisfied if we increase the LO photon
number by decreasing \(\delta^{2}\) while keeping the resolution \(r\)
constant.  Since \(r=\delta\sigma\), increasing the LO photon number
by a factor of \(k\) corresponds to an increase in equivalent
photodiode added noise a factor of \(\sqrt{k}\). Thus, if the actual
added noise is larger than the one computed above, it suffices to
increase the LO photon number accordingly.  If the actual added noise
is smaller than the computed one, we can, conservatively, add noise to
the measurement outcomes to compensate and ensure the bounds.  This
mainly guarantees that there are no problems with conditional actions that depend on
fine details of the measurement outcomes, for example such actions
that depend discontinuously on the outcomes.
\newcommand{\fnlon}[3]{(#2) < (#1) ? (#3)*((#1)/(#2))^2 : (#3)}%

We conclude our analysis with Table~\ref{tab:gkpbounds} showing the parameters
and bounds for a few choices of GKP qubits and fidelity bounds. \MKc{Consider adding more rows to the table to show effects of some relevant changes.}
\begin{table}[h]
  \centering
  \caption{Table of required homodyne parameters for square GKP qubits
    and fidelity bounds. The GKP qubit parameters \(\bar n\) with
    corresponding squeezing in \(\mySI{db}\) and entanglement fidelity
    performances are from Fig.~\ref{fig:gkp_numerics}.  For each GKP
    qubit, we choose a combination of noise variance
    \(\sigma_{\text{noise}}\) and entanglement infidelity
    \(\epsilon_{\text{ec}}\) for which the analytic approximation
    according to the figure gives an entanglement infidelity close to
    and less than \(\epsilon_{\text{ec}}\). We then choose an
    allowance for the physical noise \(\sigma_{0}\) on the GKP qubit
    and the prepared ancillas. The noise variance \(\sigma_{0}^{2}\)
    must be less than \(\sigma_{\text{noise}}^{2}/3\).  From this we
    determine the required measurement resolution \(r\).  We then
    choose an allowance for the sum \(\epsilon_{m}\) of the
    measurement infidelities from the homodyne measurements.  Assuming
    no coherent contribution when chaining the infidelities, the
    overall infidelity is \(\epsilon_{\text{ec}}+\epsilon_{m}\).  From
    these parameters, we compute a minimum required LO photon number
    \(n_{\text{LO}}\), and the equivalent photodiode added noise
    \(\sigma_{e}\) in electrons at this LO photon number.  If the
    photodiode added noise is less than \(\lowdetnoise\) or
    \(\highdetnoise\)~\cite{hansen2001ultrasensitive,gerrits2011balanced},
    we give the required LO photon number for which the photodiode
    added noise is consistent with \(r\). Using an LO with this average photon number
    guarantees the calculated bounds on infidelity.}
  \vspace*{\baselineskip}
  \begin{tabular}{|c|c|c|c|c|c|c|c|c|c|c|}
    \hline
    \(\bar n\)
    &
      \(\mySI{db}\)
    &
      \(\sigma_{\text{noise}}\)
    &
      \(\sigma_{0}\)
    &
      \(\epsilon_{\text{ec}}\)
    &
      \(\epsilon_{m}\)
    &
      \(r\)
    &
      \(n_{\text{LO}}\)
    &
      \(\sigma_{e}\)
    &
      \shortstack[c]{\(n_{\text{LO}}\) at\\ \(\sigma_{e}=\lowdetnoise\)}
    &
      \shortstack[c]{\(n_{\text{LO}}\) at\\ \(\sigma_{e}=\highdetnoise\)}
      \\
    \hline\hline
    \(\gkpnbnd\)
    &
      \(\num[round-mode=places,round-precision=1]{\gkpsqz}\)
    &
      \(\num[round-mode=places,round-precision=3]{\sigmanoise}\)
    &
      \(\num[round-mode=places,round-precision=3]{\sigmaphys}\)
    &
      \(\num[round-mode=places,round-precision=3]{\infidbnd}\)
    &
      \(\num[round-mode=places,round-precision=3]{\fpeval{\minfidbnd}}\)
    &
      \(\num[round-mode=places,round-precision=3]{\measres}\)
    &
      \(\num[exponent-mode=scientific,round-mode=figures,round-precision=2]{\nlo}\)
    &
      \(\num[round-mode=places,round-precision=1]{\equivdetnoise}\)
    &
\(\num[exponent-mode=scientific,round-mode=figures,round-precision=2]{\fpeval{\fnlon{\lowdetnoise}{\equivdetnoise}{\nlo}}}\)
    &
\(\num[exponent-mode=scientific,round-mode=figures,round-precision=2]{\fpeval{\fnlon{\highdetnoise}{\equivdetnoise}{\nlo}}}\)
    \\
    \noalign{
    \gdef\gkpnbnd{7.6}%
    \gdef\sigmagkp{0.182}%
    \gdef\infidbnd{0.015}%
    \gdef\minfidbnd{0.005}%
    \gdef\sigmabnd{0.1}%
    \gdef\sigmaphys{0.05}%
    \xdef\gkpsqz{\fpeval{\fsqueeze{\sigmagkp}}}%
    \xdef\measres{\fpeval{\fmeasres{\sigmabnd}{\sigmaphys}}}%
    \xdef\sigmanoise{\fpeval{\fsigmanoise{\sigmaphys}{\measres}}}%
    \xdef\quadsq{\fpeval{\fquadsq{\gkpnbnd}}}%
    \xdef\nprimep{\fpeval{\fnprimep{\quadsq}{\sigmaphys}}}%
    \xdef\nprimeq{\fpeval{\fnprimeq{\quadsq}{\sigmaphys}}}%
    \xdef\gkpctwo{\fpeval{\fgkptwo{\measres}{\nprimep}{\nprimeq}}}%
    \xdef\gkpideltasq{\fpeval{\fgkpideltasq{\minfidbnd}{\gkpctwo}}}%
    \xdef\nlo{\fpeval{\fnlo{\gkpideltasq}}}%
    \xdef\gkpideltasq{\fpeval{\fgkpideltasq{\minfidbnd}{\gkpctwo}}}%
    \xdef\equivdetnoise{\fpeval{\fequivdetnoise{\measres}{\gkpideltasq}}}%
    }
    \hline
    \(\gkpnbnd\)
    &
      \(\num[round-mode=places,round-precision=1]{\gkpsqz}\)
    &
      \(\num[round-mode=places,round-precision=3]{\sigmanoise}\)
    &
      \(\num[round-mode=places,round-precision=3]{\sigmaphys}\)
    &
      \(\num[round-mode=places,round-precision=3]{\infidbnd}\)
    &
      \(\num[round-mode=places,round-precision=3]{\fpeval{\minfidbnd}}\)
    &
      \(\num[round-mode=places,round-precision=3]{\measres}\)
    &
      \(\num[exponent-mode=scientific,round-mode=figures,round-precision=2]{\nlo}\)
    &
      \(\num[round-mode=places,round-precision=1]{\equivdetnoise}\)
    &
\(\num[exponent-mode=scientific,round-mode=figures,round-precision=2]{\fpeval{\fnlon{\lowdetnoise}{\equivdetnoise}{\nlo}}}\)
    &
\(\num[exponent-mode=scientific,round-mode=figures,round-precision=2]{\fpeval{\fnlon{\highdetnoise}{\equivdetnoise}{\nlo}}}\)
    \\
    \noalign{
    \gdef\gkpnbnd{12.0}%
    \gdef\sigmagkp{0.144}%
    \gdef\infidbnd{0.002}%
    \gdef\minfidbnd{0.0005}%
    \gdef\sigmabnd{0.1}%
    \gdef\sigmaphys{0.05}%
    \xdef\gkpsqz{\fpeval{\fsqueeze{\sigmagkp}}}%
    \xdef\measres{\fpeval{\fmeasres{\sigmabnd}{\sigmaphys}}}%
    \xdef\sigmanoise{\fpeval{\fsigmanoise{\sigmaphys}{\measres}}}%
    \xdef\quadsq{\fpeval{\fquadsq{\gkpnbnd}}}%
    \xdef\nprimep{\fpeval{\fnprimep{\quadsq}{\sigmaphys}}}%
    \xdef\nprimeq{\fpeval{\fnprimeq{\quadsq}{\sigmaphys}}}%
    \xdef\gkpctwo{\fpeval{\fgkptwo{\measres}{\nprimep}{\nprimeq}}}%
    \xdef\gkpideltasq{\fpeval{\fgkpideltasq{\minfidbnd}{\gkpctwo}}}%
    \xdef\nlo{\fpeval{\fnlo{\gkpideltasq}}}%
    \xdef\gkpideltasq{\fpeval{\fgkpideltasq{\minfidbnd}{\gkpctwo}}}%
    \xdef\equivdetnoise{\fpeval{\fequivdetnoise{\measres}{\gkpideltasq}}}%
    }
    \hline
    \(\gkpnbnd\)
    &
      \(\num[round-mode=places,round-precision=1]{\gkpsqz}\)
    &
      \(\num[round-mode=places,round-precision=3]{\sigmanoise}\)
    &
      \(\num[round-mode=places,round-precision=3]{\sigmaphys}\)
    &
      \(\num[round-mode=places,round-precision=4]{\infidbnd}\)
    &
      \(\num[round-mode=places,round-precision=4]{\fpeval{\minfidbnd}}\)
    &
      \(\num[round-mode=places,round-precision=3]{\measres}\)
    &
      \(\num[exponent-mode=scientific,round-mode=figures,round-precision=2]{\nlo}\)
    &
      \(\num[round-mode=places,round-precision=1]{\equivdetnoise}\)
    &
\(\num[exponent-mode=scientific,round-mode=figures,round-precision=2]{\fpeval{\fnlon{\lowdetnoise}{\equivdetnoise}{\nlo}}}\)
    &
\(\num[exponent-mode=scientific,round-mode=figures,round-precision=2]{\fpeval{\fnlon{\highdetnoise}{\equivdetnoise}{\nlo}}}\)
    \\
    \noalign{
    \gdef\gkpnbnd{12.0}%
    \gdef\sigmagkp{0.144}%
    \gdef\infidbnd{0.008}%
    \gdef\minfidbnd{0.002}%
    \gdef\sigmabnd{0.18}%
    \gdef\sigmaphys{0.09}%
    \xdef\gkpsqz{\fpeval{\fsqueeze{\sigmagkp}}}%
    \xdef\measres{\fpeval{\fmeasres{\sigmabnd}{\sigmaphys}}}%
    \xdef\sigmanoise{\fpeval{\fsigmanoise{\sigmaphys}{\measres}}}%
    \xdef\quadsq{\fpeval{\fquadsq{\gkpnbnd}}}%
    \xdef\nprimep{\fpeval{\fnprimep{\quadsq}{\sigmaphys}}}%
    \xdef\nprimeq{\fpeval{\fnprimeq{\quadsq}{\sigmaphys}}}%
    \xdef\gkpctwo{\fpeval{\fgkptwo{\measres}{\nprimep}{\nprimeq}}}%
    \xdef\gkpideltasq{\fpeval{\fgkpideltasq{\minfidbnd}{\gkpctwo}}}%
    \xdef\nlo{\fpeval{\fnlo{\gkpideltasq}}}%
    \xdef\gkpideltasq{\fpeval{\fgkpideltasq{\minfidbnd}{\gkpctwo}}}%
    \xdef\equivdetnoise{\fpeval{\fequivdetnoise{\measres}{\gkpideltasq}}}%
    }
    \hline
    \(\gkpnbnd\)
    &
      \(\num[round-mode=places,round-precision=1]{\gkpsqz}\)
    &
      \(\num[round-mode=places,round-precision=3]{\sigmanoise}\)
    &
      \(\num[round-mode=places,round-precision=3]{\sigmaphys}\)
    &
      \(\num[round-mode=places,round-precision=4]{\infidbnd}\)
    &
      \(\num[round-mode=places,round-precision=4]{\fpeval{\minfidbnd}}\)
    &
      \(\num[round-mode=places,round-precision=3]{\measres}\)
    &
      \(\num[exponent-mode=scientific,round-mode=figures,round-precision=2]{\nlo}\)
    &
      \(\num[round-mode=places,round-precision=1]{\equivdetnoise}\)
    &
\(\num[exponent-mode=scientific,round-mode=figures,round-precision=2]{\fpeval{\fnlon{\lowdetnoise}{\equivdetnoise}{\nlo}}}\)
    &
\(\num[exponent-mode=scientific,round-mode=figures,round-precision=2]{\fpeval{\fnlon{\highdetnoise}{\equivdetnoise}{\nlo}}}\)
    \\
    \noalign{
    \gdef\gkpnbnd{12.0}%
    \gdef\sigmagkp{0.144}%
    \gdef\infidbnd{0.008}%
    \gdef\minfidbnd{0.002}%
    \gdef\sigmabnd{0.18}%
    \gdef\sigmaphys{0.045}%
    \xdef\gkpsqz{\fpeval{\fsqueeze{\sigmagkp}}}%
    \xdef\measres{\fpeval{\fmeasres{\sigmabnd}{\sigmaphys}}}%
    \xdef\sigmanoise{\fpeval{\fsigmanoise{\sigmaphys}{\measres}}}%
    \xdef\quadsq{\fpeval{\fquadsq{\gkpnbnd}}}%
    \xdef\nprimep{\fpeval{\fnprimep{\quadsq}{\sigmaphys}}}%
    \xdef\nprimeq{\fpeval{\fnprimeq{\quadsq}{\sigmaphys}}}%
    \xdef\gkpctwo{\fpeval{\fgkptwo{\measres}{\nprimep}{\nprimeq}}}%
    \xdef\gkpideltasq{\fpeval{\fgkpideltasq{\minfidbnd}{\gkpctwo}}}%
    \xdef\nlo{\fpeval{\fnlo{\gkpideltasq}}}%
    \xdef\gkpideltasq{\fpeval{\fgkpideltasq{\minfidbnd}{\gkpctwo}}}%
    \xdef\equivdetnoise{\fpeval{\fequivdetnoise{\measres}{\gkpideltasq}}}%
    }
    \hline
    \(\gkpnbnd\)
    &
      \(\num[round-mode=places,round-precision=1]{\gkpsqz}\)
    &
      \(\num[round-mode=places,round-precision=3]{\sigmanoise}\)
    &
      \(\num[round-mode=places,round-precision=3]{\sigmaphys}\)
    &
      \(\num[round-mode=places,round-precision=4]{\infidbnd}\)
    &
      \(\num[round-mode=places,round-precision=4]{\fpeval{\minfidbnd}}\)
    &
      \(\num[round-mode=places,round-precision=3]{\measres}\)
    &
      \(\num[exponent-mode=scientific,round-mode=figures,round-precision=2]{\nlo}\)
    &
      \(\num[round-mode=places,round-precision=1]{\equivdetnoise}\)
    &
\(\num[exponent-mode=scientific,round-mode=figures,round-precision=2]{\fpeval{\fnlon{\lowdetnoise}{\equivdetnoise}{\nlo}}}\)
    &
\(\num[exponent-mode=scientific,round-mode=figures,round-precision=2]{\fpeval{\fnlon{\highdetnoise}{\equivdetnoise}{\nlo}}}\)
    \\
    \noalign{
    \gdef\gkpnbnd{12.0}%
    \gdef\sigmagkp{0.144}%
    \gdef\infidbnd{0.005}%
    \gdef\minfidbnd{0.005}%
    \gdef\sigmabnd{0.15}%
    \gdef\sigmaphys{0.075}%
    \xdef\gkpsqz{\fpeval{\fsqueeze{\sigmagkp}}}%
    \xdef\measres{\fpeval{\fmeasres{\sigmabnd}{\sigmaphys}}}%
    \xdef\sigmanoise{\fpeval{\fsigmanoise{\sigmaphys}{\measres}}}%
    \xdef\quadsq{\fpeval{\fquadsq{\gkpnbnd}}}%
    \xdef\nprimep{\fpeval{\fnprimep{\quadsq}{\sigmaphys}}}%
    \xdef\nprimeq{\fpeval{\fnprimeq{\quadsq}{\sigmaphys}}}%
    \xdef\gkpctwo{\fpeval{\fgkptwo{\measres}{\nprimep}{\nprimeq}}}%
    \xdef\gkpideltasq{\fpeval{\fgkpideltasq{\minfidbnd}{\gkpctwo}}}%
    \xdef\nlo{\fpeval{\fnlo{\gkpideltasq}}}%
    \xdef\gkpideltasq{\fpeval{\fgkpideltasq{\minfidbnd}{\gkpctwo}}}%
    \xdef\equivdetnoise{\fpeval{\fequivdetnoise{\measres}{\gkpideltasq}}}%
    }
    \hline
    \(\gkpnbnd\)
    &
      \(\num[round-mode=places,round-precision=1]{\gkpsqz}\)
    &
      \(\num[round-mode=places,round-precision=3]{\sigmanoise}\)
    &
      \(\num[round-mode=places,round-precision=3]{\sigmaphys}\)
    &
      \(\num[round-mode=places,round-precision=4]{\infidbnd}\)
    &
      \(\num[round-mode=places,round-precision=4]{\fpeval{\minfidbnd}}\)
    &
      \(\num[round-mode=places,round-precision=3]{\measres}\)
    &
      \(\num[exponent-mode=scientific,round-mode=figures,round-precision=2]{\nlo}\)
    &
      \(\num[round-mode=places,round-precision=1]{\equivdetnoise}\)
    &
\(\num[exponent-mode=scientific,round-mode=figures,round-precision=2]{\fpeval{\fnlon{\lowdetnoise}{\equivdetnoise}{\nlo}}}\)
    &
\(\num[exponent-mode=scientific,round-mode=figures,round-precision=2]{\fpeval{\fnlon{\highdetnoise}{\equivdetnoise}{\nlo}}}\)
    \\
    \noalign{
    \gdef\gkpnbnd{15.4}%
    \gdef\sigmagkp{0.127}%
    \gdef\infidbnd{0.0008}%
    \gdef\minfidbnd{0.0002}%
    \gdef\sigmabnd{0.14}%
    \gdef\sigmaphys{0.05}%
    \xdef\gkpsqz{\fpeval{\fsqueeze{\sigmagkp}}}%
    \xdef\measres{\fpeval{\fmeasres{\sigmabnd}{\sigmaphys}}}%
    \xdef\sigmanoise{\fpeval{\fsigmanoise{\sigmaphys}{\measres}}}%
    \xdef\quadsq{\fpeval{\fquadsq{\gkpnbnd}}}%
    \xdef\nprimep{\fpeval{\fnprimep{\quadsq}{\sigmaphys}}}%
    \xdef\nprimeq{\fpeval{\fnprimeq{\quadsq}{\sigmaphys}}}%
    \xdef\gkpctwo{\fpeval{\fgkptwo{\measres}{\nprimep}{\nprimeq}}}%
    \xdef\gkpideltasq{\fpeval{\fgkpideltasq{\minfidbnd}{\gkpctwo}}}%
    \xdef\nlo{\fpeval{\fnlo{\gkpideltasq}}}%
    \xdef\gkpideltasq{\fpeval{\fgkpideltasq{\minfidbnd}{\gkpctwo}}}%
    \xdef\equivdetnoise{\fpeval{\fequivdetnoise{\measres}{\gkpideltasq}}}%
    }
    \hline
    \(\gkpnbnd\)
    &
      \(\num[round-mode=places,round-precision=1]{\gkpsqz}\)
    &
      \(\num[round-mode=places,round-precision=3]{\sigmanoise}\)
    &
      \(\num[round-mode=places,round-precision=3]{\sigmaphys}\)
    &
      \(\num[round-mode=places,round-precision=4]{\infidbnd}\)
    &
      \(\num[round-mode=places,round-precision=4]{\fpeval{\minfidbnd}}\)
    &
      \(\num[round-mode=places,round-precision=3]{\measres}\)
    &
      \(\num[exponent-mode=scientific,round-mode=figures,round-precision=2]{\nlo}\)
    &
      \(\num[round-mode=places,round-precision=1]{\equivdetnoise}\)
    &
\(\num[exponent-mode=scientific,round-mode=figures,round-precision=2]{\fpeval{\fnlon{\lowdetnoise}{\equivdetnoise}{\nlo}}}\)
    &
\(\num[exponent-mode=scientific,round-mode=figures,round-precision=2]{\fpeval{\fnlon{\highdetnoise}{\equivdetnoise}{\nlo}}}\)
    \\
    \hline
   \end{tabular}
  \label{tab:gkpbounds}
\end{table}

\pagebreak


\section{Acknowledgments}
We thank Karl Mayer for discussions and suggestions that helped us
improve this paper. We also thank Zachary Levine and Michael Mazurek
for assistance with reviewing the paper before submission.  AI
technology was used to assist with creating diagrams and in the final
editing stages for checking grammar, spelling and style. E.S. and
J.R.v.M. acknowledge support from the Professional Research Experience
Program (PREP) operated jointly by NIST and the University of Colorado
under the financial assistance Award No. 70NANB18H006 from the
U.S. Department of Commerce, National Institute of Standards and
Technology.  A.K. acknowledges support by the NSA-LPS Qubit
Collaboratory (LQC) National Quantum Fellowship administered by the
Oak Ridge Institute for Science and Education (ORISE) through an
interagency agreement between the U.S. Department of Energy (DOE) and
the Department of Defense (DOD). ORISE is managed by ORAU under DOE
contract number DE-SC0014664.  The use of trade names and software
does not imply endorsement by the US government, nor does it imply
these are necessarily the best available for the purpose used here.
The views and conclusions contained herein are those of the authors
and should not be interpreted as necessarily representing the official
policies or endorsements, either expressed or implied, of the DOD,
DOE, ORAU/ORISE, or the U.S. Government.  This work includes
contributions of the National Institute of Standards and Technology,
which are not subject to U.S. copyright.

\bibliography{Contents}

\appendix

\section{Functional analysis}
\label{app:pvm}

Hilbert spaces in this paper are assumed to be separable. We make frequent use of projection-valued measures and their representations. For definitions of projection-valued measures and their use in the spectral theorem and the Borel functional calculus for self-adjoint operators, see Ref.~\cite{ReedSimon}, Chpt. VIII. In this paper, the
domains of defined projection-valued measures are assumed to be standard Borel spaces. Functions that are introduced or defined are required to be Borel measurable. We do not explicitly specify the required Borel measurability in our results. All measures introduced are \(\sigma\)-finite by assumption. For the relevant measure theoretic
concepts see real analysis textbooks such as~\cite{royden1968real}.
When we write expressions such as
\(\cH=\cH_{A}\otimes \cH_{B}\) or \(\cH=\bigoplus_{k\in J}\cH_{k}\), the equality symbols imply an implicit isometry that we do not need to refer to explicitly in calculations.

We symbolize projection-valued measures by the expression \(dQ(x)\)
for \(x\) in a standard Borel space \(\cB\). The projection determined
by the projection-valued measure \(dQ(x)\) for the Borel subset \(X\)
of \(\cB\) is \(\int_{x\in X}d Q(x)\). For every state \(\ket{\psi}\),
\(\bra{\psi}d Q(x)\ket{\psi}\) is a bounded measure on \(\cB\). The
measure is defined by
\(\int_{x\in X}\bra{\psi}d Q(x)\ket{\psi} = \bra{\psi}\int_{x\in
  X}dQ(x) \ket{\psi}\). According to the spectral theorem in
projection-valued measure form, for every self-adjoint operator
\(\hat{A}\), its projection-valued measure \(dA(x)\) on the reals
satisfies that if \(h(x)\) a Borel function on the reals, then
\(h(\hat{A})=\int_{x\in rls}h(x)dA(x)\) is a self-adjoint operator
whose domain consists of the states \(\ket{\psi}\) for which the
measure \(|h(x)|^{2}\bra{\psi}d A(x)\ket{\psi}\) is bounded. For
\(h(x) = x\), \(h(\hat{A})=\hat{A}\).

Every projection-valued measure \(dQ(x)\) on a standard Borel space
\(\cY\) has a direct-integral representation. The direct-integral
representation simplifies computations and helps to define conditional
unitary operations. The theory of direct integrals is described in
many functional analysis books. For example, see
Ref.~\cite{kadison:qf1997a} Chpt. 14 and
Ref.~\cite{takesaki1979theory} Chpt. IV.  Let \(\cH_{k}\) be the
Hilbert space of dimension \(k\) for \(k=1,\ldots,\infty\), where
\(\cH_{\infty}\) is the infinite-dimensional separable Hilbert
space. According to the direct-integral representation, there is a
\(\sigma\)-finite measure \(d\mu(x)\) on \(\cY\) and a measurable
partition \(\cY=\bigcup_{k=1}^{\infty}\cY_{k}\) such that the Hilbert
space \(\cH\) of \(dQ(x)\) can be decomposed as
\(\cH=\bigoplus_{k=1}^{\infty}L^{2}(\cY_{k},d\mu(x)|_{\cY_{k}},\cH_{k})\),
and the projector onto \(L^{2}(\cY_{k},d\mu(x)|_{\cY_{k}},\cH_{k})\)
in this decomposition is \(\int_{x\in\cY_{k}}dQ(x)\). We write
\(\cH(x)=\cH_{k}\) for \(x\in\cY_{k}\).  The Hilbert space
\(L^{2}(\cY_{k},d\mu(x)|_{\cY_{k}},\cH_{k})\) is the Hilbert space of
functions from \(\cY_{k}\) to \(\cH_{k}\) that are
\(\cY_{k}\)-measurable and square integrable with respect to
\(d\mu(x)\) restricted to \(\cY_{k}\). Such functions are unique up to
null-sets of \(d\mu(x)\).  This Hilbert space is naturally isomorphic
to \(L^{2}(\cY_{k},d\mu(x)|_{\cY_{k}})\otimes \cH_{k}\), where
\(L^{2}(\cY_{k},d\mu(x)|_{\cY_{k}})\) is the Hilbert space of
\(\cY_{k}\)-measurable functions from \(\cY_{k}\) to the complex
numbers.  \(L^{2}(\cY_{k},d\mu(x)|_{\cY_{k}})\) has the natural
projection-valued measure \(d\Pi_{k}(x)\) defined by
\(\int_{X}d\Pi_{k}(x)\ket{g(\cdot)} = \ket{\chi_{X}(\cdot)g(\cdot)}\),
where \(\chi_{X}(x)\) is the characteristic function of the set \(X\).
The projection valued measure \(d Q(x)\) restricted to \(\cY_{k}\) can
then be written as \(d\Pi_{k}(x)\otimes\one\) in the factorization
\(L^{2}(\cY_{k},d\mu(x)|_{\cY_{k}})\otimes \cH_{k}\).  We formally
write \(\cH=\int_{x\in\cY}^{\oplus}\cH(x)\), where \(\cH(x)=\cH_{k}\)
for \(x\in\cY_{k}\).  Here the \(\oplus\) in the superscript specifies
that this is a direct integral, not a conventional
integral. Accordingly, states \(\ket{\psi}\) of \(\cH\) can be written
formally in the form
\(\ket{\psi}=\int_{x\in\cY}^{\oplus} d\mu(x) \ket{\psi(x)}\), where
for \(x\in\cY_{k}\), \(\ket{\psi(x)}\in\cH_{k}\).  Restricted to
\(\cY_{k}\), \(\ket{\psi(x)}\) is a measurable function from
\(\cY_{k}\) to \(\cH_{k}\) that is square-integrable on
\(\cY_{k}\) with respect to \(d\mu(x)|_{\cY_{k}}\).
This means that for all \(\ket{\phi}\) in \(\cH_{k}\),
\(\braket{\phi}{\psi(x)}\) is measurable on \(\cY_{k}\) and
\(\sum_{k=1}^{\infty}\int_{x\in\cY_{k}}
d\mu(x)\braket{\psi(x)}<\infty\).  In the complete direct integral
expression
\(\ket{\psi}=\int_{x\in\cY}^{\oplus} d\mu(x) \ket{\psi(x)}\), the
superscript \(\oplus\) on the integral indicates that this is to be
interpreted as a vector-valued measure.  Informally, the expression
\(d\mu(x)\) in the integral is a kind of square-root measure. It is
determined up to equivalence of measures with respect to absolute
continuity, and change of measure is accompanied by the square root of
the Radon-Nikodym derivative. Decomposable operators are bounded
operators \(\hat{A}\) on \(\cH\) that can be written in the form
\(\hat{A}=\int_{x\in\cY} dQ(x) \hat{A}(x)\), where for
\(x\in\cY_{k}\), \(\hat{A}(x)\) is an operator on \(\cH_{k}\), and the
function \(\hat{A}(x)\) satisfies certain measurability conditions for
\(\hat{A}\) to be well-defined. Intuitively, \(\hat{A}\) is an
operator that acts conditionally on \(x\in\cY\).  Decomposable
operators are characterized as operators that are in the commutant of
the von-Neumann algebra of operators of the form
\(\int_{x\in\cY}dQ(x) h(x)\) for bounded Borel functions \(h(x)\).

If \(\cH\) is a tensor factor of a larger Hilbert space
\(\cH_{T}=\cH_{A}\otimes\cH\), then the direct integral tensors with
\(\cH_{A}\) as
\begin{align}
  \cH_{T}&=\bigoplus_{k=1}^{\infty}\cH_{A}\otimes L^{2}(\cY_{k},d\mu(k)|_{\cY_{k}},\cH_{k})
             \notag\\
    &=\bigoplus_{k=1}^{\infty}L^{2}(\cY_{k},d\mu(k)|_{\cY_{k}},\cH_{k}\otimes\cH_{A}),
\end{align}
We can treat \(dQ(x)\) as a projection
valued measure of \(\cH_{T}\) in the usual way and write product
states \(\ket{\phi}_{A}\otimes\ket{\psi}\) of \(\cH_{T}\) in the form
\(\int_{x\in\cY}^{\oplus}d\mu(x)\ket{\phi}_{A}\otimes \ket{\psi(x)}\).

We make use of bounded complex measures on \(\rls\).  The definitions
and properties of signed and complex measures can be found in standard
books on real analysis such as Ref.~\cite{royden1968real}, Chpt. 11
Sect. 5 and Chpt. 11 Sect. 6, Prob. 36.  For our purposes, a bounded
complex measure \(d\tilde\mu(\kappa)\) on the reals is defined by
\(d\tilde\mu(\kappa)=e^{i\theta(\kappa)}d\mu(\kappa)\), where
\(\theta(\kappa)\) is a measurable real function and
\(d\mu(\kappa)\)is a bounded positive measure. When introducing
bounded complex measures, we denote the complex measure with a tilde,
such as \(d\tilde\mu(\kappa)\) and denote the associated positive
measure without the tilde, such as \(d\mu(\kappa)\). The measure
\(d\mu(\kappa)\) can be thought of as the absolute value of the
measure \(d\tilde\mu(\kappa)\).

For a function \(\tilde f(\kappa)\) on the reals in a function space
for which Fourier transforms can be defined, we use the convention
that its Fourier transform is
\(f(x)=\int_{\kappa\in\rls}e^{i\kappa x}\tilde f(\kappa)d\kappa\).
The inverse Fourier transform of \(f(x)\) is
\(\tilde f(\kappa) = \frac{1}{2\pi}\int_{x\in\rls}e^{-i\kappa
  x}f(x)dx\). With this definition, the convolution
\((\tilde f *\tilde g)(\kappa)\) of \(\tilde f(\kappa)\) and
\(\tilde g(\kappa)\) has Fourier transform \(f(x)g(x)\). The
\(L_{2}\)-norm of \(f(x)\) is
\(\|f(x)\|_{2}= \sqrt{2\pi} \| \tilde f(\kappa)\|_{2}\).

The Fourier transform
\(f(x)=\int_{\kappa\in\rls} e^{i\kappa x} d\tilde\mu(\kappa)\) of a
bounded complex measure is a continuous, bounded function. Functions
such as \(f(x)=e^{ix\kappa}\) and functions with integrable Fourier
transforms are the Fourier transforms of bounded complex measures.

We claimed that for observables with continuous spectrum, exact
measurements do not exist. For the measurement of \(\hat{q}_{0}\), the
final state of the system and apparatus after a putative exact
measurement would be
\begin{align}
  \ket{\phi_{0}} &= \int_{x} d\Pi_{0}(x)\ket{\psi}\otimes\ket{x}_{M},
\end{align}
where $\ket{\psi}$ is the state of the measured system and $\ket{x}_M$
is the state of the apparatus recording the measurement result, $x$.
Consider the case where the system and apparatus Hilbert spaces are
\(L^{2}(\rls)\).  Then the final state ought to be in
\(L^{2}(\rls^{2})\). However, the ideal measurement ought to result in
a state that is supported on the diagonal \(\{(x,x): x\in\rls\}\) of
\(\rls^{2}\). But since the diagonal has measure zero, no such state
exists. Similar problems occur when attempting to define the
classical-quantum postmeasurement state. These problems are associated
with the fact that for a given value of \(x\), \(\Pi_{0}(x)\) cannot
be interpreted as a bounded or unbounded operator. The expression
\(\Pi_{0}(x)\ket{\psi}\) ought to be a state whose wavefunction is
supported on \(\{x\}\) but no such state exists in the Hilbert space.

\section{Often used relationships}
\label{app:inequalities}
\label{app:fidelities}

We use the following basic inequalities:
\begin{align}
  \Re\qty(e^{i\theta})
  & = \cos(\theta) \geq 1- {\theta^{2}}/{2},
    \label{eq:reeith-ineq}\\
  \sqrt{|x|+|y|}
  &\leq \sqrt{|x|} + \sqrt{|y|},
    \label{eq:sqrtx+y-ineq}\\
  (x+y)^{2}
  &\leq 2(x^{2}+y^{2}),
    \label{eq:x+ysqrt-ineq}            
\end{align}
where in the last inequality, \(x\) and \(y\) are real.

For bounded operators \(\hat{B}\) we define
\(\Re\qty(\hat{B}) = \hat{B}+\hat{B}^{\dagger}\).  For a self-adjoint
operator \(\hat{A}\), theinequality Eq.~\eqref{eq:reeith-ineq} implies
the operator inequality
\begin{align}
  \Re\qty(e^{i\hat{A}}) \geq 1-\hat{A}^{2}/2,
  \label{eq:reeith-ineq-op}
\end{align}
where the right-hand side has expectation \(-\infty\)  for states not in the domain
of \(\hat{A}\).

For self-adjoint operators \(\hat A\) and \(\hat B\) with \(\hat{B}\)
bounded, the operator \(\hat{A}+\hat{B}\) is self-adjoint with the
same domain as \(\hat{A}\). This is a consequence of the Kato-Rellich
theorem, see Ref.~\cite{reed1975ii} Thm. X.12.  For a state
\(\ket{\psi}\) in the domain of \(\hat{A}\), we have
\begin{align}
  |\langle \hat{A}\rangle - \langle \hat{B}\rangle|^{2}
  &=
    |\bra{\psi} \hat{A}-\hat{B}\ket{\psi}|^{2}
    \notag\\
  &= \bra{\psi} (\hat{A}-\hat{B})^{\dagger}
    \ketbra{\psi} (\hat{A}-\hat{B})\ket{\psi}
    \notag\\
  &\leq \bra{\psi}
    (\hat{A}-\hat{B})^{\dagger}(\hat{A}-\hat{B})\ket{\psi}
    \notag\\
  &= |(\hat{A}-\hat{B})\ket{\psi}|^{2}.
    \label{app:eq:a-bexp}
\end{align}

Let \(\hat{p}\) be a normalized quadrature of a mode.  Let \(\hat{q}\)
be its normalized conjugate quadrature in the mode so that
\([\hat{q},\hat{p}] = 2i\) with our normalization conventions. Then
\begin{align}
  \hat{p}^{2} \leq \hat{p}^{2}+\hat{q}^{2} = 4(\hat{n}+1/2).
  \label{app:eq:quadsq-hatn}
\end{align}

We take advantage of inequalities relating fidelity to overlaps and
distances.  If  \(\Re(\braket{\phi_{1}}{\phi_{2}}) \geq 1-\epsilon\), then
\begin{align}
  F(\ketbra{\phi_{1}},\ketbra{\phi_{2}})
  &= |\braket{\phi_{1}}{\phi_{2}}|^{2}
    \notag\\
  & = \Re(\braket{\phi_{1}}{\phi_{2}})^{2}
    + \Im(\braket{\phi_{1}}{\phi_{2}})^{2}
    \notag\\
  & \geq \Re(\braket{\phi_{1}}{\phi_{2}})^{2}
    \notag\\
  & = (1-\epsilon)^{2}
    \notag\\
  &\geq 1-2\epsilon.
    \label{eq:fidfromoverlap}
\end{align}
If \(|\ket{\phi_{1}}-\ket{\phi_{2}}|^{2}\leq \epsilon\), then from
\(|\ket{\phi_{1}}-\ket{\phi_{2}}|^{2} =
2-2\Re(\braket{\phi_{1}}{\phi_{2}})\) we obtain
\(\Re(\braket{\phi_{1}}{\phi_{2}}) \geq 1-\epsilon/2\) and from
Eq.~\eqref{eq:fidfromoverlap},
\(F(\ketbra{\phi_{1}},\ketbra{\phi_{2}}) \geq 1-\epsilon\).
Conversely, suppose that \(F(\ketbra{\phi_{1}},\ketbra{\phi_{2}}) \geq 1-\epsilon\).
There exists a phase factor \(e^{i\phi}\) such that
\(\bra{\phi_{1}}e^{i\phi}\ket{\phi_{2}}\) is real and positive and
\(|\ket{\phi_{1}}-e^{i\phi}\ket{\phi_{2}}|^{2} = 2-2|\braket{\phi_{1}}{\phi_{2}}|
=2\qty(1-\sqrt{F(\ketbra{\phi_{1}},\ketbra{\phi_{2}})})\leq 2\qty(1-\sqrt{1-\epsilon})\leq 2\epsilon\).

\section{Bounds on the distance between states coupled with \(\hat{q}_{\delta}\) and \(\hat{q}\)}
\label{app:appendix3}

For convenience, we omit the tensor
product symbol \(\otimes\) and use the convention that operators and
states labeled with \(M\) belong to the apparatus system.
We wish to determine a bound on
\begin{align}
  |e^{-i\hat{q}_{\delta}\hat{s}_{M}}\ket{\psi} -
  e^{-i\theta}e^{-i\hat{q}\hat{s}_{M}}\ket{\psi}|^{2}
  &= 2-2\Re(e^{-i\theta}\bra{\psi}e^{i\hat{q}_{\delta}\hat{s}_{M}}
    e^{-i\hat{q}\hat{s}_{M}}\ket{\psi})
    \label{eq:measbnd1}
\end{align}
for states \(\ket{\psi}\) that are vacuum on the LO modes, where we
can minimize the bound over real \(\theta\) since we do not
distinguish states with different global phases.  Because the action on \(M\)
is diagonal with respect to the projection valued measure
\(d\Pi(s)_{M}\) of \(\hat{s}_{M}\), we can write
\begin{align}
  e^{i\hat{q}_{\delta}\hat{s}_{M}}e^{-i\hat{q}\hat{s}_{M}}
  &= \int_{s}d\Pi(s)_{M}  \qty(e^{i\hat{q}_{\delta}\hat{s}_{M}}e^{-i\hat{q}\hat{s}_{M}})
    \notag\\
  &=\int_{s}d\Pi(s)_{M}\qty(e^{i\hat{q}_{\delta}s}e^{-i\hat{q}s})
    \label{eq:couplexps}
\end{align}
and substitute in Eq.~\eqref{eq:measbnd1}. For an upper bound on the
left-hand side of this identity, it suffices to determine a lower
bound on
\(\Re\qty(\bra{\psi}e^{i\hat{q}_{\delta}s}e^{-i\hat{q}s}\ket{\psi})\)
as a function of \(s\) and integrate.  For this purpose, we first
express \(e^{i\hat{q}_{\delta}s}e^{-i\hat{q}s}\) as a convenient
conjugation by displacements.  The terms for different mode pairs
\(a_{k}, b_{k}\) in \(\hat{q}_{\delta}\) and \(\hat{q}\) commute. 
In consideration of Eq.~\eqref{eq:bbpdiffop}, the definition of \(\delta\) and the relationship between \(\bm{\alpha}\)
and \(\bm{\beta}\), we can write
\begin{align}
  e^{i\hat{q}_{\delta}s}e^{-i\hat{q}s} &=
  \prod_{k=1}^{N}e^{i( \omega_{k} ( \hat{a}_{k}^{\dagger} \beta_{k} + \beta_{k}^{*} \hat{a}_{k} ) + \delta \omega_{k} ( \hat{a}_{k}^{\dagger} \hat{b}_{k} + \hat{b}_{k}^{\dagger} \hat{a}_{k} ))s}
  e^{-i( \omega_{k} ( \hat{a}_{k}^{\dagger} \beta_{k} + \beta_{k}^{*} \hat{a}_{k} ))s}.
    \label{eq:logbymode}
\end{align}
For the next steps, fix \(k\in\{1,\ldots,N\}\) and define
\begin{align}
  \hat{q}_{a}&=
               \hat{a}_{k}^{\dagger} \beta_{k} + \beta_{k}^{*} \hat{a}_{k},
               \notag\\
  \hat{q}_{b} &=
                -i(\hat{b}_{k}^{\dagger} \beta_{k} - \beta_{k}^{*} \hat{b}_{k}),
                \notag\\
  \hat{B} &=\hat{a}_{k}^{\dagger} \hat{b}_{k} + \hat{b}_{k}^{\dagger} \hat{a}_{k} .
            \label{eq:genericmodealgebra1}
\end{align}
Let \(t=\omega_{k}s\) so that the mode \(k\) factor of
Eq.~\eqref{eq:logbymode} can be written as
\begin{align}
  \hat{F}_{k}=e^{it(\hat{q}_{a}+\delta \hat{B})}e^{-it\hat{q}_{a}}.
\end{align}
We have
the following commutation relationships between \(\hat{q}_{a}\),
\(\hat{q}_{b}\) and \(\hat{B}\)
\begin{align}
  [\hat{q}_{a},\hat{q}_{b}] &=0
                              ,\notag\\
  [\hat{q}_{a},\hat{B}] &= -i \hat{q}_{b}
                          ,\notag\\
  [\hat{q}_{b},\hat{B}] &=i\hat{q}_{a}
                          , 
\end{align}
which implies that  the three operators generate an affine extension
of the circle group with \(\hat{q}_{a}\) and \(\hat{q}_{b}\) generating
displacements and \(\hat{B}\) generating rotations.
Let \(\vec{q}=(\hat{q}_{a},\hat{q}_{b})^{T}\), and
define the rotation matrix \(R_{\phi} =
\begin{pmatrix}
  \cos(\phi)& -\sin(\phi)\\
  \sin(\phi)& \cos(\phi)
\end{pmatrix}
\) and the matrix
\(S =
\begin{pmatrix}
  0&1\\-1&0
\end{pmatrix}
\).

\Pc{
  \newcommand{\Adj}{{\rm adj.}}
  A representation \(\pi\) of the generated group as affine transformations that is faithful
  for the adjoint representation of the unitary operators generated by
  \(\hat{B}\), \(\hat{q}_{a}\) and \(\hat{q}_{b}\) can be defined by
  \begin{align}
    \pi\qty(\Adj e^{-i\vec{u}^{T}\vec{q}}) \vec{x} &= \vec{x}+\vec{u}
                                 \notag\\
    \pi\qty(\Adj e^{-i\phi \hat{B}})\vec{x} &= R_{\phi}\vec{x},
                                   \notag
  \end{align}
  for vectors \(\vec{x}\) and \(\vec{r}\) in \(\rls^{2}\), where
  \(\Adj U = \rho \mapsto U \rho U^{\dagger}\) denotes the adjoint
  action of \(U\).  As \(3\times 3\) matrices acting on \(\cmplxs^{3}\) 
  the group elements and their generators can be represented by
  \begin{align}
    \pi(u^{T}\vec{q}) &=
                        \begin{pmatrix}
                          0&0&iu_{1}\\
                          0&0&iu_{2}\\
                          0&0&0
                        \end{pmatrix}
                               ,
                             &
                               \pi\qty(\Adj e^{-i u^{T}\vec{q}}) &=
                                                              \begin{pmatrix}
                                                                1&0&u_{1}\\
                                                                0&1&u_{2}\\
                                                                0&0&1
                                                              \end{pmatrix}
                                                                     ,\notag\\
    \pi(\hat{B}) &=
                        \begin{pmatrix}
                          0&-i&0\\
                          i&0&0\\
                          0&0&0
                        \end{pmatrix}
                               ,
                             &
   \pi\qty( \Adj e^{-i\phi\hat{B}}) &=
                  \begin{pmatrix}
                    \cos(\phi)&-\sin(\phi)&0\\
                    \sin(\phi)&\cos(\phi)&0\\
                    0&0&1
                  \end{pmatrix}
                         .
                         \notag
  \end{align}
  The representation relates to the adjoint representation with the commutation
  relationships of Eq.~\eqref{eq:commutationofbq} given below according to
  \begin{align}
    e^{-i\vec{u}^{T}\vec{q}}(\vec{x}^{T}S^T\vec{q} + \hat{B})e^{i\vec{u}^{T\vec{q}}}
    &= \vec{x}^{T}S^{T}\vec{q}+\hat{B}-\vec{u}^{T}S\vec{q}
      \notag\\
    &= (\vec{x}+\vec{u})^{T}S^T\vec{q} + \hat{B},
      \notag\\
    e^{-i\phi B}(\vec{x}^{T}S^T\vec{q} + \hat{B})e^{i\phi B}
    &=\vec{x}^{T}S^TR_{-\phi}\vec{q}+\hat{B}
      \notag\\
    &=(xR_{\phi})^{T}S^T\vec{q}+\hat{B},
      \notag
  \end{align}
  because \(S^{T}=-S\) and
  \( \vec{x}^{T}S^TR_{-\phi} = \vec{x}^{T}R_{-\phi}S^T =
  (R_{\phi}\vec{x})^{T}S^{T}\).  }

The commutation
relationships imply conjugation relationships
\begin{align}
  e^{-i\hat{B}\phi}\vec{q}e^{i\hat{B}\phi}
  &= R_{-\phi}\vec{q},
    \notag\\
  e^{-iu^{T}\vec{q}}\hat{B}e^{i\vec{u}^{T}\vec{q}}
  &=\hat{B} - \vec{u}^{T}S\vec{q}
    \label{eq:commutationofbq}
\end{align}
for two-dimensional, real row vectors \(\vec{u}\).
Consequently
\begin{align}
  e^{-i\vec{u}^{T}\vec{q}} e^{i\phi \hat{B}} e^{i\vec{u}^{T}\vec{q}}
  &=e^{i\phi (B- \vec{u}^{T}S\vec{q}))}
    \label{eq:expqb1}\\
  &=e^{i\phi \hat{B}}e^{i\vec{u}^{T}(1- R_{-\phi})\vec{q}},
\end{align}
where the second line is obtained from the left-hand side of the first
by moving \(e^{i\phi \hat{B}}\) to the left and merging the two
commuting exponentials that result on the right.  By solving for
\(\vec{u}^{T}\) in the identity
\(\vec{v}^{T}=\vec{u}^{T}(1-R_{-\phi})\), we can rewrite the identity
between the two operators on the right-hand sides as
\begin{align}
  e^{i\phi \hat{B}} e^{i \vec{v}^{T}\vec{q}}
  &= e^{i\phi (\hat{B} - \vec{v}^{T}(1-R_{-\phi})^{-1}S\vec{q})},
    \label{eq:expqb2}
\end{align}
provided that \(R_{-\phi}\ne \one\).  We modify \(\hat{F}_{k}\) by
replacing \(\hat{q}_{a}\) with \(\vec{w}^{T}\vec{q}\) and compute
\begin{align}
  e^{it( \vec{w}^{T}\vec{q}+\delta\hat{B})}e^{-it \vec{w}^{T}\vec{q}}
    &=
    e^{-i \vec{w}^{T}S\vec{q}/\delta}e^{it\delta\hat{B}}e^{-it\vec{w}^{T}\vec{q}}e^{i\vec{w}^{T}S\vec{q}/\delta}.
    \label{eq:simplerotationandshifts1}
\end{align}
For this computation, we applied Eq.~\eqref{eq:expqb1} to the first
factor, where we replaced \(\vec{u}^{T}\) by \(\vec{w}^{T}S/\delta\)
and \(\phi\) by \(\delta\), and exchanged commuting terms.  Let
\(\phi=[t\delta]_{\pi} = ((t\delta+\pi)\mod(2\pi)) - \pi\) be the
angle \(t\delta\) expressed in the interval \([-\pi,\pi]\) modulo
\(2\pi\).  The term \(e^{it\delta\hat{B}}\) is \(2\pi\) periodic in
\(t\delta\) as an operator, so
\(e^{it\delta\hat{B}}= e^{i\phi\hat{B}}\). This is because \(\hat{B}\)
is the beam-splitter Hamiltonian and can be written as the difference
of two number operators of orthogonal modes with annihilation
operators \((\hat{a}_{k}\pm\hat{b}_{k})/\sqrt{2}\).  Therefore
\begin{align}
  e^{it( \vec{w}^{T}\vec{q}+\delta\hat{B})}e^{-it \vec{w}^{T}\vec{q}}
  &=
    e^{-i \vec{w}^{T}S\vec{q}/\delta}e^{i\phi\hat{B}}e^{-it\vec{w}^{T}\vec{q}}e^{i\vec{w}^{T}S\vec{q}/\delta}
    \notag\\
  &=
  e^{-i \vec{w}^{T}S\vec{q}/\delta}e^{i\phi(\hat{B}+
  t\vec{w}^{T}(1-R_{-\phi})^{-1}S\hat{q})}e^{i\vec{w}^{T}S\vec{q}/\delta}
    \notag\\
  &=
    e^{-i \vec{w}^{T}S\vec{q}/\delta}e^{it\vec{w}^{T}(1-R_{-\phi})^{-1}\vec{q}}e^{i\phi\hat{B}}
    e^{-it\vec{w}^{T}(1-R_{-\phi})^{-1}\vec{q}}e^{i\vec{w}^{T}S\vec{q}/\delta}.
    \label{eq:shiftedrotation1}
\end{align}
To obtain the second line, we applied Eq.~\eqref{eq:expqb2} to the middle two terms.
For the third line, we applied Eq.~\eqref{eq:expqb1} to the middle term.
The matrix \((1-R_{-\phi})^{-1}\) evaluates to 
\begin{align}
  \begin{pmatrix}
    1-\cos(\phi) & -\sin(\phi)\\
    \sin(\phi)& 1-\cos(\phi)
  \end{pmatrix}^{-1}
               &=
                 \frac{1}{(1-\cos(\phi))^{2}+\sin(\phi)^{2}}
                 \begin{pmatrix}
                   1-\cos(\phi) & \sin(\phi)\\
                   -\sin(\phi)&1-\cos(\phi)
                 \end{pmatrix}
                               \notag\\
                 &=
                   \frac{1}{2(1-\cos(\phi))}
                 \begin{pmatrix}
                   1-\cos(\phi) & \sin(\phi)\\
                   -\sin(\phi)&1-\cos(\phi)
                 \end{pmatrix}
                               \notag\\
                 &=
                   \frac{1}{2}
                   \begin{pmatrix}
                     1& \sin(\phi)/(1-\cos(\phi))\\
                     -\sin(\phi)/(1-\cos(\phi)) &1
                   \end{pmatrix}
                                                  \notag\\
                 &=
                   \frac{1}{2}
                   \begin{pmatrix}
                     1& \cot(\phi/2)\\
                     -\cot(\phi/2) &1
                   \end{pmatrix},
\end{align}
where we used the trigonometric identities
\(\sin(\phi)=2\sin(\phi/2)\cos(\phi/2)\) and
\(1-\cos(\phi) = 2\sin(\phi/2)^{2}\) for the last line.
We substitute \(\vec{w}=(1,0)\) in Eq.~\eqref{eq:shiftedrotation1}
and write 
\begin{align}
  e^{it( \hat{q}_{a}+\delta\hat{B})}e^{-it \hat{q}_{a}}
  &=
    e^{-i (\hat{q}_{b}/\delta-t(\hat{q}_{a}+\cot(\phi/2)\hat{q}_{b})/2)}e^{i\phi\hat{B}}
    e^{i(\hat{q}_{b}/\delta-t(\hat{q}_{a}+\cot(\phi/2)\hat{q}_{b})/2)}
    \notag\\
  &=
    e^{i t\qty(\frac{1}{2}\hat{q}_{a}
    - \qty(\frac{1}{\delta t}-\frac{\cot(\phi/2)}{2})\hat{q}_{b})}
    e^{i\phi\hat{B}}
    e^{-i t\qty(\frac{1}{2}\hat{q}_{a}
    - \qty(\frac{1}{\delta t}-\frac{\cot(\phi/2)}{2})\hat{q}_{b})}
    \notag\\
  &=
    e^{i t\qty(\frac{1}{2}\hat{q}_{a})}
     e^{-it\qty(\frac{1}{\delta t}-\frac{\cot(\phi/2)}{2})\hat{q}_{b}}
    e^{i\phi\hat{B}}
     e^{it\qty(\frac{1}{\delta t}-\frac{\cot(\phi/2)}{2})\hat{q}_{b}}
    e^{-i t\qty(\frac{1}{2}\hat{q}_{a})}
    \notag\\
  &=
    e^{i t\qty(\frac{1}{2}\hat{q}_{a})}
    e^{i\qty(\phi\hat{B}+t\qty(\frac{\phi}{t\delta}-\frac{\phi}{2}\cot(\phi/2))\hat{q}_{a})}
    e^{-i t\qty(\frac{1}{2}\hat{q}_{a})}.
    \label{eq:permodeops}
\end{align}

Before we continue, we evaluate \(g(\phi) = (\phi/2) \cot(\phi/2)\) for
\(\phi\in [-\pi,\pi]\).  From the double angle formula for the cotangent,
\((\phi/2)\cot(\phi/2) = (\phi/4)\cot(\phi/4) -(\phi/4) \tan(\phi/4)\). Recursively
applying the double angle formula to the cotangent on the right-hand
side gives
\begin{align}
  g(\phi) &= (\phi/2^{k})\cot(\phi/2^{k}) - \sum_{j=2}^{k}(\phi/2^{j})\tan(\phi/2^{j}).
\end{align}
The limit as \(y\) goes to zero of \(y\cot(y)\) is \(1\), so letting \(k\rightarrow\infty\)
we get
\begin{align}
  g(\phi) &= 1-\phi \sum_{j=2}^{\infty}\tan(\phi/2^{j})/2^{j}.
\end{align}
For \(z\in[0,\pi/4]\), we have \(z\leq \tan(z) \leq (4/\pi)z\). Therefore, for \(\phi\in [0,\pi]\)
\begin{align}
  \frac{1}{12} \phi = \sum_{j=2}^{\infty}\phi \,2^{-2j} &\leq
  \sum_{j=2}^{\infty}\tan(\phi/2^{j})/2^{j}
  \leq \frac{4}{\pi}\phi\sum_{j=2}^{\infty}2^{-2j}
    =\frac{1}{3\pi} \phi \leq \frac{1}{9}\phi.
\end{align}
Since \(g(\phi)\) is symmetric in \(\phi\), we deduce that for \(\phi\in[-\pi,\pi]\),
\begin{align}
  g(\phi) &\in 1-\phi^{2}\qty[\frac{1}{12},\frac{1}{3\pi}]
            \subseteq 1-\phi^{2}\qty[\frac{1}{12},\frac{1}{9}].
\end{align}

We return to Eq.~\eqref{eq:logbymode}. For this purpose we make the
mode index explicit by writing \(\hat{q}_{a,k}\), \(\hat{q}_{b,k}\),
\(\hat{B}_{k}\), \(\phi_{k}\) for \(\hat{q}_{a}\), \(\hat{q}_{b}\),
\(\hat{B}\), \(\phi\) defined above.  We also write
\(\hat{q}_{k}=\omega_{k}\hat{q}_{a,k}\) and substitute \(\omega_{k}s\)
for \(t\). Since \(\alpha_{k} = i\omega_{k}\beta_{k}\) we have
\(\sum_{k}\omega_{k}\hat{q}_{a,k}=\hat{q}\).  The 
quadratures \(\hat{q}_{k}\) have normalization \(|\alpha_{k}|^{2}\).
After applying Eq.~\eqref{eq:permodeops} to each mode,
Eq.~\eqref{eq:logbymode} becomes
\begin{align}
  e^{i\hat{q}_{\delta}s}e^{-i\hat{q}s}
  &=
    e^{is\hat{q}/2}
    e^{i\sum_{k}\qty(
    \phi_{k}\hat{B}_{k}
    +s\qty(\frac{\phi_{k}}{\omega_{k}s\delta} - \frac{\phi_{k}}{2}\cot(\phi_{k}/2))\hat{q}_{k}
    )}
    e^{-is\hat{q}/2}.
\end{align}
This is of the form \(\hat{U} e^{i\hat{H}}\hat{U}^{\dagger}\) for the
unitary \(\hat{U}=e^{is\hat{q}/2}\) and the self-adjoint operator
\begin{align}
  \hat{H} &=\sum_{k} \phi_{k}\hat{B}_{k} + s u(\delta\,s)\hat{u}_{\delta\,s},
\end{align}
where \(\hat{u}_{\delta\,s}\) is a normalized quadrature that commutes with the \(\hat{q}_{k}\).
The quadrature \(\hat{u}_{\delta\,s}\) is given by
\begin{align}
  \hat{u}_{\delta\,s} &=
            \frac{1}{u(\delta\,s)}\sum_{k}\qty(\frac{\phi_{k}}{\omega_{k}\delta\,s}
            - \frac{\phi_{k}}{2}\cot(\phi_{k}/2))\hat{q}_{k}
\end{align}
with \(u(\delta\,s)\geq 0\) defined by
\begin{align}
  u(\delta\,s)^{2} &= \sum_{k}|\alpha_{k}|^{2}\qty(\frac{\phi_{k}}{\omega_{k}\delta\,s}
            - \frac{\phi_{k}}{2}\cot(\phi_{k}/2))^{2}.
\end{align}
For the case where the \(\omega_{k}\) are identical, \(\hat{u}_{\delta\,s}=\hat{q}\).

We now have
\begin{align}
  \qty|e^{-i\hat q_{\delta}s}\ket{\psi} - e^{-i\theta}e^{-i\hat{q}s}\ket{\psi}|^{2}
  &=
  \bra{\psi}\qty(2-2\Re\qty(e^{is\hat{q}/2}e^{i(\hat{H}-\theta)}e^{-is\hat{q}/2}))
    \ket{\psi}.
\end{align}
By applying Eq.~\eqref{eq:reeith-ineq-op} we can bound
\begin{align}
  2-2\Re\qty(\hat{U} e^{i\hat{H}-\theta}\hat{U}^{\dagger})
  &\leq
    \hat{U} (\hat{H}-\theta)^{2}\hat{U}^{\dagger}.
\end{align}
For states \(\ket{\psi}\) that are
vacuum on the LO modes, \(\ket{\psi'} = U\ket{\psi}\) is also
vacuum on the LO modes.  When contracting the operator \(\hat{H}^{2}\)
with vacuum on the LO modes, terms that are linear in each LO mode's
annihilation and creation operators are eliminated. The surviving
terms are those that involve no operators on the LO modes or those
that act as \(\hat{b}_{k}\hat{b}_{k}^{\dagger}\) on the LO modes.
Therefore
\begin{align}
  \qty|e^{-i\hat q_{\delta}s}\ket{\psi} - e^{-i\theta}e^{-i\hat{q}s}\ket{\psi}|^{2}
 \hspace*{-2in}&\notag\\
  &\leq \bra{\psi}e^{is\hat{q}/2}
    \qty(\sum_{k}\phi_{k}^{2}\hat{a}_{k}^{\dagger}\hat{a}_{k} 
    +\qty(su(\delta\,s)\hat{u}_{\delta\,s}-\theta)^{2})e^{-is\hat{q}/2}\ket{\psi}
    \notag\\
  &=
    \sum_{k}\phi_{k}^{2}
    \bra{\psi}e^{is\hat{q}_{k}/2}
    \hat{a}_{k}^{\dagger}\hat{a}_{k}
    e^{-is\hat{q}_{k}/2}\ket{\psi}
    +\bra{\psi}\qty(s u(\delta\,s)\hat{u}_{\delta\,s}-\theta)^{2}\ket{\psi}
    \notag\\
  &=
    \sum_{k}\phi_{k}^{2}
    \bra{\psi}
    \qty(\hat{a}_{k}^{\dagger}-s\alpha_{k}^{*}/2)
    \qty(\vphantom{\hat{a}_{k}^{\dagger}}\hat{a}_{k}-s\alpha_{k}/2)
    \ket{\psi}
    +\bra{\psi}\qty(s u(\delta\,s)\hat{u}_{\delta\,s}-\theta)^{2}\ket{\psi}
    \notag\\
  &\leq
    \sum_{k}\phi_{k}^{2}
    \qty(2\bra{\psi}
    \hat{n}_{k}    \ket{\psi}+s^{2}|\alpha_{k}|^{2}/2)
    +\bra{\psi}\qty(s u(\delta\,s)\hat{u}_{\delta\,s}-\theta)^{2}\ket{\psi}.
    \label{eq:detailedbnd0}
\end{align}
For the last line, we applied the inequality
\(\qty(\hat{A}-x)^{\dagger}
\qty(\hat{A}-x)
\leq 2\hat{A}^{\dagger}\hat{A}+2|x|^{2}\). To see this, for states \(\ket{\phi}\) in the
domain of \(\hat{A}\) we have
\begin{align}
  \bra{\phi}\qty(\hat{A}-x)^{\dagger}
  \qty(\hat{A}-x)\ket{\phi}
  &\leq   \bra{\phi}\qty(\hat{A}-x)^{\dagger}
  \qty(\hat{A}-x)\ket{\phi} +   \bra{\phi}\qty(\hat{A}+x)^{\dagger}
    \qty(\hat{A}+x)\ket{\phi}
    \notag\\
  & = 2\bra{\phi}\hat{A}^{\dagger}\hat{A}\ket{\phi}
         + 2|x|^{2}\braket{\phi} = \bra{\phi}\qty( 2\hat{A}^{\dagger}\hat{A}+2|x|^{2})\ket{\phi}.
\end{align}
In Eq.~\eqref{eq:detailedbnd0}, for fixed \(s\) we can optimize
the choice of \(\theta\) by setting
\(\theta=\bra{\psi}su(\delta\,s)\hat{u}_{\delta\,s}\ket{\psi}\).

The bound of Eq.~\eqref{eq:detailedbnd0} is our most specific bound on
the difference between the states obtained after evolving according to
\(\hat{q}_{\delta}\) or \(\hat{q}\) for  the fixed eigenvalue \(s\) of the
apparatus operator \(\hat{s}_{M}\). In principle, it suffices to
substitute the bound in Eq.~\eqref{eq:measbnd1} after inserting the integral
over the apparatus projection-valued measure \(d\Pi(s)_{M}\).  Because
of the unwieldy functional dependence on \(s\) in
Eq.~\eqref{eq:detailedbnd0}, this does not result in readily
interpretable bounds.  In order to get more interpretable bounds, we
derive conservative simplified bounds.  For simplicity, and because
\(s\) is to be replaced by \(\hat{s}_{M}\), we set \(\theta=0\).
We first treat the general case and then obtain specific bounds for the
case of standard pulsed homodyne, in which case \(\omega_{k}=\omega=1\).

We begin with an upper bound on
\(\bra{\psi}u(\delta\,s)^{2}\hat{u}_{\delta
  s}^{2}\ket{\psi}\). The bound in Eq.~\eqref{eq:detailedbnd0} is
infinite if \(\ket{\psi}\) is not in the domain of every
\(\hat{n}_{k}\) with non-zero \(\phi_{k}\).  We therefore assume that
\(\bra{\psi}\hat{n}_{k}\ket{\psi}\) is finite for every \(k\).
Equivalently, we assume that the total photon number expectation is
finite.  Define
\(h_{k}(\delta\,s) = \phi_{k}/(\omega_{k}\delta\,s) -
(\phi_{k}/2)\cot(\phi_{k}/2)\).  Let \(\hat{n}_{\delta\,s}\) be the
number operator for the mode associated with the normalized quadrature
\(\hat{u}_{\delta\,s}\). With our conventions for quadrature
normalization, we have
\(\hat{u}_{\delta\,s}^{2}\leq 4(\hat{n}_{\delta\,s}+1/2)\), see
Eq.~\eqref{app:eq:quadsq-hatn}.  An annihilation operator for the mode
of \(\hat{u}_{\delta\,s}\) is given by
\begin{align}
  \hat{a}_{u}&=\frac{1}{u(\delta\,s)}\sum_{k}\alpha_{k}^{*}h_{k}(\delta\,s) \hat{a}_{k},
\end{align}
where we overloaded the suffixes of \(a\) for convenience.
This gives
\begin{align}
  \bra{\psi}u(\delta\,s)^{2}\hat{n}_{u}\ket{\psi}
  &= \qty| u(\delta\,s)\hat{a}_{u}\ket{\psi}|^{2}
    \notag\\
  &=
    \qty|\sum_{k}\alpha_{k}^{*}h_{k}(\delta\,s) \hat{a}_{k}\ket{\psi}|^{2}.
\end{align}
We need the following version of the Cauchy-Schwarz inequality.  For
coefficients \(\gamma_{k}\), operators \(\hat{A}_{k}\) and a state
\(\ket{\varphi}\) in the domain of every operator \(\hat{A}_{k}\), we
have
\begin{align}
  \qty|\sum_{k}\gamma_{k}\hat{A}_{k}\ket{\varphi}|^{2}
  &\leq \sum_{k}|\gamma_{k}|^{2}
    \sum_{k}\qty|\hat{A}_{k}\ket{\varphi}|^{2}.
\end{align}
To apply this inequality, we set \(\gamma_{k}=\alpha_{k}^{*}\sqrt{h_{k}(\delta\,s)}\)
and \(\hat{A}_{k}=\sqrt{h_{k}(\delta\,s)}\hat{a}_{k}\). This gives
\begin{align}
  \bra{\psi}u(\delta\,s)^{2}\hat{n}_{\delta\,s}\ket{\psi}
  &\leq \qty(\sum_{k}|\alpha_{k}|^{2}|h_{k}(\delta\,s)|)
            \sum_{k}|h_{k}(\delta\,s)|\bra{\psi} \hat{n}_{k}\ket{\psi}.
\end{align}
Define \(\bar h(\delta\,s) = \sum_{k}|\alpha_{k}|^{2}|h_{k}(\delta\,s)|\),
\(\overline{h^{2}}(\delta\,s) =
\sum_{k}|\alpha_{k}|^{2}h_{k}(\delta\,s)^{2} = u(\delta\,s)^{2}\) and
\(\overline{\phi^{2}}(\delta\,s) = \sum_{k}|\alpha_{k}|^{2}\phi_{k}^{2}\), which are weighted averages of the \(|h_{k}(\delta\,s)|\), \(h_{k}(\delta\,s)^{2}\) and
\(\phi_{k}^{2}\).  Substituting in Eq.~\eqref{eq:detailedbnd0} with
\(\theta=0\) gives
\begin{align}
    \qty|e^{-i\hat q_{\delta}s}\ket{\psi} - e^{-i\theta}e^{-i\hat{q}s}\ket{\psi}|^{2}
  \hspace*{-1in}&\notag\\
  &\leq
    \sum_{k}\qty(2\phi_{k}^{2}+4 s^{2}|h_{k}(\delta\,s)| \bar h(\delta\,s) )
    \bra{\psi}\hat{n}_{k}\ket{\psi}
    + s^{2}\qty(\overline{\phi^{2}}(\delta\,s)/2 + 2\overline{h^{2}}(\delta\,s)).
    \label{eq:detailedbnd1}
\end{align}

To bound the coefficients in Eq.~\eqref{eq:detailedbnd1} with more simple
expressions, we use the following inequalities:
\begin{align}
  \phi_{k}^{2} &\leq (\omega_{k} \delta\,s )^{2}
  , \;
  \phi_{k}^{2} \leq \pi^{2},
    \notag\\
  \qty|\frac{\phi_{k}}{\omega_{k}\delta\,s}-1|
   &\leq\frac{2}{\pi^{2}}(\omega_{k}\delta\,s)^{2}
  ,\;
    \qty|\frac{\phi_{k}}{\omega_{k} \delta\,s}-1|  \leq 2,
    \notag\\
  \qty|1-\frac{\phi_{k}}{2}\cot(\frac{\phi_{k}}{2})|
  &\leq \frac{1}{3\pi}\phi_{k}^{2} \leq \frac{1}{3\pi} (\omega_{k}\delta\,s)^{2}
  ,\;
    \qty|1-\frac{\phi_{k}}{2}\cot(\frac{\phi_{k}}{2})|
    \leq 1,
    \notag\\
  |h_{k}(\delta\,s)| &\leq \qty|\frac{\phi_{k}}{\omega_{k}\delta\,s} -1|
                    + \qty|1 - \frac{\phi_{k}}{2}\cot(\frac{\phi_{k}}{2})|
                    \notag\\
               &\leq \frac{1}{\pi^{2}}(2+\pi/3)(\omega_{k}\delta\,s)^{2}
                 \leq \frac{1}{3}(\omega_{k}\delta\,s)^{2},
    \notag\\
  |h_{k}(\delta\,s)| &\leq 3.
                      \label{eq:phi+bnds}
\end{align}
For the last inequality for \(|h_{k}(\delta\,s)|\), we simplified the
factor by applying \(2+\pi/3\leq \pi\) and \(\pi>3\).  Define
\(\overline{\omega^{2}} = \sum_{k}|\alpha_{k}|^{2}\omega_{k}^{2}\).
For substitution into Eq.~\eqref{eq:detailedbnd1} we bound the terms
according to
\begin{align}
  \phi_{k}^{2}
  & \leq (\delta\,s)^{2}\omega _{k}^{2},
    \notag\\
  |h_{k}(\delta\,s)|\bar h(\delta\,s)
  &\leq 3 |h_{k}(\delta\,s)| \leq (\delta\,s)^{2}\omega_{k}^{2},
    \notag\\
  \overline{\phi^{2}}(\delta\,s)
  &\leq (\delta\,s)^{2}\overline{\omega^{2}}
    \notag\\
  \overline{h^{2}}(\delta\,s)
  &=\sum_{k}|\alpha_{k}|^{2}h_{k}(\delta\,s)^{2}
    \notag\\
  &\leq 3\sum_{k}|\alpha_{k}|^{2}|h_{k}(\delta\,s)|
    \notag\\
   &=3 \bar{h}(\delta\,s) \leq (\delta\,s)^{2} \overline{\omega^{2}}.
\end{align}
Define \(\hat{\Omega} = \sum_{k}\omega_{k}^{2}\hat{n}_{k}\). We obtain
\begin{align}
    \qty|e^{-i\hat q_{\delta}s}\ket{\psi} - e^{-i\hat{q}s}\ket{\psi}|^{2}
  &\leq
    (\delta\,s)^{2}\bra{\psi}
    \qty(2\qty(1+2s^{2})\hat{\Omega} + s^{2}(5/2)\overline{\omega^{2}})
    \ket{\psi}
    \label{eq:simplifiedbndm}
    \\
  &\leq
    4(\delta\,s)^{2}\bra{\psi}
    \qty(\qty(1+s^{2})\hat{\Omega} + s^{2}\overline{\omega^{2}})
    \ket{\psi}.
    \label{eq:simplifiedbnd0}
\end{align}
The last inequality is conservative for the sake of simplicity but
retains the same factor for the coefficient of
\(\delta^{2}s^{4}\hat{\Omega}\).  The second bound in
Thm.~\ref{thm:mainbnd} is obtained by moving the leading factor of
\(s^{2}\) into the bra-ket expression and substituting \(\hat{s}_{M}\)
for \(s\).

The bound of Eq.~\eqref{eq:simplifiedbndm} is not optimized for the
case of small \(\delta\,s\) in that it conservatively estimated extra
factors of order \((\delta\,s)^{2}\) for small \(\delta\,s\). To obtain
a bound that includes these extra factors we bound the coefficients of
Eq.~\eqref{eq:detailedbnd1} according to
\begin{align}
  \bar h(\delta\,s)
  & \leq (\delta\,s)^{2}\overline{\omega^{2}}/3
  \notag\\
  |h_{k}(\delta\,s)|\bar h(\delta\,s)
  &\leq (\delta\,s)^{4}\omega_{k}^{2}\overline{\omega ^{2}}/9
    \notag\\
  \overline{h^{2}}(\delta\,s)
  &\leq (\delta\,s)^{4}\sum_{k}|\alpha_{k}|^{2}\omega_{k}^{4}/9
  = (\delta\,s)^{4}\overline{\omega^{4}}/9,
\end{align}
where we introduced \(\overline{\omega^{4}}=\sum_{k}|\alpha_{k}|^{2}\omega_{k}^{4}\) .
We get
\begin{align}
  \qty|e^{-i\hat q_{\delta}s}\ket{\psi} - e^{-i\hat{q}s}\ket{\psi}|^{2}
  \hspace*{-1.5in}&
                  \notag\\
  &\leq
    (\delta\,s)^{2}\bra{\psi}
    \qty(2\qty(1+(2/9)s^{2}(\delta\,s)^{2}\overline{\omega^{2}})\hat{\Omega} + s^{2}\qty((1/2)\overline{\omega^{2}} +(2/9)(\delta\,s)^{2}\overline{\omega^{4}}))
    \ket{\psi}
    \notag\\
  &=2(\delta\,s)^{2}\bra{\psi}
    \qty(\hat{\Omega} + (1/4)s^{2}\overline{\omega^{2}})
    \ket{\psi} + O((\delta\,s)^{4})
  \label{eq:simplifiedbnd1}
\end{align}

Next we specialize to the case of standard pulsed homodyne to
obtain tighter bounds. For standard pulsed homodyne,
\(\omega_{k}=\omega\) for all \(k\), and with our conventions,
\(\omega=\omega_{1}=1\). This gives
\(\phi_{k}=\phi=[\delta\,s]_{\pi}\) and \(h_{k}(\delta\,s) =
h(\delta\,s) = \phi/(\delta\,s)-(\phi/2)\cot(\phi/2)\)
for all \(k\). As noted when we
defined \(\hat{u}_{\delta\,s}\), we now have
\(\hat{u}_{\delta\,s}=\hat{q}\).
According to the definitions, \(u(\delta\,s)^{2} = h(\delta\,s)\).
From Eq.~\eqref{eq:phi+bnds},  \(\phi^{2}\leq (\delta\,s)^{2}\) and \(|h(\delta\,s)|\leq
(\delta\,s)^{2}/3\).
Substituting in the bound of
Eq.~\eqref{eq:detailedbnd0} with \(\theta=0\) gives
\begin{align}
  \qty|e^{-i\hat q_{\delta}s}\ket{\psi} - e^{-i\hat{q}s}\ket{\psi}|^{2}
  &\leq
    (\delta\,s)^{2}\qty(2\bra{\psi}\hat{n}_{\tot}\ket{\psi} 
    +s^{2}\qty(\frac{1}{9}(\delta\,s)^{2}\bra{\psi}\hat{q}^{2}\ket{\psi}+\frac{1}{2})).
    \label{eq:detailedbnd_standard}
\end{align}

For comparison in the regime of small \(\delta\,s\) and \(s\), we determine lower
bounds on the minimum over angles \(\theta\) of
\(|e^{-i\hat{q}_{\delta}\hat{s}_{M}}\ket{\psi} -
e^{-i\theta}e^{-i\hat{q}\hat{s}_{M}}\ket{\psi}|^{2}\), where
\(\ket{\psi}\) is in a family of coherent states.  We begin by
considering a generic mode with the operators defined in
Eq.~\eqref{eq:genericmodealgebra1}, where we omit the mode index \(k\).
Without loss of generality, we can assume that \(\alpha\) and
\(\beta\) are non-negative real.  For
\(\vec{\gamma}=(\gamma_{a},\gamma_{b})^{T}\) with \(\gamma_{a}\) and
\(\gamma_{b}\) real, let \(\ket{\vec{\gamma}}=\ket{i\gamma_{a},\gamma_{b}}\)
be the state with amplitudes \(i\gamma_{a}\) and \(\gamma_{b}\) on
modes \(a\) and \(b\).  We have 
\begin{align}
  e^{-i\vec{u}^{T}\vec{q}}\ketbra{\gamma}e^{i\vec{u}^{T}\vec{q}}
  &=
    \ketbra{\vec{\gamma}-\beta \vec{u}},
    \notag\\
  e^{-i\phi \hat{B}}\ketbra{\vec{\gamma}}e^{i\phi \hat{B}}
  &=
    \ketbra{R_{\phi}\vec{\gamma}}.
\end{align}
\Pc{
  To calculate the effect of \(e^{-i\phi\hat{b}}\), write
  \begin{align}
    \pmatr{\hat{a}\\\hat{b}}    \ket{\vec{\gamma}}
    = \pmatr{i&0\\0&1}
                              \vec{\gamma}\ket{\vec{\gamma}}.
  \end{align}
  Let \(\hat{U}_{\phi}=e^{-i\phi\hat{B}}\) and compute
  \begin{align}
    \pmatr{\hat{a}\\\hat{b}}\hat{U}_{\phi}\ket{\vec{\gamma}}
    &=
      \hat{U}_{\phi}
      \qty(\hat{U}_{\phi}^{\dagger}\pmatr{\hat{a}\\\hat{b}}\hat{U}_{\phi})
    \vec{\gamma}.
  \end{align}
  To check the rotation direction, we take the derivative with respect
  to \(\phi\) at \(\phi=0\) of
  \(\hat{U}_{\phi}^{\dagger}\pmatr{\hat{a}\\\hat{b}}\hat{U}_{\phi}\).
  This derivative is given by \(\qty[i\hat{B},\pmatr{\hat{a}\\\hat{b}}]\).
  We have \([\hat{B},\hat{a}] = -b\) and \([\hat{B},\hat{b}] = -a\),
  so 
  \begin{align}
    \qty[i\hat{B},\pmatr{\hat{a}\\\hat{b}}]
    &= \pmatr{0&-i\\-i&0}\pmatr{\hat{a}\\\hat{b}}.
  \end{align}
  From this we infer that
  \begin{align}
    \hat{U}_{\phi}^{\dagger}\pmatr{\hat{a}\\\hat{b}}\hat{U}_{\phi}
    &= \exp(\phi\pmatr{0&-i\\-i&0})\pmatr{\hat{a}\\\hat{b}}
    \notag\\
    &=\pmatr{\cos(\phi)&-i\sin(\phi)\\-i\sin(\phi)&\cos(\phi)}\pmatr{\hat{a}\\\hat{b}}.
  \end{align}
  Accordingly
  \begin{align}
    \pmatr{\hat{a}\\\hat{b}}\hat{U}_{\phi}\ket{\vec{\gamma}}
    &=
      \hat{U}_{\phi}
      \pmatr{\cos(\phi)&-i\sin(\phi) \\-i\sin(\phi)&\cos(\phi)}
                    \pmatr{\hat{a}\\\hat{b}}
    \ket{\vec{\gamma}}
    \notag\\
    &=  \hat{U}_{\phi}
      \pmatr{\cos(\phi)&-i\sin(\phi) \\-i\sin(\phi)&\cos(\phi)}
                                                     \pmatr{i&0\\0&1}
                                                                    \vec{\gamma}\ket{\vec{\gamma}}
                                                                    \notag\\
    &=
      \pmatr{i&0\\0&1}\pmatr{\cos(\phi)&-\sin(\phi) \\\sin(\phi)&\cos(\phi)}\vec{\gamma}
                                                                \hat{U}_{\phi}\ket{\vec{\gamma}}
                                                                \notag\\
    &=
      \pmatr{i&0\\0&1}R_{\phi}\vec{\gamma} \hat{U}_{\phi}\ket{\vec{\gamma}},
  \end{align}
  From which we deduce the action of \(e^{-i\phi\hat{B}}\) given
  above.  } We consider initial states of the form
\(\ket{\vec{\gamma}_{0}}=\ket{i\gamma_{a},0}\) with
\(\vec{\gamma}_{0}=(\gamma_{a},0)^{T}\) that are vacuum on the LO
mode.  We apply Eq.~\eqref{eq:simplerotationandshifts1} to determine
\begin{align}
  e^{it( \vec{w}^{T}\vec{q}+\delta\hat{B})}e^{-it \vec{w}^{T}\vec{q}}
  \ket{\vec{\gamma}_{0}}
  &= e^{i\varphi}
    \ket{R_{-\phi}(\vec{\gamma}_{0}-\beta t \vec{w}) + (\beta/\delta) (1-R_{-\phi}) S \vec{w}}
\end{align}
for some angle \(\varphi\), where \(\phi=[t\delta]_{\pi}\) as defined
after Eq.~\eqref{eq:simplerotationandshifts1}.  Substituting \(\vec{w}=(1,0)^{T}\) gives
\begin{align}
    e^{it( \hat{q}_{a}+\delta\hat{B})}e^{-it \hat{q}_{a}}\ket{\gamma_{0}}
  &=e^{i\varphi}
    \big|i(\cos(\phi)(\gamma_{a}-\beta t) + (\beta/\delta)\sin(\phi)),
    \notag\\
  &\hphantom{=e^{i\varphi}\big|\;}
     -\sin(\phi)(\gamma_{a}-\beta t) -(\beta/\delta)(1-\cos(\phi))\big\rangle.
\end{align}
Accordingly, \(\qty|\bra{\gamma_{0}}e^{it( \hat{q}_{a}+\delta\hat{B})}e^{-it \hat{q}_{a}}\ket{\gamma_{0}}| = e^{-|\Delta(\gamma_{a})|^{2}/2}\) with
\begin{align}
  \Delta(\gamma_{a}) &= \begin{pmatrix}
    (\cos(\phi)-1)\gamma_{a}+(\sin(\phi)/(\delta t)-\cos(\phi))\beta t\\
    -\sin(\phi)(\gamma_{a}-\beta t) -(\beta/\delta)(1-\cos(\phi))
  \end{pmatrix}.
\end{align}
We now reintroduce the mode index. Let \(\gamma_{k}\) be the coherent
amplitude \(\gamma_{a_{k}}\) and \(\Delta_{k}(\gamma_{k})\) the corresponding
amplitude difference vector. Let \(\ket{\psi}\) be the coherent state
whose amplitudes on mode \(a_{k}\) are \(i\gamma_{a_{k}}\) and vacuum
on the LO modes. Multiplying over the modes gives
\begin{align}
  \qty|\bra{\psi} e^{i\hat{q}_{\delta}s}e^{-i\hat{q}s}\ket{\psi}|
  &=
    e^{-\sum_{k}|\Delta_{k}(\gamma_{k})|^{2}/2},
\end{align}
where
\begin{align}
  \Delta_{k}(\gamma_{k}) &= \begin{pmatrix}
    (\cos(\phi_{k})-1)\gamma_{k}+(\sin(\phi_{k})/(\omega_{k}\delta\,s)-\cos(\phi)_{k})\alpha_{k}s\\
    -\sin(\phi_{k})(\gamma_{k}-\alpha_{k} s)
    -(\alpha_{k}s/(\omega_{k} \delta\,s))(1-\cos(\phi_{k}))
  \end{pmatrix}.
\end{align}
To compare with the bound of Eq.~\eqref{eq:simplifiedbnd1}, we
consider small \(\delta\,s\) and \(s\), and choose \(\gamma_{k}\) to
witness a lower bound approaching that of
Eq.~\eqref{eq:simplifiedbnd1}. Consider the case where for all \(k\),
\(\omega_{k}\delta\,s \in [-\pi/4,\pi/4]\), 
so that \(\phi_{k}=\omega_{k}\delta\,s\). The second coordinate of
\(\Delta_{k}(\gamma_{k})\) has contributions of order \(\delta\,s\) and
\(\delta\,s^{2}\) and dominate the contributions of order
\((\delta\,s)^{2}\) of the first coordinate.  We neglect the first
coordinate for the bound
\begin{align}
  |\Delta_{k}(\gamma_{k})|^{2}
  &\geq
    \qty(\sin(\phi_{k})(\gamma_{k}-\alpha_{k} s)
    +(\alpha_{k}s/\phi_{k})(1-\cos(\phi_{k})))^{2}
    \notag\\
  &=
    \qty(\sin(\phi_{k})\gamma_{k}
    +\alpha_{k}s((1-\cos(\phi_{k}))/\phi_{k}-\sin(\phi_{k})))^{2}.
\end{align}
The expression on the right-hand side that is being squared is odd in \(\phi_{k}\).
For  \(0\leq x\leq \pi/4\), we have
\begin{align}
  \sin(x)&\geq (4/(\sqrt{2}\pi))x\geq (9/10)x,
           \notag\\
  (1-\cos(x)) &\leq x^{2}/2
                \notag\\
  \sin(x) - (1-\cos(x))/x & \geq (4/(\sqrt{2}\pi) - 1/2)x \geq (2/5)x.
\end{align}
We choose \(\gamma_{k}\) to have the opposite sign of \(s\), so that
\(\sin(\phi_{k})\gamma_{k}\) and
\(\alpha_{k}s((1-\cos(\phi_{k}))/\phi_{k}-\sin(\phi_{k}))\) have
the same sign.
Then, omitting the cross term when squaring,
\begin{align}
  |\Delta_{k}(\gamma_{k})|^{2}
  &\geq
    (9/10)^{2}\phi_{k}^{2}|\gamma_{k}|^{2}
    + (2/5)^{2}\phi_{k}^{2}|\alpha_{k} s|^{2}
    \notag\\
  &= (9/10)^{2}(\delta\,s)^{2}
    \qty(\omega_{k}^{2}|\gamma_{k}|^{2} 
    + (4/9)^{2} \omega_{k}^{2}|\alpha_{k}|^{2}s^{2}).
\end{align}
We sum this over \(k\) for
\begin{align}
  \sum_{k}  |\Delta_{k}(\gamma_{k})|^{2}
  &\geq
    (9/10)^{2}(\delta\,s)^{2}\bra{\vec{\gamma}_{0}}
    \hat{\Omega} +
    (4/9)^{2} s^{2}\overline{\omega^{2}}
    \ket{\vec{\gamma}_{0}}.
\end{align}
By applying \(e^{-x^{2}/2}= 1-x^{2}/2+O(x^{4})\), we can now estimate
\begin{align}
  \min_{\theta}|e^{-i\hat{q}_{\delta}\hat{s}_{M}}\ket{\psi} -
  e^{-i\theta}e^{-i\hat{q}\hat{s}_{M}}\ket{\psi}|^{2}
  &=
    2\qty(1- e^{-\sum_{k}|\Delta_{k}(\gamma_{k})|^{2}/2})
    \notag\\
  &=
    (9/10)^{2}(\delta\,s)^{2}\bra{\vec{\gamma}_{0}}
    \hat{\Omega} +
    (4/9)^{2} s^{2}\overline{\omega^{2}}
    \ket{\vec{\gamma}_{0}} + O((\delta\,s)^{4}).
\end{align}
This may be compared to Eq.~\eqref{eq:simplifiedbnd1}, which gives an upper
bound of the same form up to order \((\delta\,s)^{4}\) with somewhat larger
constants.

\section{Bounds for standard pulsed homodyne}
\label{app:standardhomodyne}

We establish the results of Sect.~\ref{sec:sconv} for standard pulsed
homodyne to obtain better bounds. See Sect.~\ref{sec:onemode} for
comparisons. The proofs parallel those of
Sect.~\ref{sec:sconv}, and we point out the differences as needed.
Throughout this section, we set \(\omega_{k}=\omega=1\). The total
photon number operator for modes \(\bm{a}\) is \(\hat{n}_{\tot}\).

\begin{theorem}[Thm.~\ref{thm:mainbnd} for standard pulsed
homodyne]\label{thm:mainbnd_sph}
Let \(\ket{\psi}\) be a joint state of the signal, LO modes, the
apparatus \(M\) and any other relevant systems including those needed
for purifying the state. Assume that these states are vacuum on the LO
modes.
  Then
  \begin{align}
    \qty|e^{-i\hat q_{\delta}s}\ket{\psi} -
e^{-i\hat{q}s}\ket{\psi}|^{2}
    &
      \leq
      (\delta\,s)^{2}\qty(2\bra{\psi_{s}}\hat{n}_{\tot}\ket{\psi_{s}}
      +s^{2}\qty(\frac{1}{9}(\delta
s)^{2}\bra{\psi_{s}}\hat{q}^{2}\ket{\psi_{s}}+\frac{1}{2})).
      \label{thm:eq:mains_sph}
  \end{align}
  For every self-adjoint apparatus operator \(\hat{s}_{M}\),
  \begin{align}
    \qty|e^{-i\hat q_{\delta}\hat{s}_{M}}\ket{\psi}
    - e^{-i\hat{q}\hat{s}_{M}}\ket{\psi}|^{2}
    &\leq
      \delta^{2}\bra{\psi_{s}}\hat{s}_{M}^{2}\qty(2\hat{n}_{\tot}
      +\hat{s}_{M}^{2}\qty(\frac{1}{9}\delta^{2}\hat{s}_{M}^{2}\hat{q}^{2}+\frac{1}{2}))\ket{\psi_{s}}.
      \label{thm:eq:mainhats_sph}
  \end{align}
\end{theorem}

\begin{proof}
  The inequality of Eq.~\eqref{thm:eq:mains_sph} is identical to that
  of Eq.~\eqref{eq:detailedbnd_standard}. The proof of
  Eq.~\eqref{thm:eq:mainhats_sph} parallels the proof of the
  corresponding bound in Thm.~\ref{thm:mainbnd}, which justifies
  substituting the operator \(\hat{s}_{M}\) for \(\hat{s}\) in
  Eq.~\eqref{thm:eq:mains_sph}.
\end{proof}

\begin{theorem}[Thm.~\ref{thm:qmeas} for standard pulsed homodyne]
  \label{thm:qmeas_sph}
Let \(\hat{s}_{M}\) have non-degenerate spectrum so that the
apparatus Hilbert space has a representation \(L^{2}(\rls, d\mu(s))\),
and let \(\ket{g(\bm{\cdot})}\) be the initial state of the
apparatus. Define
\begin{align} b_{g,2l} &= \int_{s\in\rls} d\mu(s) s^{2l}|g(s)|^{2}.
\end{align} Let \(\ket{\psi}\) the initial state of the signal, LO
modes and systems other than the apparatus, where \(\ket{\psi}\) is
vacuum on the LO modes. Then, with \(\ket{\phi_{\delta}}\) as in
Eq.~\eqref{eq:psimdelta_general}, we have
\begin{align} \qty|\ket{\phi_{\delta}}-\ket{\phi_{0}}|^{2} &\leq
\delta^{2}\qty( 2 b_{ g,2}\bra{\psi}\hat{n}_{\tot}\ket{\psi} + \frac{1}{9}\delta^{2}b_{
g,6}\bra{\psi}\hat{q}^{2}\ket{\psi}+ \frac{1}{2}b_{
g,4} ).
\label{eq:qmeastheorem1_sph}
\end{align}
\end{theorem}

\begin{proof}
The proof parallels that of Thm~\ref{thm:qmeas}.
\end{proof}

\begin{proposition}[Prop.~\ref{prop:fbndedmeas} for standard pulsed
  homodyne]
  \label{prop:fbndedmeas_sph}
  Suppose that \(f(x)\) is the Fourier transform of a bounded
  complex measure \(d\tilde\mu(\kappa)\). Write
\(d\tilde\mu(\kappa)=e^{i\theta(\kappa)}d\mu(\kappa)\) for a bounded
positive measure \(d\mu(\kappa)\).
  Define \(\tilde{f}_{l}=\int_{\kappa\in\rls} d\mu(\kappa)
|\kappa|^{l}\). Then for all states \(\ket{\psi}\) of the signal and
LO modes and other relevant systems,
  where \(\ket{\psi}\) is vacuum on the LO modes, we have the
following bound:
  \begin{align}
    \qty|f(\hat{q}_{\delta})\ket{\psi}-f(\hat{q}_{0})\ket{\psi}|^{2}
    &\leq
          \delta^{2}(\tilde f_{0}+\tilde f_{2})\qty(
      2\bra{\psi}\hat{n}_{\tot}\ket{\psi} + \frac{1}{9}\delta^{2}\tilde f_{4}\bra{\psi}\hat{q}^{2}\ket{\psi} + \frac{1}{2}\tilde f_{2}).
                                                              \label{eq:cor:cfunq_sph}
  \end{align}
\end{proposition}

\begin{proof} The proof of Prop.~\ref{prop:fbndedmeas} needs to be
modified starting at the third to last line of
Eq.~\eqref{eq:propfbndedmeas}, where we now apply Eq.~\eqref{thm:eq:mains_sph}
with \(s\) replaced by \(-\kappa\):
\begin{align}
  \qty|\bra{\phi}f(\hat{q}_{\delta}) - f(\hat{q}_{0})\ket{\psi}|^{2}
  \hspace*{-1in}
  &\notag\\
  &\leq
      \int_{\kappa\in\rls}d\mu(\kappa)(1+\kappa^{2})
      \int_{\kappa\in\rls}d\mu(\kappa)\frac{1}{1+\kappa^{2}}
      \qty|\qty(e^{i\kappa\hat{q}_{\delta}}-
      e^{i\kappa\hat{q}_{0}})\ket{\psi}|^{2}
      \notag\\
  &\leq
    \delta^{2}(\tilde f_{0}+\tilde f_{2})
      \int_{\kappa\in\rls}d\mu(\kappa)\frac{1}{1+\kappa^{2}}\kappa^{2}
      \qty(2 \bra{\psi}\hat{n}_{\tot}\ket{\psi}
      +
      \kappa^{2}\qty(\frac{1}{9}\delta^{2}\kappa^{2}\bra{\psi}\hat{q}^{2}\ket{\psi}
      + \frac{1}{2})
      )
      \notag\\
  &\leq
    \delta^{2}(\tilde f_{0}+\tilde f_{2})\qty(
      2\bra{\psi}\hat{n}_{\tot}\ket{\psi} + \frac{1}{9}\delta^{2}\tilde f_{4}\bra{\psi}\hat{q}^{2}\ket{\psi} + \frac{1}{2}\tilde f_{2})
    .
\end{align}
The result then follows as before.
\end{proof}

\section{A generalization of Prop.~\ref{prop:fbndedmeas}}
\label{app:generalize_fbndedmeas}

In general, defining functions of two non-commuting self-adjoint
operators is problematic and requires operator ordering conventions
and care with the operators' projection valued measures .  A
systematic approach can be based on multiple operator integrals as
explained in Ref.~\cite{peller2016multiple}. Because for \(\delta>0\)
\(\hat{q}_{\delta}\) has discrete spectrum, for bounded measurable
functions \(f(x,y)\), we can directly define
\begin{align}
  f(\hat{q}_{\delta},\hat{q})\ket{\psi}
  &=
    \int_{x\in\rls,y\in\rls} f(x,y)d\Pi_{\delta}(x)d\Pi(y)\ket{\psi}
    = \sum_{l}\Pi_{\delta,l}\int_{y\in\rls}f(x_{l},y)d\Pi(y)\ket{\psi},
\end{align}
where \(\{x_{l}\}_{l}\) enumerates the distinct values in the spectrum
of \(\hat{q}_{\delta}\), and \(\Pi_{\delta,l}\) are the corresponding
eigenspace projectors. Because \(f(x,y)\) is bounded, we have that
\(f(\hat{q}_{\delta},\hat{q})\) is a bounded operator.

The following proposition generalizes Prop.~\ref{prop:fbndedmeas} with
a better bound, and its proof expands on that of the latter
proposition.  We consider multiple BBP homodyne measurements in
parallel.  For this purpose we vectorize the measurement operators.
We consider \(m\) independent pairs of signal and LO modes and let
\(\bm{\hat{q}}_{\delta}\) be the vector whose \(j\)'th entry consists
of the \(j\)'th pair's operator \(\hat{q}_{\delta}\), which we write
as \(\hat{q}_{j,\delta}\).  We let \(d\Pi_{\delta}(\bm{x})\) be the
joint projection valued measure for \(\bm{\hat{q}}_{\delta}\).  Let
\(\hat{\Omega}_{j}\) be the operator \(\hat{\Omega}\) defined in
Thm.~\ref{thm:mainbnd} for the \(j\)'th pair, and
\(\overline{\omega_{j}^{2}}\) the quantity \(\overline{\omega^{2}}\)
for the \(j\)'th pair.

\begin{proposition}\label{prop:fxybndedmeas}
  Suppose that \(f(\bm{x},\bm{y})\) be a function on
  \(\rls^{m}\times\rls^{m}\) that is the Fourier transform of a
  bounded complex measure \(d\tilde\mu(\bm{\kappa},\bm{\kappa}')\) and
  satisfies \(f(\bm{x},\bm{x})=0\).  Write
  \(d\tilde\mu(\bm{\kappa},\bm{\kappa}')
  =e^{i\theta(\bm{\kappa},\bm{\kappa}')}d\mu(\bm{\kappa},\bm{\kappa}')\)
  with \(d\mu(\bm{\kappa},\bm{\kappa}')\) positive.  Define the
  marginal distributions
  \(d\mu(\bm{\kappa})
  =\int_{\bm{\kappa}'\in\rls^{m}}d\mu(\bm{\kappa},\bm{\kappa}')\) and
  \(d\mu(\kappa_{j}) =
  \int_{\bm{\kappa}\in\rls^{m},\bm{\kappa}'\in\rls^{m}:(\bm{\kappa})_{j}=\kappa_{j}}
  d\mu(\bm{\kappa},\bm{\kappa}')\).
  Let
  \(\tilde{f}_{j,l}=\int_{\kappa_{j}\in\rls} d\mu(\kappa_{j})
  |\kappa_{j}|^{l}\), \(\tilde{f}_{\tot,l}=\sqrt{\sum_{j=1}^{m}\tilde{f}_{j,l}^{2}}\),
  \(\hat{\Omega}_{\tot} = \sum_{j=1}^{m}\hat{\Omega}_{j}\), and
  \(\overline{\omega^{2}_{\tot}} =
  \sum_{j=1}^{m}\overline{\omega^{2}_{j}}\).  Then for all states
  \(\ket{\psi}\) of the signal and LO modes and other relevant system
  that are vacuum on the LO modes
  \begin{align}
    |f(\hat{\bm{q}}_{\delta},\hat{\bm{q}})\ket{\psi}|^{2}
    &\leq 4\delta^{2}
            \bra{\psi}\qty(
      \qty(\tilde f_{\tot,1}^{2}+\tilde f_{\tot,2}^{2})\hat{\Omega}_{\tot}
         +\tilde f_{\tot,2}^{2}\overline{\omega^{2}_{\tot}})\ket{\psi}
      .
      \label{prop:appD1}
  \end{align}
\end{proposition}

Before we prove the proposition, to see that this bound improves on
that in Prop.~\ref{prop:fbndedmeas}, apply the Cauchy-Schwarz inequality
twice for
\begin{align}
  \tilde f_{\tot,1}^{2}
  &= \sum_{j=1}^{m}\qty(\int_{\kappa_{j}\in\rls}|\kappa_{j}|d\mu(\kappa_{j}))^{2}
    \notag\\
  &\leq
    \sum_{j=1}^{m}\int_{\kappa_{j}\in\rls}d\mu(\kappa_{j})
    \int_{\kappa_{j}\in\rls}|\kappa_{j}|^{2}d\mu(\kappa_{j})
    \notag\\
  &= \sum_{j=1}^{m}\tilde f_{j,0}\tilde f_{j,2}
    \notag\\
  &\leq \sqrt{\sum_{j=1}^{m}\tilde f_{j,0}^{2}}\sqrt{\sum_{j=1}^{m}\tilde f_{j,2}^{2}}
    \notag\\
  &=\tilde f_{\tot,0}\tilde f_{\tot,2}.
\end{align}
Therefore, the coefficient of \(\hat{\Omega}_{\tot}\)
is bounded by \((\tilde f_{\tot,0}+\tilde f_{\tot,2})\tilde f_{\tot,2}\),
and so is the coefficient of \(\overline{\omega^2_\tot}\).

\begin{proof}
  The quantity \(  |f(\hat{\bm{q}}_{\delta},\hat{\bm{q}})\ket{\psi}|\)
  can be expressed as
  \begin{align}
    |f(\hat{\bm{q}}_{\delta},\hat{\bm{q}})\ket{\psi}|
    &= \max_{\ket{\phi}}\Re(\bra{\phi}f(\hat{\bm{q}}_{\delta},\hat{\bm{q}})\ket{\psi}).
  \end{align}
  In computing this maximum it suffices to consider \(\ket{\phi}\)
  such that \(\bra{\phi}f(\hat{\bm{q}}_{\delta},\hat{\bm{q}})\ket{\psi}\)
  is real. Consider an arbitrary such \(\ket{\phi}\).
  We have
  \begin{align}
    \bra{\phi}f(\hat{\bm{q}}_{\delta}, \hat{\bm{q}})\ket{\psi}
    &= \int_{\bm{x}\in\rls^{m},\bm{y}\in\rls^{m}}f(\bm{x},\bm{y})\bra{\phi}d\Pi_{\delta}(\bm{x})d\Pi(\bm{y})\ket{\psi},
      \notag\\
    &= \int_{\bm{x}\in\rls^{m},\bm{y}\in\rls^{m}}\qty(\int_{\bm{\kappa}\in\rls^{m},\bm{\kappa}'\in\rls^{m}}d\tilde\mu(\bm{\kappa},\bm{\kappa}')
      e^{i(\bm{\kappa}\cdot \bm{x}+\bm{\kappa}'\cdot\bm{y})})
      \bra{\phi}d\Pi_{\delta}(\bm{x})d\Pi(\bm{y})\ket{\psi},
  \end{align}
  The integrand in this integral is bounded and in view of the
  comments at the beginning of the proof of
  Prop.~\ref{prop:fbndedmeas} and taking advantage of the discrete
  spectrum of \(\hat{\bm{q}}_{\delta}\) for \(\delta>0\), the
  expression
  \(d\tilde \mu(\bm{\kappa},\bm{\kappa}')
  \bra{\phi}d\Pi_{\delta}(\bm{x})d\Pi(\bm{y})\ket{\psi}\) defines a
  bounded complex measure on \(\rls^{m}\times \rls^{m}\).  We can
  therefore apply the Fubini-Tonelli theorem to change order of
  integration.  We obtain
  \begin{align}
    \bra{\phi}f(\hat{\bm{q}}_{\delta}, \hat{\bm{q}})\ket{\psi}
    &=   \int_{\bm{\kappa}\in\rls^{m},\bm{\kappa}'\in\rls^{m}}d\tilde\mu(\bm{\kappa},\bm{\kappa}')
      \qty(\int_{\bm{x}\in\rls^{m},\bm{y}\in\rls^{m}}
      e^{i(\bm{\kappa}\cdot \bm{x}+\bm{\kappa}'\cdot\bm{y})}
      \bra{\phi}d\Pi_{\delta}(\bm{x})d\Pi(\bm{y})\ket{\psi}).
      \label{eq:fxyqdq1}
  \end{align}
  The inner integral is 
  \begin{align}
    \int_{\bm{x}\in\rls^{m},\bm{y}\in\rls^{m}}
    e^{i(\bm{\kappa}\cdot \bm{x}+\bm{\kappa}'\cdot \bm{y})}
    \bra{\phi}d\Pi_{\delta}(\bm{x})d\Pi(\bm{y})\ket{\psi}
    \hspace*{-1in}&\notag\\
    &=
      \bra{\phi} e^{i\bm{\kappa}\cdot \hat{\bm{q}}_{\delta}}
      e^{i\bm{\kappa}'\cdot \hat{\bm{q}}}\ket{\psi}
      \notag\\
    &=
      \bra{\phi}e^{i(\bm{\kappa}+\bm{\kappa}')\cdot \hat{\bm{q}}}\ket{\psi}
      +
      \bra{\phi}\qty(e^{i\bm{\kappa}\cdot \hat{\bm{q}}_{\delta}}
      -e^{i\bm{\kappa}\cdot \hat{\bm{q}}})
      e^{i\bm{\kappa}'\cdot \hat{\bm{q}}}\ket{\psi}.
      \label{eq:prechainedtr}
  \end{align}
  To continue, we chain the triangle inequality as follows:
  Define \(\hat{\bm{q}}_{\delta|j}\) so that
  the \(j'\)'th entry of \(\hat{\bm{q}}_{\delta|j}\) is \(\hat{q}_{j',0}\)
  for \(j'\leq j\) and \(\hat{q}_{j',\delta}\) otherwise. Then
  \(\hat{\bm{q}}_{\delta|0}=\hat{\bm{q}}_{\delta}\)
  and \(\hat{\bm{q}}_{\delta|m}=\hat{\bm{q}}\). Consequently
  \begin{align}
    e^{i\bm{\kappa}\cdot\hat{q}_{\delta}}-e^{i\bm{\kappa}\cdot\hat{q}}
    &=
      \sum_{j=0}^{m-1}    e^{i\bm{\kappa}\cdot\hat{\bm{q}}_{\delta|j}}
      -e^{i\bm{\kappa}\cdot\hat{\bm{q}}_{\delta|j+1}}
      \notag\\
    &=
      \sum_{j=0}^{m-1}
      \qty(e^{i\kappa_{j+1}\hat{q}_{j+1,\delta}}-e^{i\kappa_{j+1}\hat{q}_{j+1,0}})
     e^{i\qty(\bm{\kappa}\cdot\hat{\bm{q}}_{\delta|j+1}-\kappa_{j+1}\hat{q}_{j+1,0})}.
  \end{align}
  After substituting into the last summand of Eq.~\eqref{eq:prechainedtr} we get
  a sum of terms of the form
  \begin{align}
    \bra{\phi}
    \qty(e^{i\kappa_{j+1}\hat{q}_{\delta|j}}-e^{i\kappa_{j+1}\hat{q}_{j}})\hat{V}_{j+1}
    \ket{\psi},
  \end{align}
  where \(\hat{V}_{j}\) is unitary and does not affect the vacuum
  state of the LO modes of the \(j\)'th pair of signal and LO modes.
  Each term can therefore be bounded by applying Eq.~\eqref{thm:eq:mains}
  for
  \begin{align}
    \qty|\bra{\phi}
    \qty(e^{i\kappa_{j+1}\hat{q}_{\delta,j}}-e^{i\kappa_{j+1}\hat{q}_{j}})\hat{V}_{j}
    \ket{\psi}|
    &\leq
      2\delta|\kappa_{j+1}|\sqrt{ \bra{\psi}
      \qty(\hat{\Omega}_{j+1}
      + \kappa_{j+1}^{2}\qty(\hat{\Omega}_{j+1}+\overline{\omega_{j+1}^{2}}))
      \ket{\psi}},      
  \end{align}
  with \(j\in\{0,\ldots, m-1\}\).
  Continuing from Eq.~\eqref{eq:prechainedtr} we obtain
  \begin{align}
    \int_{\bm{x}\in\rls^{m},\bm{y}\in\rls^{m}}
    e^{i(\bm{\kappa}\cdot \bm{x}+\bm{\kappa}'\cdot \bm{y})}
    \bra{\phi}d\Pi_{\delta}(\bm{x})d\Pi(\bm{y})\ket{\psi}
    \hspace*{-1in}&\notag\\
    &\in
      \bra{\phi}e^{i(\bm{\kappa}+\bm{\kappa}')\cdot \hat{\bm{q}}}\ket{\psi}
      + \bm{B}_{1} 2\delta\sum_{j=1}^{m}
      |\kappa_{j}|\sqrt{ \bra{\psi}
    \qty(\hat{\Omega}_{j} + \kappa_{j}^{2}\qty(\hat{\Omega}_{j}+\overline{\omega_{j}^{2}}))
      \ket{\psi}},
      \label{eq:postchainedtr}
  \end{align}
  with \(\bm{B}_{1}\) the unit disk in the complex plain.
  Substituting the right-hand side into Eq.~\eqref{eq:fxyqdq1} and
  applying the bounds gives
  \begin{align}
    \bra{\phi}f(\hat{\bm{q}}_{\delta}, \hat{\bm{q}})\ket{\psi}
    &\leq
      \qty|\int_{\bm{\kappa}\in\rls^{m},\bm{\kappa}'\in\rls^{m}}d\tilde\mu(\bm{\kappa},\bm{\kappa}')
      \bra{\phi}e^{i(\bm{\kappa}+\bm{\kappa}')\hat{\bm{q}}}\ket{\psi}|
      \notag\\
    &\hphantom{{}\leq \int}
      {}+2\delta\sum_{j=1}^{m}\int_{\bm{\kappa}\in\rls^{m},\bm{\kappa}'\in\rls^{m}}d\mu(\bm{\kappa},\bm{\kappa}')
      |\kappa_{j}|\sqrt{ \bra{\psi}
      \qty(\hat{\Omega}_{j} +
      \kappa_{j}^{2}\qty(\hat{\Omega}_{j}+\overline{\omega_{j}^{2}}))
      \ket{\psi}}.
  \end{align}
  The integral under the absolute value evaluates to
  \begin{align}
     \int_{\bm{\kappa}\in\rls^{m},\bm{\kappa}'\in\rls^{m}}d\tilde\mu(\bm{\kappa},\bm{\kappa}')
    \bra{\phi}e^{i(\bm{\kappa}+\bm{\kappa}')\hat{\bm{q}}}\ket{\psi}
    &=
      \int_{\bm{x}\in\rls^{m}}\int_{\bm{\kappa}\in\rls^{m},\bm{\kappa}'\in\rls^{m}}
      d\tilde\mu(\bm{\kappa},\bm{\kappa}')e^{i(\bm{\kappa}+\bm{\kappa}')\cdot\bm{x}}
      \bra{\phi}d\Pi(\bm{x})\ket{\psi}
      \notag\\
    &=
      \int_{\bm{x}\in\rls^{m}} f(\bm{x},\bm{x})      \bra{\phi}d\Pi(\bm{x})\ket{\psi}
      \notag\\
    &= 0.
  \end{align}
  For the second integral we get
  \begin{align}
    2\delta\sum_{j=1}^{m}
    \int_{\bm{\kappa}\in\rls^{m},\bm{\kappa}'\in\rls^{m}}d\mu(\bm{\kappa},\bm{\kappa}')
    |\kappa_{j}|\sqrt{ \bra{\psi}
    \qty(\hat{\Omega}_{j}
    + \kappa_{j}^{2}\qty(\hat{\Omega}_{j}+\overline{\omega_{j}^{2}}))
    \ket{\psi}}
    \hspace*{-2.5in}
    &\notag\\
    &\leq
      {2}\delta \sum_{j=1}^{m}\int_{\kappa_{j}\in\rls} d\mu(\kappa_{j})\,
      |\kappa_{j}|\qty(\sqrt{
      \bra{\psi}\hat{\Omega}_{j}\ket{\psi}}
      + |\kappa_{j}|
      \sqrt{\bra{\psi}\hat{\Omega}_{j}\ket{\psi}+\overline{\omega_{j}^{2}}})
      \notag\\
    &=
      {2}\delta\sum_{j=1}^{m}
      \qty(\tilde f_{j,1}\sqrt{\bra{\psi}\hat{\Omega}_{j}\ket{\psi}}+ \tilde f_{j,2}
      \sqrt{\bra{\psi}\hat{\Omega}_{j}+\overline{\omega_{j}^{2}}\ket{\psi}})
      \notag\\
    &\leq
      {2}\delta\qty(
      \tilde f_{\tot,1}\sqrt{\bra{\psi}\hat{\Omega}_{\tot}\ket{\psi}}
      + \tilde f_{\tot,2}\sqrt{\bra{\psi}\hat{\Omega}_{\tot}
      +\overline{\omega^{2}_{\tot}}\ket{\psi}}
      ),
  \end{align}
  where we used the inequality of Eq.~\eqref{eq:sqrtx+y-ineq} for the second line and
  applied the
  the Cauchy-Schwarz inequality to each summand for the last line.
  The bound in the proposition follows by substitution
  and the inequality of Eq.~\eqref{eq:x+ysqrt-ineq}.
\end{proof}

\section{Simulation of GKP fidelity}
\label{app:gkp_numerics}

\begin{figure}[h!]
  \includegraphics[scale = 0.75]{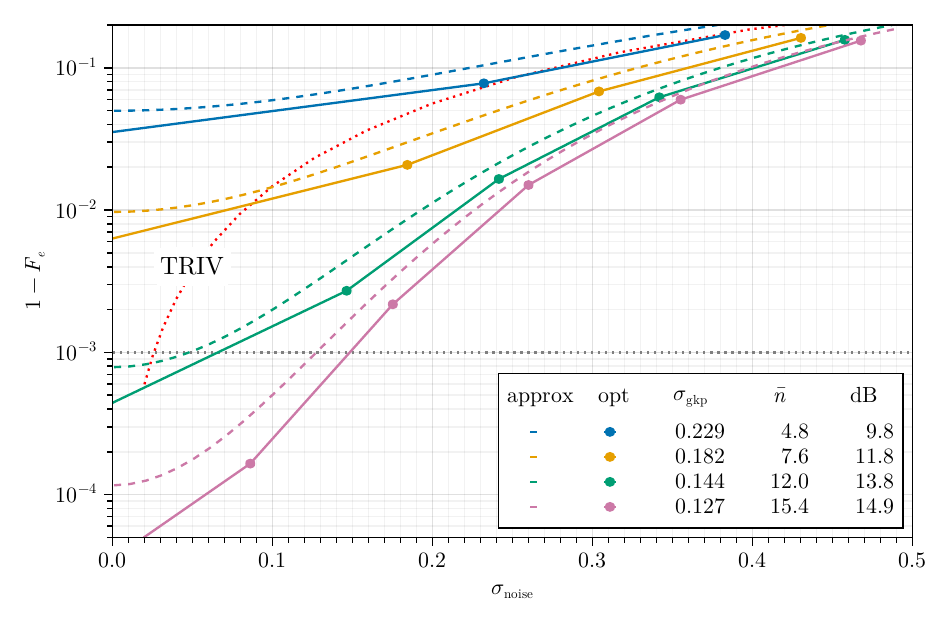} 
  \caption{ Entanglement infidelity as a function of displacement
    noise strength for square GKP qubits. The displacement noise is
    Gaussian and isotropic with variance
    \(\sigma_{\text{noise}}^2\). Each series corresponds to a GKP code
    with with variance \(\sigma_{\text{gkp}}^2\), average photon
    number $\bar{n}$ of the code space projector, or the given
    squeezing in decibels (dB).  According to
    Ref.~\cite{aghaee2025scaling}, for general GKP fault tolerance, a
    squeezing safely above \(9.75\mySI{dB}\) is sufficient. See
    Fig. S15 of the supplementary information of the reference, which
    shows that the logical error rate at the threshold is greater than
    \(10\mySI{\%}\) and error rates for reasonable coding distances do
    not drop below \(.1\mySI{\%}\) for squeezings less than
    \(10\mySI{dB}\).  Ref.\cite{vahlbruch2024detection} claims
    \(15\mySI{dB}\) optical squeezing.  We also include the TRIV curve
    for the trivial qubit code spanned by the oscillator ground
    $\ket{0}$ and first excited $\ket{1}$ states. The solid lines
    correspond to the numerically optimized recovery while the dashed
    lines correspond to an analytic approximation for teleported error
    correction circuit in Fig.~\ref{fig:gkptelerr}. For the analytic
    approximation, \(\sigma_{\text{gkp}}^{2}\) is treated as an
    effective noise on ideal GKP qubits. Because each of the two
    ancillas has the same effective noise, by the error propagation
    analysis of Sect.~\ref{sect:cvec}, this corresponds to an
    effective noise variance of \(3\sigma_{\text{gkp}}^{2}\) on the
    incoming GKP qubit with ideal ancillas. For the curves shown,
    noise with variance \(\sigma_{\text{noise}}^{2}\) is added to the
    incoming GKP qubit only.  For the numerically optimized recovery,
    there are no ancillas. The incoming GKP qubit is the same as for
    the analytic approximation with squeezing corresponding to
    \(\sigma_{\text{gkp}}\).  For comparing to the analytic
    approximation, we optimized the recovery for an added displacement
    noise with variance
    \(\sigma_{\text{opt}}^{2}=2\sigma_{\text{gkp}}^{2}+\sigma_{\text{noise}}^{2}\).
    The straight solid lines going off to the left connect to the infidelity of
    the optimal recovery for a value of \(\sigma_{\text{opt}}^{2}\)
    that is less than \(2\sigma_{\text{gkp}}^{2}\). Their horizontal
    coordinates correspond to a negative values \(-x\) of
    \(\sigma_{\text{noise}}^{2}\), which for illustration purposes were
    identified with the unplotted locations of \(-\sqrt{x}\) on the
    horizontal axis }
  \label{fig:gkp_numerics}
\end{figure}

Fig.~\ref{fig:gkp_numerics} shows the performance of square GKP qubits with
different average photon numbers. The infidelities for recovery maps after
displacement noise applied to the GKP qubit is shown. Two recovery maps are
considered. The first is a numerically optimized recovery map based on
Ref.~\cite{totey2023performance}. The second is recovery based on teleported
error correction. For the second we use an analytic approximation from
Ref.~\cite{hillmann2022performance}. The numerical optimization was implemented
using the \textsf{Julia} programming language~\cite{bezanson2017julia}. The code
used for these simulations was also used in Ref.~\cite{totey2023performance} to
obtain fidelities of rotation-symmetric codes against the simultaneous
loss-dephasing channel.

The noise channel $\mathcal N(\hat\rho)$ we consider here is Gaussian
distributed random displacement noise which can be obtained by solving the
master equation
\begin{equation}\label{eq:master}
\dot {\hat{\rho}} = \kappa\left(\kappa\mathcal D[\hat a] + \mathcal D[\hat a^\dagger]\right)(\hat\rho)
\end{equation}
where
$\mathcal D[\hat L] (\hat\rho) = \hat L \hat\rho \hat L^\dagger - \frac{1}{2}
\hat L^\dagger \hat L \hat\rho - \frac{1}{2} \hat\rho \hat L^\dagger \hat L$.
The solution to this master equation can be obtained by exponentiating the
superoperator representation $\mathcal{L}[\hat{L}]$ of the Lindbladian
$\mathcal{D}[\hat{L}]$. The superoperator representation is given by
$\mathcal{L}[\hat{L}] = \hat L^* \otimes \hat L - \frac{1}{2} \hat I\otimes \hat
L^\dagger \hat L - \frac{1}{2} (\hat L^\dagger \hat L)^T \otimes \hat I$.
The noise channel evolved for time $t$ is then given by
\begin{equation}\label{eq:disp_ch}
\mathcal{N}_{\kappa t}(\hat{\rho}) = \left[e^{\kappa t(\mathcal L[\hat{a}] + \mathcal L[\hat{a}^\dagger])}\right] (\hat{\rho})\;.
\end{equation}
where $[\cdot](\hat{\rho})$ represents application of the given functional to
$\hat{\rho}$. Thus $\kappa t$ represents the unitless strength of the random
displacement rate. This noise strength can be related to the variance of the
random displacement distribution given by the Kraus-operator representation of
this channel
\begin{equation}\label{eq:disp_ch_krauss}
\mathcal N_\sigma(\hat\rho) = \int \frac{\dd[2]{\alpha}}{\pi\sigma^2} \,
\exp(-\frac{\abs{\alpha}^2}{\sigma^2}) \hat{D}_{\alpha}\, \hat{\rho}\, \hat{D}^\dagger_{\alpha}\,.
\end{equation}
By writing Eq.~\eqref{eq:disp_ch} and Eq.~\eqref{eq:disp_ch_krauss} in terms of
the characteristic function, we find that $\kappa t = \sigma^2$ which sets the
\(x\)-coordinates in Fig.~\eqref{fig:gkp_numerics}. Note that for canonical
quadratures, this channel corresponds to noise with variance $\sigma^2$ for both
$\hat{x}$ and $\hat{p}$.

For the numerically optimized recovery operation, the total quantum channel we consider is
\begin{equation}\label{eq:ec:totalchannel}
  \mathcal{L} = \mathcal R^{\text{SDP}} \circ \mathcal N_{\sigma} \circ \mathcal E,
\end{equation}
where \(\mathcal{E}\) is the encoding isometry from a qubit into the GKP qubit
space, $\mathcal N_{\sigma}$ is the noise map (Eq.~\eqref{eq:disp_ch_krauss}),
$\mathcal R^{\text{SDP}}$ is the numerically optimized recovery isometry from
the GKP qubit space back to a qubit. The encoder is given by %
$\mathcal{E}[\hat\Pi](\hat\rho) = \hat\Pi \hat\rho {\hat\Pi}^\dagger$ where %
$\hat\Pi = \ket{0_L}\bra{0} + \ket{1_L}\bra{1}$ is the codespace projector while
$\{\ket{i_L}\}_{i}$ is a basis for the logical subspace. We quantify the
performance of the recovery map $\mathcal{R}^{\text{SDP}}$ using the entanglement fidelity
\begin{equation}\label{eq:chan_fid}
  F_e(\mathcal L) = \frac{1}{d^{2}}\sum_k \abs{\tr(\hat M_k)}^2\, ,
\end{equation}
where $\{ \hat M_k\}_{k}$ are the Kraus operators for $\mathcal{L}$ and $d$ is the
dimension of the channel's input/output Hilbert
space. Thus $d=2$ for $\mathcal{L}$ given by
Eq.~\eqref{eq:ec:totalchannel} and so the \(y\)-coordinate in
Fig.~\ref{fig:gkp_numerics} gives the entanglement infidelity
$1 - F_e(\mathcal{L})$ for the logical qubit channel.

We numerically implement the procedure in Ref.~\cite{flectcher2007optimum} which
uses a semidefinite program (SDP) to compute the recovery
$\mathcal R^\text{SDP}$ that optimizes the channel fidelity in
Eq.~\eqref{eq:chan_fid}. The finite energy GKP codewords are numerically
represented in the Fock basis with a finite occupation cutoff using structures
provided by the \textsf{QuantumOptics.jl}
package~\cite{kramer2018quantumoptics}. They are created by superposing coherent
states on a square grid in phase space with amplitudes exponentially damped by their
distance from the origin. This procedure results in only approximately
orthogonal code states $\bra{1_L}\ket{0_L} \neq 0$ and at large damping factors,
this nonzero codeword overlap limits the achievable channel fidelity. To avoid
this issue, we enforce orthogonality by removing the overlapping component from
$\ket{1}_L$ and renormalizing. The resulting code states then always have zero
mean \(\bra{i_L}\hat{q}_{\theta}\ket{i_L} = 0\) for all quadratures
\(\hat{q}_{\theta}\).

We take advantage of the algorithm for sparse matrix exponentiation in
Refs.~\cite{hogben2011software,kuprov2011diagonalization} to
efficiently numerically compute Eq.~\eqref{eq:disp_ch} and obtain the
channel with noise strength $\kappa t$. In the Fock basis
representation, the tensor products used to form the superoperator
representation $\mathcal{L}[\hat{L}]$ can be realized using the usual
Kronecker matrix product and applied through right multiplication with
the column-stacked vectorized $|\rho\rangle\rangle$. The SDP is then a
function of $\mathcal{L}[\hat{L}]$ and $\hat{\Pi}$ with both
numerically represented in the Fock basis. To solve this SDP, we use
the \textsf{COSMO.jl}~\cite{garstka2021cosmo} convex optimizer via the
\textsf{JuMP.jl}~\cite{lubin2023jump} interface, both of which are
\textsf{Julia}-based packages. We set the relevant numerical
tolerances for convergence of the solver to $10^{-8}$, well beyond the
logical infidelities shown in Fig.~\ref{fig:gkp_numerics}. We
furthermore ensure that the Fock cutoff is sufficiently large so that
its effect on the calculations remains negligible.  This cutoff is set
so that there is at most 1e-6 of infidelity between the codespace
projectors calculated at the actually used cutoff versus the maximum
cutoff (set to 400). This allows for dynamically scaling the cutoff
based on the amount of GKP squeezing, which greatly reduces the
simulation time and memory usage.

We also compute an analytic approximation of the channel fidelity for the linear
optical recovery circuit shown in Fig.~\ref{fig:gkptelerr}. This analytic
approximation is based off of Section V.B.1. and Appendix D of
Ref.~\cite{hillmann2022performance} along with Appendix A of
Ref.~\cite{noh2020fault}. The approximation relies on displacement twirling,
which allows approximation of a finitely squeezed GKP codeword as an incoherent
mixture of infinite energy GKP codewords. The amount of squeezing is described
using an envelop damping operator $\exp[-\Delta^2\hat{n}]$ applied to infinite
energy GKP. Meanwhile the incoherent mixture of infinite energy GKP codewords is
approximated as arising due to application of Eq.~\eqref{eq:disp_ch_krauss} to
an infinite energy GKP codeword %
$\mathcal{N}_{\sigma_{\text{gkp}}}(\ket{s}\bra{s}_{\text{gkp}})$. For these
approximate GKP codewords, the probability of successful recovery using the
closest lattice-point decoder for extracting the syndrome of one quadrature is
given by
\begin{align}\label{eq:psucc}
 p_{\text{succ}}(\sigma_{\text{eff}}) &=
\frac{1}{\sqrt{2\pi\sigma_{\text{eff}}^2}}\sum_{n\in\mathbb{Z}}
\int_{(2n-\frac{1}{2})\sqrt{\pi}}^{(2n+\frac{1}{2})\sqrt{\pi}}\dd{z}\, e^{-\frac{z^2}{2\sigma_{\text{eff}}^2}} \;.
\end{align}
For an isotropically squeezed GKP codespace, the entanglement fidelity
of the logical channel is then given by %
$F_e = [p_{\text{succ}}(\sigma_{\text{eff}})]^2$ where we set the
effective noise variance according to %
$\sigma_{\text{eff}}^2 = 3\sigma_{\text{gkp}}^2 +
\sigma_{\text{noise}}^2$ for the analytic curves in
Fig.~\ref{fig:gkp_numerics}.  For the Fig.~\ref{fig:gkp_numerics}
legend, we give several alternative metrics besides $\sigma^2$ for
quantifying the energy of a GKP code. The average photon number is
given by $\bar{n} = \Tr(\hat{n}\Pi_L)$ while the amount of GKP
squeezing in decibels is given by
$-10\log_{10}(2\sigma^2)\,\text{dB}$. The relationships between these
parameters and GKP envelop damping parameter $\Delta$ is given by
$\bar{n} = 1/(2\Delta^2)$ and %
$\sigma_{\text{gkp}}^2 = (1-e^{-\Delta^2})/(1+e^{-\Delta^2})$.


\end{document}